\documentclass[aps,twocolumn,prx,preprintnumbers, nofootinbib,10pt]{revtex4-2}
\usepackage{xpatch}
\makeatletter
\patchcmd{\@ssect@ltx}
    {\addcontentsline{toc}{#1}{\protect\numberline{}#8}}
    {}
    {}
    {}
\makeatother
\pdfoutput=1
\usepackage{color}
\usepackage{enumitem}
\usepackage{mathtools}
\usepackage[dvipsnames]{xcolor}
\usepackage{adjustbox}

\definecolor{amaranth}{rgb}{0.9, 0.17, 0.31}
\definecolor{forestForestGreen(web)}{rgb}{0.13, 0.55, 0.13}
\definecolor{blue(munsell)}{HTML}{005567}
\definecolor{bblue}{rgb}{0.0, 0.58, 0.71}

\usepackage{pgfplots}
\pgfplotsset{compat=1.18}
\usepackage{amssymb}
\usepackage{amsfonts}
\usepackage{graphicx}
\usepackage{epstopdf}
\usepackage{dcolumn}
\usepackage{amsmath}
\usepackage{latexsym,bm}
\usepackage{amsthm}
\usepackage{slashed}
\usepackage{float}
\usepackage{color}
\usepackage{url}
\usepackage{rotating}
\usepackage[normalem]{ulem}
\usepackage{faktor}
\usepackage{tikz-cd}
\usepackage{tensor}
\usepackage{array}

\usepackage{tabularx}
\usepackage{makecell}
\usepackage[normalem]{ulem}
\usepackage{changepage}
\numberwithin{equation}{section}

\newcommand{\hide}[1]{}

\usepackage{tikz}
\usepackage{diagbox}
\usetikzlibrary{positioning}
\usetikzlibrary{calc}
\usetikzlibrary{decorations.pathreplacing,calligraphy}
\usetikzlibrary{arrows.meta,decorations.markings}
\usepackage{xstring}
\usetikzlibrary{decorations.pathmorphing} 
\usetikzlibrary{decorations.markings} 
\usetikzlibrary{arrows} 
\usetikzlibrary{shapes} 
\usetikzlibrary{matrix} 
\usetikzlibrary{positioning} 
\usepackage[english]{babel} 
\usepackage[autostyle]{csquotes}
\usepackage{pifont}

\usetikzlibrary{shapes.multipart}

\tikzset{->-/.style={decoration={
  markings,
  mark=at position .5 with {\arrow{>}}},postaction={decorate}}}

    \definecolor{myblue}{RGB}{29, 66, 166}
    \definecolor{myorange}{RGB}{230,120,20}

    \tikzset{
        fusion/.style={draw=myblue,line width=1.35pt},
        onstrand/.style={
            postaction={
                decorate,
                decoration={
                    markings,
                    mark=at position #1 with {\arrow{Latex[length=5pt,width=5pt]}}
                }
            }
        },
        orangepath/.style={
            draw=myorange,
            line width=0.9pt,
            dash pattern=on 1.5pt off 1.8pt,
            postaction={
                decorate,
                decoration={
                    markings,
                    mark=at position #1 with {\arrow{Latex[length=4pt,width=4pt]}}
                }
            }
        }
    }

\newcommand{\bea}{\begin{eqnarray}}
\newcommand{\eea}{\end{eqnarray}}
\newcommand{\be}{\begin{equation}}
\newcommand{\ee}{\end{equation}}
\newcommand{\ba}{\begin{aligned}}
\newcommand{\ea}{\end{aligned}}
\newcommand{\bit}{\begin{itemize}}
\newcommand{\eit}{\end{itemize}}
\newcommand{\ben}{\begin{enumerate}}
\newcommand{\een}{\end{enumerate}}
\newcommand{\nn}{\nonumber}

\newcommand{\id}{\text{id}}

\newcommand{\Bsym}{\mathfrak{B}^{\text{sym}}}

\newcommand{\wh}{\widehat}
\newcommand{\ot}{\otimes}

\newcommand{\Z}{{\mathbb Z}}
\newcommand{\R}{{\mathbb R}}
\newcommand{\C}{{\mathbb C}}

\newcommand{\cB}{\mathcal{B}}
\newcommand{\cc}{\mathcal{C}}
\newcommand{\cC}{\mathcal{C}}
\newcommand{\cD}{\mathcal{D}}
\newcommand{\cE}{\mathcal{E}}
\newcommand{\cF}{\mathcal{F}}

\newcommand{\cL}{\mathcal{L}}
\newcommand{\cM}{\mathcal{M}}
\newcommand{\cN}{\mathcal{N}}
\newcommand{\cO}{\mathcal{O}}

\newcommand{\cS}{\mathcal{C}}
\newcommand{\cT}{\mathcal{T}}

\newcommand{\cX}{\mathcal{X}}
\newcommand{\cY}{\mathcal{Y}}
\newcommand{\cZ}{\mathcal{Z}}

\newcommand{\A}{\mathsf{A}}

\renewcommand{\R}{\mathbb{R}}
\newcommand{\bH}{\mathbb{H}}
\newcommand{\D}{\mathsf{D}}

\usepackage[colorlinks=true,
linkcolor=blue(munsell),
citecolor=blue(munsell),
filecolor=magenta,
urlcolor=blue(munsell)
]{hyperref}
\usepackage[nameinlink,capitalize,noabbrev]{cleveref}

\newcommand{\fB}{\mathfrak{B}}
\newcommand{\fZ}{\mathfrak{Z}}

\newcommand{\Tr}{\text{Tr}}

\newcommand{\Hom}{\text{Hom}}
\newcommand{\End}{\text{End}}
\renewcommand{\Vec}{\mathsf{Vec}}
\newcommand{\Rep}{\mathsf{Rep}}
\newcommand{\Mod}{\mathsf{Mod}}

\newcommand{\Bimod}{\mathsf{Bimod}}
\renewcommand{\dim}{\text{dim}}

\newcommand{\lid}{\mathbf{1}}

\newcommand{\Ising}{\mathsf{Ising}}
\newcommand{\TY}{\mathsf{TY}}

\newcommand{\sym}{\text{sym}}

\usepackage{euscript}

\makeatother

\allowdisplaybreaks

\usepackage{physics}
\renewcommand{\ol}{\overline}

\usepackage{float}

\usepackage{makecell}

\makeatletter
\newcommand\xlabel[2][]{\phantomsection\def\@currentlabelname{#1}\label{#2}}
\makeatother

\newtheorem{theorem}{Theorem}[section]
\newtheorem*{theorem*}{Theorem}

\theoremstyle{definition}
\newtheorem{definition}[theorem]{Definition}
\newtheorem{remark}{Remark}[section]
\newtheorem{example}[theorem]{Example}

\makeatletter
\def\l@subsubsection#1#2{}
\makeatother

\newcommand{\cTY}{\mathsf{TY}_{\overline{\C}}}

\begin{document}

\title{Categorical Time-Reversal Symmetries}

\author{Rui Wen}
\author{Sakura Sch\"afer-Nameki}

 \affiliation{Mathematical Institute, University of Oxford, Woodstock Road, Oxford, OX2 6GG, United Kingdom}

\begin{abstract} 
\noindent 
The classification of phases using categorical symmetries has greatly expanded the landscape of gapped and gapless phases. So far, however, these developments have largely been restricted to phases with unitary (higher-)categorical symmetries over $\mathbb{C}$. In this work, we incorporate anti-unitary symmetries, such as time-reversal symmetry $\mathbb{Z}_2^T$, and show that the relevant physical structures are naturally described by fusion categories over $\mathbb{R}$. A class of real fusion categories, which we call \emph{Galois-real fusion categories}, provides the correct categorical model for anti-unitary symmetries. A simple example is the time-reversal symmetry $\mathbb{Z}_2^T$ itself. We discuss the basic structures of real fusion categories and present a range of examples, including the group-theoretical categories $(G^T)^{\omega}$ and $\mathsf{Rep}(G^T)$ associated to anti-linear groups $G^T$, as well as non-invertible time-reversal symmetries described by a real analogue of Tambara--Yamagami fusion categories. We then classify gapped phases enriched with anti-linear symmetries in terms of module categories over Galois-real fusion categories. 
We furthermore apply the categorical formulation to prove dualities (i.e. gauge or Morita equivalences) of anti-linear symmetries generated by gauging subgroups. Complementing this, we also develop a Symmetry Topological Field Theory (SymTFT) framework, in which Galois-real fusion categories arise as boundary conditions of a $\mathbb{Z}_2^T$-enriched SymTFT. Morita equivalent anti-linear symmetries are shown to arise as different boundaries of the same $\mathbb{Z}_2^T$-enriched SymTFT.
\end{abstract}

\maketitle

\tableofcontents


\section{Introduction}

The modern understanding of symmetries in quantum many-body systems and quantum field theory has expanded well beyond the traditional setting of groups.  This was initiated by the insight of~\cite{Gaiotto:2014kfa} that topological defects {\bf are} symmetries. 
One of the key physically relevant implications of these extensions is that symmetries need not have inverses and form groups, but can form more general, non-invertible structures. 
These are subsumed under the name of categorical symmetries and were recently reviewed in \cite{Schafer-Nameki:2023jdn, Shao:2023gho, Luo:2023ive, Bhardwaj:2023kri}. 

The main focus in the literature has been, more precisely, unitary or at least complex ($\C$)-linear fusion categories. 
As is well-known, going back to the work of Wigner,  symmetry groups in quantum mechanics can however also be anti-unitary. The main example of an anti-unitary symmetry is time-reversal $\Z_2^T$. Time-reversal symmetry plays a crucial role in many phases of matter, starting with the Haldane chain in (1+1)d, which is a $\Z_2^T$-SPT phase \cite{Haldane:1983ru}. Other anti-unitary symmetries such as $\Z_4^T, U(1)\rtimes \Z_2^T, U(1)\times \Z_2^T$ are also common in topological insulators/superconductors~\cite{Chen:2011pg,Vishwanath:2012tq,wang2013boson}.

The main purpose of this paper is to  explore the categorical generalization, including identifying a mathematical home for such anti-unitary categorical symmetries and studying their implications in the context of the categorical classification of phases and phase transitions \cite{Thorngren:2019fuc, Bhardwaj:2023fca, Bhardwaj:2023idu, Bhardwaj:2024qrf, Bhardwaj:2023bbf}. We show that the natural setting for categorical anti-unitary symmetries is {\bf real fusion categories}. The mathematical literature on this is surprisingly scarce \cite{sanford2025fusion, plavnik2024TYreal}, and we will be providing some of the background on these categories in the following.

\smallskip
\noindent
{\bf Two types of Real Fusion Categories.} 
Unlike (unitary) fusion categories that are defined over $\C$, real fusion categories are defined over the field of real numbers $\R$.
There are two types of real fusion categories, which are distinguished by the type of endomorphisms of the identity object: a fusion category over $\mathbb{R}$ is called {\bf $\R$-real} if $\mathrm{End}(1)\cong \mathbb{R}$, and {\bf Galois-real} if $\mathrm{End}(1)\cong \mathbb{C}$, in which case it has a $\Z_2^T$-grading 
\be
\cC = \cC_1 \oplus \cC_T 
\ee
with $\cC_1$ the linear sector and $\cC_T$ the anti-linear. 
Crucially, we will argue that for a quantum system, to have a complex Hilbert space as the state space, it is the Galois-real fusion categories that have an interpretation as symmetry categories -- $\R$-real categories essentially have $\R$ acting on the local state space, which we disregard as a consistent setup for quantum theories. The two $\Z_2^T$-graded components of a Galois-real fusion category naturally correspond to the linear and anti-linear sectors of a generalized anti-linear symmetry. 
In particular, $\cC_1$ alone is an ordinary fusion category over $\C$.

However, this does not mean that we should disregard the $\R$-real categories from considerations. These can still arise as the categories of charges, and perhaps surprisingly, can be gauge (Morita) equivalent to Galois-real fusion categories. 
The simplest example of this is the category $\Rep (\Z_2^T) = \Vec_\R$, which is the category of real vector spaces, but is also Morita dual to $\Vec_{\Z_2^T}$, a Galois-real fusion category.  An $\R$-real fusion category has an interpretation as the category of defects in a given gapped phase, but not as an abstract symmetry category. 
E.g. $\Rep(G^T)$ is $\R$-real, and thus inadmissible as a symmetry category. Nevertheless $\Rep(G^T)$ appears naturally as the category of charges in a $G^T$-SPT phase.

\smallskip
\noindent
{\bf The Strange World of Real Fusion Categories.}
As the theory of real fusion categories is relatively unknown, we will spend one section in this paper, \Cref{sec_basic_rfc}, explaining the basics and giving concrete examples for both types of real fusion categories. 
One class of examples are group theoretical, based on the $G^T$ anti-linear group symmetry. Here there is a group homomorphism
\be\label{sGrade}
s: \quad G^T \to \Z_2^T = \{1, T\} \,,
\ee
such that the pre-image of $1$ is the linear part and the pre-image of $T$ the anti-linear (time-reversal) part of the symmetry. We construct the real fusion category $\Vec_{G^T}^\omega$, which is Galois-real, and describe a possibly anomalous anti-unitary $G^T$-symmetry.

Morita (i.e. gauge) equivalence for real fusion categories offers some surprises. Naively ``gauging $G^T$ in $\Vec_{G^T}$" results in $\Rep(G^T)$, however this is an $\R$-real fusion category and thus is inadmissible as a symmetry category -- again this is not surprising given that this would involve gauging time-reversal symmetry, which we believe is not possible within the framework of standard quantum theory\footnote{Whether or not time-reversal symmetry should or should not be gauged in systems of quantum gravity is an even wider point of discussion \cite{Harlow:2023hjb, Susskind:2026moe}. Here we simply state that time-reversal should not be gauged in quantum systems, irrespective of coupling to gravity.}. 

Perhaps even more surprisingly, in the world of real-fusion categories the following two symmetries are gauge related: let $A$ be an abelian group, then  
\be\label{RealSurprise}
\Vec_{A\times \Z_2^T} \text{ Morita equivalent to } \Vec_{A\rtimes\Z_2^T}\,,
\ee
where in the latter category $\Z_2^T$ acts on $A$ as inversion and  generically the latter is a non-abelian group. 
Note that there are various versions in general of a group $G^T$ depending on (\ref{sGrade}), e.g. here we specifically mean the $\Z_2^T$ acts as the reflections in the dihedral group. 

Another surprising property is that the $\R$-real $\Rep(A^T)$ categories for abelian $A$ can have non-invertible fusion. This is in stark contrast with the unitary/linear category case where $\Rep (G)$ has non-invertible fusion if and only if there are higher than 1d irreps, which requires $G$ to be non-abelian. 
An  example that we will discuss is $\Rep (\Z_4^T)$, which has a simple object $Q$, that realizes a quaternionic representation, and has fusion 
\be
Q \otimes Q =  4 \times \bm{1} \,.
\ee

\smallskip
\noindent
{\bf Non-Invertible Time-Reversal.}
Our categorical framework allows for a unified description of both group-like invertible anti-unitary symmetry and more exotic non-invertible time-reversal symmetry. A non-invertible time-reversal symmetry can arise, for instance, when an invertible anti-unitary symmetry is combined with an internal non-invertible symmetry. One family of non-invertible time-reversal symmetries we will consider in this paper are a Galois-real version of Tambara-Yamagami categories, where the Kramers-Wannier duality defect is anti-unitary and thus time-reversal symmetry becomes non-invertible. These are denoted by ${\TY_{\overline{\C}}(A)}$. 
Unlike the case of unitary TY categories, here there is a choice of \emph{anti-symmetric} bicharacter instead of a symmetric one, and there is  no Frobenius-Schur indicator. A full classification was given in \cite{plavnik2024TYreal}. Here we will reconsider this from a physical perspective: the SymTFT provides a surprisingly simple and intuitive classification for this family of non-invertible time-reversal symmetries. 
There are also $\R$-real $\TY_\R(A)$ categories, which we do not consider as symmetry categories for the reason explained before.

\begin{figure}
$$
\begin{tikzpicture}[
    every node/.style={align=center},
    x=2.4cm,y=2.0cm
]

\node (TL) at (-1, 1) {$\mathbb{Z}_4^T$-trivial};
\node (TR) at ( 1, 1) {$\mathbb{Z}_4^T$-SPT};
\node (BL) at (-1,-1) {$\mathbb{Z}_4^T$-SSB};
\node (BR) at ( 1,-1) {$\mathbb{Z}_4^T \to \mathbb{Z}_2$\\ partial SSB};

\path
(TL) edge[loop above] node {$\operatorname{Rep}(\mathbb{Z}_4^T)$} ()
(TR) edge[loop above] node {$\operatorname{Rep}(\mathbb{Z}_4^T)$} ()
(BL) edge[loop below] node {\textcolor{blue}{$\Vec_{\mathbb{Z}_4^T}$}} ()
(BR) edge[loop below] node {\textcolor{blue}{$\Vec_{\mathbb{Z}_2 \times \mathbb{Z}_2^T}^{\alpha}$}} ()

(TL) edge[<->] node[above] {$\Vec$} (TR)
(TL) edge[<->] node[left] {$\Vec$} (BL)
(TR) edge[<->] node[right] {$\operatorname{Rep}(\mathbb{Z}_2)$} (BR)
(BL) edge[<->] node[below] {$\Vec_{\mathbb{Z}_2^T}$} (BR)

(TL) edge[<->] node[pos=.33, above left, fill=white, inner sep=1.5pt]
    {$\operatorname{Rep}(\mathbb{Z}_2)$} (BR)
(BL) edge[<->] node[pos=.67, below right, fill=white, inner sep=1.5pt]
    {$\Vec$} (TR);

\end{tikzpicture}
$$
\caption{Full structure of $(1+1)$d gapped phases with $\Z_4^T$ symmetry. Each vertex is a gapped phase. Each arrow is a category of symmetric domain walls. The blue categories are Morita equivalent symmetry (i.e. Galois-real) categories. The category $\Rep(\Z_4^T)$ in contrast is not an admissible symmetry category even though it is Morita equivalent to $\Vec_{\Z_4^T}$, as it is an $\R$-real fusion category. \label{fig:Z4TPhases}}
\end{figure}

    \begin{table*}[t]
    \centering
\renewcommand{\arraystretch}{1.35}
\setlength{\tabcolsep}{8pt}
\begin{adjustbox}{max width=\textwidth}
\begin{tabular}{c|c c c c c c}
 & $(S_3^T,1)$
 & $(S_3^T,\psi)$
 & $(\mathbb Z_2^T,1)$
 & $(\mathbb Z_2^T,\psi)$
 & $(\mathbb Z_3,1)$
 & $(1,1)$ \\
\hline
$(S_3^T,1)$
& $\operatorname{Rep}(S_3^T)$
& $\Vec_{\bH}^{\oplus 3}$
& $\Vec_{\mathbb R}$
& $\Vec_{\bH}$
& $\operatorname{Rep}(\mathbb Z_3)$
& $\Vec$ \\

$(S_3^T,\psi)$
& 
& $\operatorname{Rep}(S_3^T)$
& $\Vec_{\bH}$
& $\Vec_{\mathbb R}$
& $\operatorname{Rep}(\mathbb Z_3)$
& $\Vec$ \\

$(\mathbb Z_2^T,1)$
&
&
& $\operatorname{Rep}_{\mathbb R}(\mathbb Z_3)$
& $\Vec_{\bH}\oplus \Vec$
& $\Vec$
& $\Vec_{\mathbb Z_3}$ \\

$(\mathbb Z_2^T,\psi)$
&
&
&
& $\operatorname{Rep}_{\mathbb R}(\mathbb Z_3)$
& $\Vec$
& $\Vec_{\mathbb Z_3}$ \\

$(\mathbb Z_3,1)$
&
&
&
&
& \textcolor{blue}{$\Vec_{\mathbb Z_3 \times \mathbb Z_2^T}$}
& $\Vec_{\mathbb Z_2^T}$ \\

$(1,1)$
&
&
&
&
&
& \textcolor{blue}{$\Vec_{S_3^T}$}
\end{tabular}
\end{adjustbox}
\caption{$(1+1)$d gapped $S_3^T$-phases and categories of defects among them. The categories in blue are Morita equivalent symmetry categories, i.e. those that are Galois-real.}
\label{tab:S3Phases}
\end{table*}

\smallskip
\noindent
{\bf Gapped Phases.} 
We propose a module category approach for the classification of gapped phases with generalized anti-linear symmetries. This is akin to the unitary case studied in \cite{Thorngren:2019fuc}. Gapped phases with symmetry are in one-to-one correspondence with module categories over the now Galois-real symmetry category, which characterize how the symmetry category acts on the underlying gapped phase. E.g. we can reproduce with this approach the known twisted group-cohomology classification for invertible $G^T$-symmetries. 
The module category approach also provides an explicit description of the categories of defects within each gapped phase and between different phases in terms of module functors. 
We compute the categorical structure of all gapped phases with group-like symmetry $\Z_4^T$ -- see \Cref{fig:Z4TPhases} -- and $S_3^T=\Z_3\rtimes \Z_2^T$ -- see \Cref{tab:S3Phases} -- including the categories of defects within each phase and the categories of domain walls between phases. 
As a consequence of the analysis, we establish the known Morita equivalence~\footnote{It is known that $\mathbb{Z}_2 \times \mathbb{Z}_2^T$ with mixed anomaly is related to $\mathbb{Z}_4^T$ by gauging~\cite{Pace:2025hpb}. This relation is usually established through field-theoretic arguments or explicit lattice constructions. To the best of our knowledge, however, a proof in terms of Morita equivalence has not appeared prior to the present work.}: 
\be
\Vec_{\Z_4^T} \text{ Morita equivalent to } \Vec_{\Z_2\times \Z_2^T}^\alpha\,,
\ee
where $\alpha$ is a nontrivial mixed anomaly, and the lesser known and more surprising equivalences (\ref{RealSurprise}), which e.g. for $A= \Z_N$, become the statement
\be
\Vec_{\Z_N\times \Z_2^T} \text{ Morita equivalent to } \Vec_{D_N^T} \,.
\ee


The SPT phases with time-reversal symmetry have a long history, going back to the Haldane chain. One of the challenges  studying these phases is the absence of a local or string order parameter, for a discussion, see \cite{Pollmann:2010yk, Perez-Garcia:2008hhq, Chen:2011pg, Chen:2011hnt, Pollmann:2012xdn, Shiozaki:2016cim}. 
The $\Z_2^T$ phase can be detected by computing a partition function on real projective space $\R P^2$, and on a space with boundary the edge modes can be identified with a Kramers doublet. However, there seems to be no direct bulk order parameter. 
Other (1+1)d gapped phases with time-reversal symmetry and the interplay of this with gapped phase classification has been studied in \cite{Schuch:2011niz, fidkowski2011topological, Garre-Rubio:2022uum}.

\smallskip
\noindent
{\bf Dual Symmetries.} 
One of the advantages of the categorical formulation of symmetries and phases is that it allows for direct and rigorous computation of dualities (generalized gauging) between symmetries. We extend the theory of generalized gauging~\cite{Tachikawa:2017gyf, Bhardwaj:2017xup} to include anti-linear symmetries, and use the Morita theory for real fusion categories to define dual anti-linear symmetries. As an application of our formalism, we establish the following family of dualities/gauge-relations generated by gauging central subgroups.
\begin{theorem*}[\Cref{thm_gauge_subgroup}]
    Let $(G^T, s: G^T\to \Z_2^T)$ be an anti-unitary group-like symmetry, and $N< \ker(s)$ be a unitary central subgroup. Denote by $K^T:=G^T/N$ the quotient. Then there is Morita equivalence 
    \be
\Vec_{G^T}\simeq_{\text{Morita}} \Vec_{\widehat{N}\rtimes K^T}^\omega\,,
    \ee    
where $G^T/N$ acts on $\widehat{N}:=\Hom(N, U(1))$ as inversion via $s$, and 
    \be
     \omega((k_1,\gamma_1),(k_2,\gamma_2),(k_3,\gamma_3))=\gamma_1(e_2(k_2,k_3)) \,,
    \ee
    where $e_2\in H^2(K^T, N)$ is the extension class of the extension $N\to G^T\to K^T$. 
\end{theorem*}
Compared with the result of gauging central subgroup for a unitary group, the main new feature here is the semi-direct product $\widehat{N}\rtimes K^T$.

\smallskip
\noindent
{\bf $\Z_2^T$-Enriched SymTFT.}
The SymTFT \cite{Ji:2019jhk, Gaiotto:2020iye, Apruzzi:2021nmk, Freed:2022qnc} 
has been a key tool in carrying out the Categorical Landau Paradigm Program, for gapped phases in 1+1d \cite{Bhardwaj:2023ayw, Bhardwaj:2023idu, Bhardwaj:2024kvy, Chatterjee:2024ych, Bhardwaj:2024wlr, Warman:2024lir,  Bottini:2025hri,  Schafer-Nameki:2025fiy} and 2+1d \cite{Bhardwaj:2024qiv, Xu:2024pwd, Bullimore:2024khm,  Bhardwaj:2025piv,  Inamura:2025cum,Decoppet:2024htz}, including second order, symmetric, phase transitions \cite{Chatterjee:2022tyg, Bhardwaj:2023bbf, Wen:2023otf, Huang:2023pyk, Huang:2024ror, Bhardwaj:2024qrf,  Wen:2024qsg,  Bhardwaj:2025jtf, Wen:2025thg}, and mixed phases \cite{Schafer-Nameki:2025fiy, Qi:2025tal, Luo:2025phx}, and most recently continuous internal symmetries \cite{Antinucci:2024zjp, Brennan:2024fgj, Bonetti:2024cjk, Apruzzi:2024htg, Apruzzi:2025hvs, Stockall:2025ngz}. 
The fundamental idea is to start with a theory with a global symmetry $\cC$, coupling this to background fields, and gauge in one dimension higher. This protocol applies to both internal symmetries as well as \emph{continuous} spacetime symmetries \cite{Apruzzi:2025hvs}.
However, for discrete spacetime symmetries such as time-reversal and translations, it was argued in \cite{Pace:2025hpb}, that instead of gauging these symmetries, the SymTFT for internal symmetries is merely enriched by the spacetime symmetries, in the sense of 
\cite{Barkeshli:2014cna, Barkeshli:2016mew}. From an unoriented Turaev-Viro theory perspective some of these aspects appeared in \cite{Bhardwaj:2016dtk}.

\smallskip
\noindent{\bf Boundary Conditions and Symmetry Quiches.} Boundary conditions (BCs), and more generally condensable algebras, in symmetry enriched SymTFTs or more generally topological orders have appeared before in \cite{ bischoff2019spontaneous, Cheng:2020rpl, Pace:2024acq, Lu:2025gpt, Pace:2025hpb}. One focus here has been on symmetric boundary conditions, or at least boundary conditions that are stable under the enrichment symmetry, i.e. the group of condensed anyons are invariant under the action. We will relax this in order to be able to incorporate all anti-linear symmetries, which include nontrivial extensions of $\Z_2^T$, and even non-invertible time-reversal symmetries. For such symmetry considerations, we only require the SymTFT quiche -- the SymTFT with one symmetry gapped boundary condition -- rather than the full interval compactification that classifies phases. 

If the BC is not $\Z_2^T$-invariant, or there is a non-trivial fractionalization class for the condensed anyons on the boundary, then the boundary breaks spontaneously the $\Z_2^T$ symmetry~\cite{Pace:2025hpb}. For the purpose of identifying the symmetry category this is not a hindrance -- and has the advantage of allowing us to classify all possible Morita equivalent anti-linear symmetry categories.

The spontaneous breaking of the enriching $\Z_2^T$-symmetry on the symmetry boundary complicates the formulation of SymTFT for gapped phases. In particular, if we were to use this framework to study gapped phases, i.e. include a gapped physical boundary condition, then the resulting $(1+1)$d phase will always break the anti-linear symmetry down to a linear sub-symmetry. Again, we do not attempt this here, as we have a succinct and comprehensive way to study all gapped phases in a pure $(1+1)$d framework using module categories in \Cref{sec_GappedPhases}.

In summary, in this paper we focus on SymTFT quiches, with a $\Z_2^T$-enriched SymTFT bulk and a gapped symmetry boundary and show that this is a very useful way to characterize all Morita equivalent anti-linear symmetries.

\smallskip
\noindent
{\bf Plan of the Paper.}
In~\Cref{sec_Z2TPhases} we motivate real fusion categories directly from simple $(1+1)$d systems with time-reversal symmetry. Using a $\Z_2^T$-symmetric spin chain, we derive the two basic categories associated with time-reversal, namely $\Rep(\Z_2^T)\simeq \Vec_{\R}$ and $\Vec_{\Z_2^T}$, from the categories of symmetric defects in the trivial and symmetry-broken phases. We then revisit the Haldane chain, explain the appearance of the quaternion algebra $\bH$ and the real category $\Vec_\bH$, discuss the lack of string order parameters for anti-linear symmetries, and formulate the module-category description of $(1+1)$d $\Z_2^T$-phases. In~\Cref{sec_basic_rfc} we discuss the general theory of real fusion categories. We distinguish the two types, $\R$-real and Galois-real fusion categories, explain their different physical roles, and review the theory of module categories and Morita equivalence over real fusion categories. In~\Cref{sec_example_rfc} we illustrate these structures with concrete examples. In particular, we discuss the group-theoretical categories associated with anti-linear groups $G^T$, the categories $\Rep(G^T)$, non-invertible time-reversal symmetries of Tambara--Yamagami type, and the semi-direct product categories $\cC\rtimes \mathbb{Z}_2^T$. In~\Cref{sec_GappedPhases} we turn to the classification of $(1+1)$d gapped phases with generalized anti-linear symmetries. 
We propose that such phases are classified by module categories over the Galois-real symmetry category, develop the corresponding notion of duality and gauging, and revisit the basic $\mathbb{Z}_2^T$ examples from this perspective. 
We then analyze the general structure of phases with group-like anti-linear symmetry and study in detail the examples of $\mathbb{Z}_4^T$ and $S_3^T=\mathbb{Z}_3\rtimes \mathbb{Z}_2^T$, including their defects, domain walls, and Morita dual symmetry categories. We conclude the section by proving the family of Morita equivalences generated by gauging finite unitary subgroups. 

We change gears in~\Cref{sec_T_SymTFT} and we develop the $\mathbb{Z}_2^T$-enriched  (or simply T-enriched) SymTFT framework. We explain the distinction between T-enrichment and gauging, formulate the bulk theory in terms of T-enriched topological order, and identify the relevant bulk defects with the Drinfeld center of a Galois-real fusion category. In~\Cref{sec_quiche} we apply this framework to SymTFT quiches for anti-linear symmetries. We explain how a Galois-real symmetry category arises on a boundary of a $\mathbb{Z}_2^T$-enriched SymTFT, and show through examples how Morita equivalent anti-linear symmetries appear as different symmetry boundaries of the same bulk. The appendices collect some technical material, including additional discussion of the semi-direct product categories $\cC\rtimes \mathbb{Z}_2^T$ and the proof of the gauging theorem.
We conclude with some open questions in \Cref{sec:Conclusions}.

\section{Motivating Real Fusion Categories from $\Z_2^T$-symmetric Phases}
\label{sec_Z2TPhases}

One of the key points of this paper is that real fusion categories are the relevant mathematical structure underlying the study of systems with anti-linear symmetries. 
The goal of this section is to explain, from a physical perspective, how real fusion categories naturally arise in the study of anti-linear symmetries, such as $\Z_2^T$. 

The most basic example of an anti-linear symmetry is the time-reversal symmetry $\Z_2^T$. In this section we will derive two real fusion categories associated with it, denoted  $\Vec_{\Z_2^T}$ and $\Rep(\Z_2^T)$, respectively, from a concrete 1D spin chain with $\Z_2^T$-symmetry. The notation is deliberately chosen to match with the $\Vec_G,~\Rep(G)$ categories associated with a finite, internal, $\C$-linear $G$-symmetry. 
The $\Rep(\Z_2^T)$ category is in fact the same as $\Vec_\R$ the category of (finite dimensional) real vector spaces
\be
\Rep (\Z_2^T) \simeq \Vec_\R \,.
\ee
On the other hand $\Vec_{\Z_2^T}$ is a less familiar category. It is similar to the $\C$-linear $\Vec_{\Z_2}$ category except the generator $T\in \Vec_{\Z_2^T}$ is anti-linear in a precise sense. We will argue that it is the symmetry category for time-reversal symmetry, just like $\Vec_G$ is the symmetry category for an internal $G$ group symmetry.

Recall that $\Vec_G$ can be regarded as the category of symmetric defects in a 1D $G$-SSB phase. Namely the symmetric defects in a $G$-SSB phase are the $G$-domain walls, they are labeled by the group elements and satisfy the group law under fusion. On the other hand $\Rep(G)$ is the category of symmetric defects in a 1D $G$-SPT. Namely the symmetric defects (or excitations) in a $G$-SPT are $G$-charges that are labeled by representations of $G$. They fuse according to the tensor product of representations. Put differently, in a SymTFT realization of a $G$-SSB phase, we choose the physical boundary to be associated to the Dirichlet BC, which has topological defects associated to $\Vec_G$. In turn  for a $G$-SPT, it is the Neumann BC which has topological defects given by $\Rep (G)$.

Following this example, we investigate categories of symmetric defects in a 1D spin chain with $\Z_2^T$-symmetry in the symmetric and SSB phase. Notice in 1D there is a non-trivial $\Z_2^T$-SPT called the Haldane phase~ \cite{Haldane:1983ru}. 

We start with the Ising model with Hamiltonian 
\begin{equation}\label{eq_Ising}
    H(J)=-(1-J)\sum_i X_i-J \sum_i Z_iZ_{i+1}\,,
\end{equation}
where $X_i, Z_i$ are the Pauli operators. The Hamiltonian is symmetric under various anti-unitary symmetries. One choice is 
\begin{equation}\label{UTDef}
    U_T:=K  \prod_i X_i=K U\,,
\end{equation}
where $K$ means complex conjugation in the basis where $Z_i$ are diagonal and $U = \prod_i X_i$ is the usual $\Z_2$-symmetry generator in the transverse Ising model. The choice of the anti-unitary symmetry action is not unique and will not affect the conclusions in this section. With the choice of generator (\ref{UTDef}) which we will denote by $\Z_2^T$, the ground state is $\Z_2^T$-symmetric when $J<1/2$ and breaks $\Z_2^T$ spontaneously when $J>1/2$. 

{In the following we will consider the system in infinite volume and identify defects as follows:
for $H$ a local Hamiltonian, a defect is a subspace of the total lattice Hilbert space that looks like the ground state space of $H$ near spatial infinity.  } {In other words a defect is the ground state space of a Hamiltonian that only differs from $H$ in a local region. When symmetry is involved, a symmetric defect is the ground state space of a symmetric Hamiltonian that differs from $H$ in a local region. In this language, the trivial defect is the ground state space of $H$ itself. The physical reason behind this definition is that in the IR limit only the ground state space is ``visible" for a gapped system.}

\subsection{$\Z_2^T$-Trivial Phase} 

We will now prove 
\begin{theorem}\label{th:RepZ2T}
    The category of symmetric defects in a $\Z_2^T$-symmetric trivial phase is $\Rep(\Z_2^T)\simeq \Vec_\R$. 
\end{theorem}
First we will motivate this from the spin system.
Take $J=0$, then the ground state space  of the system is 
\begin{equation}
   \Omega= \left\{z \ket{+\cdots +} \ ; \,  z\in \C\right\}\,,
\end{equation}
where $|+\rangle$ is the eigenstate of $X$ with eigenvalue 1. 
We identify this with the trivial defect.
One is tempted to say that a defect can be created by acting with $Z_k$ at site $k$, which changes the ground state space to 
\begin{equation}
   Z_k \Omega= \{z|+\cdots +-_k+\cdots+\rangle\ ; \, z\in \C\}\,.
\end{equation}
For the internal $\Z_2$-symmetry generated by $U$, this is indeed a non-trivial defect, namely the $\Z_2$-charge. However, we notice that here $iZ_k$ is a $\Z_2^T$-symmetric local operator
\be
U_T\, (i Z_k) =  (iZ_k)\, U_T \,,
\ee
that recovers the original ground state space: 
\begin{equation}
    (iZ_k)\times Z_k\times \Omega=\Omega\,.
\end{equation}
A defect/excitation that can be created or destroyed by a symmetric local operator is trivial, hence the spin-flip defect $|+\cdots -_k\cdots+\rangle$ is in fact trivial as a $\Z_2^T$-symmetric defect. We see that the category of $\Z_2^T$-symmetric defects in this $\Z_2^T$-symmetric phase has essentially one simple object, namely the ground state space itself. All defects can be related back to the ground state space by acting with local $\Z_2^T$-symmetric operators. 

A slightly more mathematical reasoning is as follows: local operators in this $\Z_2^T$-symmetric phase should transform as representations of $\Z_2^T$, where by a representation of $\Z_2^T$ we mean a complex vector space $V$ together with a $\C$-anti-linear operator 
\be 
T: V\to V\,,\qquad T^2=\mathrm{id}_V  \,.
\ee
For instance, the space generated by the local operator $Z_k$ is 
\begin{equation}
   V_{Z_k}=\{z Z_k|z\in \C\}\,,
\end{equation}
with the $\Z_2^T$-action: 
\be 
T(z Z_k)=-z^* Z_k\,.
\ee 
If we identify $V_{Z_k}$ with the complex plane via 
\begin{equation}
    (a+bi)Z_k\  \longleftrightarrow\ ai+b ,~a,b\in \R\,,
\end{equation}
then the $\Z_2^T$-action is simply complex conjugation: 
\be 
T(ai+b)=-ai+b\,.
\ee
We see that the space spanned by the operator $Z_k$ together with its $\Z_2^T$-action is equivalent to the complex plane with the standard complex conjugation action. In fact, any irreducible representation of $\Z_2^T$ is equivalent to the  complex plane with $\Z_2^T$ acting as complex conjugation, $T$. This becomes the mathematical statement that the category $\Rep(\Z_2^T)$ of $\Z_2^T$-representations contains only one simple object, $(\C,T)$, up to equivalence. 

Let us consider the endomorphisms of the trivial defect $\Omega$.  A  {morphism of defects} $\Omega\to \Omega$ is a linear map that \emph{can be realized by a symmetric local operator.} Assume the linear map $|+\cdots +\rangle\to z|+\cdots +\rangle$ is realized by a $\Z_2^T$-symmetric local operator $\cO$: 
\begin{equation}
    \cO|+\cdots +\rangle=z|+\cdots +\rangle
\end{equation}
Then since $U_T\cO U_T=\cO$, we have 
\be 
\ba
 &{z|+ \cdots +\rangle}
    =U_T\cO U_T {|+\cdots +\rangle}=U_Tz{|+\cdots  +\rangle }\cr 
    & {=z^*|+\cdots +\rangle }\,.
\ea 
\ee
Therefore only multiplication by real numbers is allowed as maps of defects. We conclude that in the category of symmetric defects, we have $\End(\lid)\simeq \R$. Put everything together we see that the category of symmetric defects in the trivial $\Z_2^T$-symmetric phase is a real fusion category generated by a single simple object $\lid$ satisfying $\End(\lid)\simeq \R$. This is nothing but the category $\Vec_\R$ of finite dimensional real vector spaces. In fact $\Rep(\Z_2^T)$ is equivalent to $\Vec_\R$: a complex vector space with anti-unitary $\Z_2^T$-action is the same as a real vector space by taking the fixed points of the $\Z_2^T$-action: 
\begin{equation}
    \Z_2^T\curvearrowright V \Leftrightarrow V_\R:= \{v\in V|T(v)=v\}\,.
\end{equation}
In summary we obtain  ~\Cref{th:RepZ2T}.

\subsection{$\Z_2^T$ SSB-Phase} \label{sec_Z2T_broken}

The other phase is an SSB-phase for $\Z_2^T$.
We prove the following theorem:  
\begin{theorem}\label{th:VecZ2T}
    The category of symmetric defects in a $\Z_2^T$-SSB phase is $\Vec_{\Z_2^T}:=\End_\R(\Vec)$.
\end{theorem}

First let us again motivate this from the spin-chain.  In the spin system it is obtained by considering $J=1$, then ground state space of the system is 
\begin{equation}
    \Omega=\left\{z|\uparrow\cdots \uparrow\rangle+w|\downarrow\cdots \downarrow\rangle \,; \ z,w\in \C \right\}\simeq \C\oplus \C\,,
\end{equation}
where $|\uparrow\rangle$ and $|\downarrow\rangle$ are the eigenstates of $Z$ with eigenvalues $\pm 1$. The domain wall defect is created by the symmetric non-local operator 
\be 
\D_k:=\prod_{j>k} X_j \,,
\ee
which acts on the ground state space as 
\be 
    \D_k: \quad 
    \left\{
    \ba ~&|\uparrow\cdots \uparrow\rangle\mapsto |\uparrow\cdots \uparrow_k \downarrow_{k+1}\cdots \downarrow\rangle\\
    &|\downarrow\cdots \downarrow\rangle\mapsto |\downarrow\cdots \downarrow_k \uparrow_{k+1}\cdots\uparrow\rangle
\ea\right. \,.
\ee
This must be a non-trivial defect, since any local operator cannot change the behavior of a ground state near spatial infinity. The location $k$ of the defect is irrelevant for the topological sector of the defect, so we will call this defect $\D$. Fusion of defects is defined by concatenation. The fusion $\D\otimes \D$ is therefore represented by the space spanned by the states with two domain walls: 
\begin{align}
 \D\otimes \D&\simeq\left\{ z   |\uparrow\cdots \uparrow_j\downarrow_{j+1}\cdots \downarrow_k\uparrow_{k+1}\cdots \uparrow\rangle \right. \nonumber\\
    &+w|\downarrow\cdots \downarrow_j\uparrow_{j+1}\cdots\uparrow_k\downarrow_{k+1}\cdots \uparrow\rangle \left|\, z,w\in \C \right\} \,.
\end{align}
The $\Z_2^T$-symmetric local operator $\prod_{j<l<k+1}X_l$ kills the pair of domain walls. We conclude the fusion rule is $\D\otimes \D\simeq \lid$. 

To fully determine the structure of the category of symmetric defects we need to do more work. 
Next we consider the endomorphisms of the trivial defect $\Omega$. Assume that there is a $\Z_2^T$-symmetric local operator $\cO$ that acts on the ground state space as a general linear transformation
\begin{equation}
    (z,w)\mapsto (\alpha z+\beta w, \gamma z+\eta w)\,,
\end{equation}
where $(z,w)$ is short hand for $z|\uparrow\cdots \uparrow\rangle+w|\downarrow\cdots \downarrow\rangle$. Local operators cannot change the behavior of states near spatial infinity, hence $\beta=\gamma=0$. Since $\cO$ is $\Z_2^T$-symmetric, we have 
\begin{align}
    (\alpha z, \eta w)&=\cO(z,w)=U_T\cO U_T(z,w)\nonumber\\
    &=U_T\cO(w^*,z^*)=U_T(\alpha w^*, \eta z^*)=(\eta^* z, \alpha^*w).
\end{align}
We see that $\eta=\alpha^*$. Therefore the only allowed maps of defects are 
\be 
(z,w)\mapsto (\alpha z, \alpha^* w),~\alpha \in \C
\,.
\ee
We conclude that in the category of symmetric defects in the $\Z_2^T$-SSB phase we have 
\be 
\End(\lid)\simeq \C\,.
\ee
One might conclude that the category of symmetric defects is the (complex) fusion category $\Vec_{\Z_2}$, with the non-trivial simple object being the domain wall defect $\D$. However, this is not the case and a subtle but crucial difference will be revealed by further analysis. Let us consider a domain wall $\D_k$ at site $k$, and a local symmetric operator $\cO_j$ realizing the endomorphism $(z,w)\mapsto (\alpha z, \alpha^* w)$ at site $j$. One choice is $\cO_j^\alpha=a\, \mathrm{id} +i \, b \, Z_j,~\alpha=a+bi$. On the $|\uparrow\cdots \uparrow\rangle$ ground state, $\cO_j$ acts as $\alpha \mathrm{id}_\uparrow$, while on the $|\downarrow\cdots \downarrow\rangle$ ground state it acts as $\alpha^*\mathrm{id}_\downarrow$. Let $j<k<l$, we notice the crucial relation
\begin{equation}
    \cO_j^\alpha \D_k=\D_k\cO_l^{\alpha^*}  (Z_j Z_l)
\end{equation}
The local operator $Z_jZ_l$ is $\Z_2^T$-symmetric hence belongs to the trivial superselection sector. Therefore up to $\Z_2^T$-symmetric local operators we have the relation
\begin{equation}
    \cO^\alpha\otimes  \D=\D\otimes \cO^{\alpha^*} \,,
\end{equation}
which says that a local scalar $\alpha \in \End(\lid)$ is complex conjugated when it passes through the defect $\D$. This is expected: the defect $\D$ is a \textbf{time-reversal domain wall}, which should be the Wick-rotation of an anti-linear action, thus it should act on scalars as complex conjugation.

 Another way to look at this is the following. The defect $\D_k$ is the same as $\lid\otimes \D_k$, this means any endomorphism of the trivial defect $\alpha: \lid\to \lid$ gives rise to an endomorphism of the defect $\D_k$ by tensor product: $\alpha \mathrm{id}_\lid \otimes \mathrm{id}_\D$. Concretely this means we obtain an endomorphism of $\D$ by putting the operator $\cO_j$ to the left of the domain wall $\D_k$($j<k$). This leads to an endomorphism:
 \begin{widetext}
\be 
\D_k\to\D_k:~\qquad 
\ba 
{} &|\uparrow\cdots \uparrow_k \downarrow_{k+1}\cdots \downarrow\rangle\mapsto \cO_j |\uparrow\cdots \uparrow_k \downarrow_{k+1}\cdots \downarrow\rangle=\alpha|\uparrow\cdots \uparrow_k \downarrow_{k+1}\cdots \downarrow\rangle\\
 &|\downarrow\cdots \downarrow_k \uparrow_{k+1}\cdots\uparrow\rangle\mapsto \cO_j|\downarrow\cdots \downarrow_k \uparrow_{k+1}\cdots\uparrow\rangle=\alpha^* |\downarrow\cdots \downarrow_k \uparrow_{k+1}\cdots\uparrow\rangle\,.
\ea
\ee
 \end{widetext}
Similarly, the defect $\D_k$ is the same as $\D_k\otimes \lid$, therefore putting any endomorphism of $\lid$ to the right of $\D_k$ also gives an endomorphism of $\D_k$. Concretely, acting with $\cO_j$ on $\D_k$ with $j>k$ leads to 
\begin{widetext}
\be
\D_k\to\D_k:\quad 
\ba{}
 &|\uparrow\cdots \uparrow_k \downarrow_{k+1}\cdots \downarrow\rangle\mapsto \cO_j |\uparrow\cdots \uparrow_k \downarrow_{k+1}\cdots \downarrow\rangle=\alpha^*|\uparrow\cdots \uparrow_k \downarrow_{k+1}\cdots \downarrow\rangle \\
 &|\downarrow\cdots \downarrow_k \uparrow_{k+1}\cdots\uparrow\rangle\mapsto \cO_j|\downarrow\cdots \downarrow_k \uparrow_{k+1}\cdots\uparrow\rangle=\alpha |\downarrow\cdots \downarrow_k \uparrow_{k+1}\cdots\uparrow\rangle \,.
\ea
\ee 
\end{widetext}
We see that the left and right actions on $\D_k$ differ by complex conjugation. We conclude that the category of symmetric defects in the $\Z_2^T$-SSB phase has the following properties. 
\begin{itemize}
    \item It is generated by two simple objects $\lid, \D$. 
    \item The fusion rule is $\D\otimes \D=\lid$. 
    \item The object $\lid$ satisfies $\End(\lid)\simeq \C$. 
    \item The object $\D$ satisfies the relation 
    \be
    z\otimes \D=\D\otimes z^*,~z\in \C.
    \ee
\end{itemize}
Although $\End(\lid)\simeq \C$, the last relation above implies that this category is \emph{not} a $\C$-linear fusion category. It is only a real fusion category. We will denote this category as $\Vec_{\Z_2^T}$. In the mathematical literature this category appeared in \cite{sanford2025fusion}, where a concrete model is given as $\Bimod_\R(\C)$, the category of $\R$-linear $\C$-bimodules. Another concrete model is $\Vec_{\Z_2^T}=\End_\R(\Vec)$ the category of $\R$-linear endomorphisms of $\Vec$. We thus find ~\Cref{th:VecZ2T}. 

We recall that in the internal $G$-symmetry case, one way to define the symmetry category is as the category of defects in the $G$-SSB phase. This gives the $\Vec_G$ category for non-anomalous $G$-symmetry, and $\Vec_G^\omega$ for anomalous $G$-symmetry. Namely, the domain walls in an SSB phase remember not only the group law but also the anomaly. 
Similarly, we should be able to define the symmetry category for $\Z_2^T$ as the category of defects in the $\Z_2^T$-SSB phase. Hence we will take $\Vec_{\Z_2^T}$ as the symmetry category for $\Z_2^T$. In $(1+1)$d, there is no $\Z_2^T$-anomaly since $H^3(\Z_2^T, U(1)_T)=0$. However, in other dimensions, there could be time-reversal anomalies and the symmetry category should be able to capture the anomalies. Indeed, the $\Vec_{\Z_2^T}$ category can be generalized to higher dimensions. In \cite{D_coppet_2025} the authors also constructed the 2-category $2\Vec_{\Z_2^T}^\omega$, where $\omega \in H^4(\Z_2^T, U(1)_T)=\Z_2$ can be a non-trivial class. In this paper we will focus on $(1+1)$d symmetries and phases, there can still be mixed anomalies between $\Z_2^T$ and other internal symmetries, which should be captured by symmetry categories generalizing $\Vec_{\Z_2^T}$. We will discuss these categories in the next section.

\subsection{$\mathbb{H}$ is for Haldane}
\label{sec_haldane}
We have discussed the symmetry and charge categories associated with  the $\Z_2^T$ symmetry. Similar to the case with unitary group $G$ symmetries, if there are two SPTs protected by $\Z_2^T$, then the category of symmetric defects is $\Rep(\Z_2^T)$ in either SPT. 
Hence the category of bulk symmetric defects cannot distinguish different SPTs. 
There are generally two approaches to detecting SPT phases. One is via boundary/edge measurements, the other is via bulk string order parameters. The latter approach naturally leads to the SymTFT paradigm by realizing that the string order parameters are naturally anyons of a $(2+1)$d topological order, that end on the physical boundary but do not end on the symmetry one. 

In this section we consider such approaches to the non-trivial $\Z_2^T$-SPT called the Haldane chain. The physical properties of the Haldane chain, including its boundary modes, are well-known in physics. Our focus here is on the categorical characterization of the Haldane chain. We will discuss two aspects of the Haldane chain from a categorical perspective. First, we show that the category of boundary conditions in the Haldane chain is $\Vec_\bH$, where $\bH$ is the quaternion algebra. Second, we will show that the Haldane chain does not have a string order parameter, consequently it cannot be described by anyon condensation in a $(2+1)$d SymTFT as for SPTs with unitary $G$-symmetry. Finally, we will provide a module category approach that provides a complete characterization of $(1+1)$d phases with $\Z_2^T$-symmetry.

\subsubsection{The Haldane Chain and its Edge-Modes}

The Haldane phase with one boundary can be modeled by the following Hamiltonian:
\begin{equation}
H_0=-\sum_{i=0}^{\infty} Z_i X_{i+1} Z_{i+2}
\end{equation}
with the $\Z_2^T$-symmetry generated as in (\ref{UTDef}) by 
\begin{equation}
U_T=K\prod_{i=0}^\infty X_i=KU \,.
\end{equation}
The ground state space of $H_0$, denoted as $\Omega$, is two dimensional. One checks that the operators 
\begin{equation}
    \cX:=X_0 Z_1,\quad \cZ:=Z_0
\end{equation}
 commute with the Hamiltonian, hence they act on the ground state space as a pair of Pauli operators. We will write $|\uparrow\rangle_0, |\downarrow\rangle_0$ for the basis of $\Omega$ that are eigenstates of $\cZ$. Since we do not have explicit ground state wavefunction, we will work out the action of $U_T$ on the ground state space algebraically. Notice that $U_T\cX U_T=-\cX$ and $U_T\cZ U_T=-\cZ$. Assume $U_T|\uparrow\rangle_0=e^{i\theta}|\downarrow\rangle_0$, then 
 \be
 U_T|\downarrow\rangle_0=U_T\cX|\uparrow\rangle_0=-e^{i\theta}|\uparrow\rangle_0
 \ee
 Next we consider endomorphisms of the boundary condition $\Omega$. Similar to the case with defects, a morphism of boundary conditions is a linear map between ground state spaces that can be realized by symmetric local operators. Assume there is a local operator $\cO$ such that 
\be
\cO|\uparrow\rangle_0=\alpha |\uparrow\rangle_0+\beta|\downarrow\rangle_0,~ \cO|\downarrow\rangle_0=\gamma|\uparrow\rangle_0+\delta |\downarrow\rangle_0.
\ee
Using $U_T\cO=\cO U_T$ and the action of $U_T$ on the ground state space, we obtain that $\delta=\alpha^*,~\gamma=-\beta^*$. Hence in the $\cZ$-basis the most general form of an endomorphism is 
\be
\begin{pmatrix}
    \alpha&\beta\\
    -\beta^*& \alpha^*
\end{pmatrix}
\ee
Such matrices form an $\R$-linear algebra, with an $\R$-linear basis given by $\lid, i \cX, i\cY, i\cZ$($\cX, \cY, \cZ$ are Pauli matrices). Checking their algebraic relations, we see that the endomorphism algebra of the boundary $\Omega$ is exactly the quaternion algebra $\bH=\R\langle1,i,j,k\rangle/\langle i^2=j^2=k^2=ijk=-1\rangle$. More generally, we observe that the restriction of $U_T$ on the ground state space satisfies $U_T^2=-1$. This is true not only for the Hamiltonian $H_0$ but for any boundary condition: assume there is a Hamiltonian with one boundary that differs from $H_0$ only near the boundary. Then there exists some $k>0$ such that $Z_iX_{i+1} Z_{i+2}=1$ on the ground state space for any $i\ge k$. Therefore restricting to the ground state space we can replace $U_T$ by 
\be
U_T \prod_{i\ge k} (Z_iX_{i+1}Z_{i+2})=X_0X_1\cdots X_k(Z_kZ_{k+1})
\ee
which is supported near the boundary and  squares to $-1$. Hence no matter what the Hamiltonian looks like near the boundary, the ground state space will be acted by an anti-unitary operator that squares to $-1$. To formalize the situation let us introduce the algebra $\C[\Z_2^T]^\eta$, which as a linear space is $\{z+w U_T|z,w\in \C\}$, and satisfies $zU_T=U_T z^*,~U_T^2=-1$. Then the above discussion shows any boundary condition leads to a ground state space that is a representation of the algebra $\C[\Z_2^T]^\eta$. In fact 
\be 
\C[\Z_2^T]^\eta \simeq \bH
\ee 
follows by identifying 
\be 
i \leadsto i, \  U_T\leadsto j, \ 
iU_T\leadsto k
\,.
\ee 
We conclude that the category of boundary conditions of the Haldane chain is exactly $\Vec_\bH$ the category of $\bH$-modules. The specific boundary condition $H_0$ can be thought of as the minimal boundary condition that corresponds to the simple object $\bH\in \Vec_\bH$.

\subsubsection{Lack of string orders}\label{sec_op_Z2T}
For an SPT protected by a unitary $G$-symmetry, it is well-known that string order parameters can be used to detect the SPT phase. Analyzing the mathematical structure of string order parameters is in fact one path towards the SymTFT paradigm. We now consider the situation with the $\Z_2^T$-SPT. We will give arguments that $\Z_2^T$-SPTs do not have string order parameters. 

A string order parameter should be a truncation of the symmetry generator on finite intervals (or half-infinite chain). This means the string order acts trivially on operators away from the interval and acts the same way as the symmetry generator on operators deep inside the interval. However, since complex conjugation is involved in the generator $U_T$, there is no sensible truncation of $U_T$ to finite intervals. This is because complex conjugation always acts non-trivially everywhere. We may write down $(U_T)_{jk}:=K \prod_{j<l<k}X_l$, but $(U_T)_{jk}$ is \emph{not} supported on the interval $[j,k]$. For example $(U_T)_{jk}(Y_n)(U_T)_{jk}=-Y_n,~\forall n\notin [j,k]$. 

A mathematical argument is given by computing the Drinfeld center of the symmetry category $\Vec_{\Z_2^T}$. For a unitary $G$-symmetry, the collection of all local and string order parameters form the Drinfeld center $\cZ(\Vec_G)$. We will argue later~(\Cref{sec_op_center}) that this is still true for anti-linear symmetries. However, it is known that $\cZ(\Vec_{\Z_2^T})=\Vec_\R$ is trivial. It will be more illuminating to look at the structure of $\cZ(\Vec_{G^T})$. For $\Vec_G$, we know that an object in $\cZ(\Vec_G)$ is labeled by a flux sector and a charge sector. For $\cZ(\Vec_{G^T})$ this is still true, but the flux sector is strictly in $G_0<G^T$, the unitary subgroup. This reflects the physical fact that any anti-unitary symmetry does not have a string order, due to the non-local nature of complex conjugation, see also \cite{Pollmann:2012xdn}.  

\subsubsection{Module category description of $\Z_2^T$-phases}

We now turn to the module category approach of classification of phases and show that it provides a complete characterization of $(1+1)$d $\Z_2^T$-symmetric phases. Conceptually this is possible because the module category approach is essentially a boundary approach: it describes phases by their boundary modes. 
We can summarize the structure of  module categories for $\Vec_{\Z_2^T}$, $\Mod(\Vec_{\Z_2^T})$ with the  following diagram:
\begin{widetext}

\be
\begin{tikzcd}
    &\arrow[loop above,"\Vec_{\Z_2^T}"]\arrow[leftrightarrow,dl,"\Vec"']\Vec_{\Z_2^T}\arrow[leftrightarrow, rd,"\Vec"]&\\
    \Vec\arrow[loop left, "\Rep(\Z_2^T)\simeq\Vec_\R"]\arrow[leftrightarrow,rr,"\Vec_{\bH}"']&&\Vec^\eta\arrow[loop right, "\Rep(\Z_2^T)\simeq\Vec_\R"]
\end{tikzcd}
\ee
If we relabel vertices by phases they represent, we obtain
\be
\begin{tikzcd}
    &\arrow[loop above,"\Vec_{\Z_2^T}"]\arrow[leftrightarrow,dl,"\Vec"']\text{SSB}\arrow[leftrightarrow, rd,"\Vec"]&\\
    \text{Trivial}\arrow[loop left, "\Vec_\R"]\arrow[leftrightarrow,rr,"\Vec_{\bH}"']&&\text{Haldane}\arrow[loop right, "\Vec_\R"]
\end{tikzcd}
\ee
\end{widetext}

This diagram contains all the information about $(1+1)$d gapped $\Z_2^T$-phases. The endomorphism category of each vertex is the category of symmetric defects in the phase. We see that the SSB phase has endomorphism $\Vec_{\Z_2^T}$, exactly the category of symmetric defects we derived in~\Cref{sec_Z2T_broken}. On the other hand the two symmetric phases have endomorphism $\Vec_\R$. The domain wall category between the trivial phase and the Haldane phase is $\Vec_\bH$, agreeing with our derivation in~\Cref{sec_haldane}. The diagram also tells us the domain wall category between an SPT and an SSB phase is $\Vec$. It is also possible to verify this using a lattice model, which we leave to the reader.

We will discuss these now in turn. 
The symmetry category we consider is $\Vec_{\Z_2^T}$. We will assume that the reader has some familiarity with module categories for complex/unitary fusion categories, see \cite{Thorngren:2019iar}. 
By a (left) module category we mean a finite semi-simple $\R$-linear category $\cM$, together with an $\R$-linear left action 
\be
\rhd: \Vec_{\Z_2^T}\times \cM\to \cM
\ee
that is associative up to chosen coherence. Given two module categories $\cM, \cN$, a module map $\cM\to \cN$ is an $\R$-linear functor that commutes with the $\Vec_{\Z_2^T}$-action up to chosen coherence.  We claim that the 2-category of $(1+1)$d $\Z_2^T$-phases, domain walls between them, and local operators, is $\Mod(\Vec_{\Z_2^T})$ the 2-category of $\Vec_{\Z_2^T}$-modules. There are several ways to compute this 2-category, by which we mean listing all the simple objects and morphism categories between them: 
\begin{itemize}
    \item There is a regular module category $\Vec_{\Z_2^T}$ with $\rhd$ acting in the  same way as the  tensor product of $\Vec_{\Z_2^T}$.
    \item There is a module category on $\Vec$, with action $T\rhd \lid=\lid$, and the associator $T\rhd (T\rhd \lid)\to (T\otimes T)\rhd \lid$ is $+1$. 
    \item There is a module category $\Vec^\eta$, which has underlying category $\Vec$, and action $T\rhd \lid=\lid$, and the associator $T\rhd (T\rhd \lid)\to (T\otimes T)\rhd \lid$ is $-1$.
\end{itemize}
In fact the possible module structures on $\Vec$ are classified by $H^2(\Z_2^T, \C^\times_T)=\Z_2$, and $\eta$ represents the non-trivial cocycle with $\eta(T,T)=-1$. Next we consider the morphism categories. 
\begin{itemize}
    \item An  endomorphism $F: \Vec_{\Z_2^T}\to \Vec_{\Z_2^T}$ is determined by an object in $\Vec_{\Z_2^T}$, namely $F(\lid)$. This is because $\Vec_{\Z_2^T}$ is the regular module category. So $\End_{\Mod(\Vec_{\Z_2^T})}(\Vec_{\Z_2^T})=\Vec_{\Z_2^T}$.
    \item An endomorphism $F: \Vec\to \Vec$ viewed as a functor is determined by the vector space $V=F(\lid)\in \Vec$. The module structure of $F$ gives isomorphism $U: F(T\rhd \lid)\simeq T\rhd F(\lid)$, which equips $V$ with an anti-linear map $U: V\to V$ such that $U^2=1$. Therefore the endomorphism category of $\Vec$ is $\Rep(\Z_2^T)\simeq \Vec_\R$. 
    \item An endomorphism $F: \Vec^\eta\to \Vec^\eta$ is also determined by a vector space $V$ together with an anti-linear involution. Hence the endomorphism category of $\Vec^\eta$ is also $\Rep(\Z_2^T)\simeq \Vec_\R$. 
    \item A morphism $\Vec_{\Z_2^T}\to \Vec$ is the same as an object in $\Vec$. This is because $\Vec_{\Z_2^T}$ is the regular module category. Hence $\Hom_{\Mod(\Vec_{\Z_2^T})}(\Vec_{\Z_2^T},\Vec)\simeq \Vec$. 
    \item Similarly $\Hom_{\Mod(\Vec_{\Z_2^T})}(\Vec_{\Z_2^T},\Vec^\eta)\simeq \Vec$.
    \item A morphism $F: \Vec\to \Vec^\eta$ viewed as a functor is determined by the vector space $V=F(\lid)$. The module structure of $F$ gives $V$ an anti-linear isomorphism $U: V\to V$ such that $U^2=-1$. We see that $V$ is a representation of the algebra $\C[\Z_2^T]^\eta=\bH$. Hence $\Hom_{\Mod(\Vec_{\Z_2^T})}(\Vec,\Vec^\eta)=\Vec_\bH$. 
\end{itemize}
In summary we arrive at the diagrams we showed at the start of the section. 

This concludes our discussion of $\Z_2^T$-symmetric phases in $(1+1)$d. Our analysis demonstrates clearly that real fusion categories are the correct mathematical structure to describe anti-unitary symmetries. As these real categories are less familiar and well-studied, we will first spend one section laying out the theory for real fusion categories, before returning to studying anti-unitary symmetries -- invertible and non-invertible -- as well as gapped phases.

\section{Basics of Real Fusion Categories}
\label{sec_basic_rfc}

One major goal of this paper is to establish a categorical framework for $(1+1)$d gapped phases with \emph{generalized} anti-linear symmetries. A new mathematical language is needed to describe such generalized anti-linear symmetries and phases protected by them, namely the language of real fusion categories~\footnote{For completeness, we recall that a fusion category over $\mathbb{R}$ is a rigid monoidal finite semi-simple $\mathbb{R}$-linear category with an indecomposable monoidal unit.}. The two standard examples $\Rep(\Z_2^T)$ and $\Vec_{\Z_2^T}$ belong to two families of real fusion categories that behave quite differently.  In this section we provide a self-contained introduction to the subject, including a review of basic properties and elementary examples.  The complementary mathematics literature is limited to \cite{sanford2025fusion, plavnik2024TYreal,Kong:2020iek,Kong:2021ups}.

It turns out there are two types of real fusion categories that behave quite differently.  We will define two types of real fusion categories as follows:
\begin{definition}
    A fusion category $\cC$ over $\R$ is called 
    \begin{itemize}
    \item \emph{$\R$-real} if $\End_\cC(\lid)\simeq \R$; 
    \item  \emph{Galois-real} if $\End_\cC(\lid)\simeq \C$. 
    \end{itemize}
\end{definition}
In the following we will discuss in detail the two types of real fusion categories.  We will also discuss the roles of these real fusion categories in describing anti-linear symmetries. Examples will be discussed in the next section.

\subsection{$\R$-Real Fusion Categories}

We now specialize to the case of a real fusion category $\cc$ with 
\be
\End_\cc (\lid) \simeq \R \,,
\ee
i.e. the {\bf  $\R$-real} case. A basic fact is that for a simple object $X \in \cc$, the endomorphism algebra $\End_\cc(X)$ is a finite-dimensional central division algebra over $\R$. By the classification of finite-dimensional division algebras over $\R$, this forces
\be
\End_\cc (X) \;\in\; \{ \R , \C , \bH \} \,.
\ee
Accordingly we will say that a simple object $X$ is {\bf  real}, {\bf  complex}, or {\bf  quaternionic} if $\End_\cc(X) \simeq \R,\C,\bH$ respectively. This trichotomy is the categorical analogue of the three types of irreducible real representations of a finite group.

\medskip

\noindent{\bf Physical interpretation.}  
An $\R$-real fusion category naturally appears as the category of symmetric defects in a system with unbroken time-reversal symmetry (such as $\Rep(\Z_2^T)$ in the $\Z_2^T$-symmetric phases). Given an $\R$-real fusion category $\cc$ of symmetric defects, we can ``forget'' time reversal by extending scalars from $\R$ to $\C$. Categorically this is done by the relative Deligne tensor product
\be
\cc^\C \;:=\; \cc \boxtimes_{\Vec_\R} \Vec_\C \,,
\ee
which is now a fusion category over $\C$. 
Here $\Vec_\R$ is the category of finite dimensional real vector spaces, and $\Vec_\C=\Vec$ is the category of finite dimensional complex vector spaces. Conversely the $\R$-real fusion category $\cc$ can be recovered by remembering the time-reversal action on the complex fusion category $\cc^\C$.  From this perspective the real/complex/quaternion trichotomy originates from three ways time-reversal can act on an object in a complex fusion category:
\begin{itemize}
    \item If a simple object $X\in \cc$ is real, then it comes from a simple object $X^\C\in \cc^\C$ that is invariant under time-reversal and $T^2|_X=1$. In other words it is a time-reversal singlet.
    \item If a simple $X\in \cc$ is complex, then it comes from a direct sum of $X^\C+\overline{X^\C}\in \cc^\C$ and time-reversal exchanges $X^\C\leftrightarrow \overline{X^\C}$. In other words it is a time-reversal doublet. 
    \item If a simple $X\in \cc$ is quaternion, then it comes from a simple object $X^\C\in \cc^\C$ that is invariant under time-reversal and $T^2|_X=-1$. In other words it is a Kramers doublet.
\end{itemize}
An $\R$-real fusion category is rigid, hence all objects have duals. It is useful to introduce an indicator that is defined for a simple $X\in \cc$ as
\be
\chi(X):=\dim_{\R}(\End_\cc(X))\in\{1,2,4\} \,.
\ee
Then the dual object $X^\vee$ satisfies 
\be
X\otimes X^\vee=\chi(X)\lid\oplus \cdots
\label{eq_dual_object}
\ee
The multiplicity of the unit on the RHS is determined by the type of the object $X$.

\subsection{Galois-Real Fusion Categories}
\label{sec_CRFC}

We now turn to the  case of a real fusion category $\cC$ with
\begin{equation}
\End_\cC (\lid) \simeq \C \,,
\end{equation}
i.e. the \textbf{Galois-real} case. Here, the entire real fusion category $\cc$ acquires a complex structure as a plain linear category: this follows because the  scalars $\C\simeq \End_\cc(\lid)$ naturally act on all morphism spaces (say, from the left):

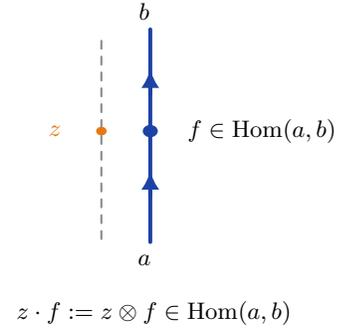
\begin{figure}[H]
\centering
\begin{tikzpicture}[x=1cm,y=0.82cm, line cap=round, line join=round]
    \definecolor{mygray}{RGB}{140,140,140}

    \draw[dashed, draw=mygray, line width=0.8pt, dash pattern=on 3pt off 3pt]
        (-0.55,0.25) -- (-0.55,3.55);

    \fill[myorange] (-0.55,2.00) circle (0.07);
    \node[text=myorange, left=0.42cm] at (-0.55,2.00) {$z$};

    \draw[myblue, line width=1.5pt] (0.1,0.2) -- (0.1,3.65);

    \draw[-{Latex[length=7pt,width=7pt]}, myblue, line width=1.5pt]
        (0.1,0.95) -- (0.1,1.35);

    \draw[-{Latex[length=7pt,width=7pt]}, myblue, line width=1.5pt]
        (0.1,2.6) -- (0.1,3);

    \fill[myblue] (0.1,2.00) circle (0.1);

    \node at (0.02,-0.10) {$a$};
    \node at (0.02,3.95) {$b$};

    \node[anchor=west] at (0.48,2.00) {$f\in \operatorname{Hom}(a,b)$};

    \node at (0.15,-0.95) {$z\cdot f := z\otimes f \in \operatorname{Hom}(a,b)$};

\end{tikzpicture}
\end{figure}

This means $\cc$ is a $\C$-linear (finite semi-simple) category.  In particular each simple $a$ satisfies 
\be
\End_\cc(a)\simeq \C,
\ee
and every object decomposes into a direct sum of finitely many simples. 
However this $\C$-structure does not necessarily make $\cc$ into a complex fusion category, due to the phenomenon called \textbf{Galois-non-triviality}. 

Consider a simple object $a\in \cc$, the scalars from $\End_\cc(\lid)$ can act on $a$ from the left and also from the right. Since $\cc$ is defined as a real fusion category, these two actions only need to agree on the subfield $\R\subset \C$. This leaves us two possibilities: either the left and right action agrees, or they differ by complex conjugation, since these are the only two $\R$-algebra homomorphisms of $\C$.  Diagrammatically, this means if we pass a scalar $z$ across a simple line $a\in \cc$, it could become $z$ or $z^*$, see ~\Cref{pic_anti-linear}.

\begin{figure}
\centering
\begin{tikzpicture}[x=1cm,y=0.82cm, line cap=round, line join=round]
    \draw[myblue, line width=2.1pt] (0,0.2) -- (0,3.65);

    \draw[-{Latex[length=7pt,width=7pt]}, myblue, line width=2.1pt]
        (0,0.95) -- (0,1.35);

    \node at (-0.02,-0.10) {$a$};

    \fill[myorange] (-1.05,2.00) circle (0.07);
    \fill[myorange] ( 1.05,2.00) circle (0.07);

    \draw[myorange, line width=0.8pt, dash pattern=on 1.5pt off 1.5pt]
        (-1.05,2.00) -- (1.05,2.00);

    \node at (-1.38,1.98) {$z$};
    \node[anchor=west] at (1.22,1.98) {$z^{s(a)}$};

\end{tikzpicture}
\caption{A simple line $a$ in a Galois-real fusion category can act on scalars as either identity or complex conjugation, depending on the Galois indicator $s(a)$.\label{pic_anti-linear}}
\end{figure}

Hence we introduce the {\bf Galois indicator} $s(a)\in \{1,T\}$ for every simple object $a$. We say a simple object $a$ is linear/anti-linear, or Galois trivial/nontrivial, when $s(a)=1/T$.  It is precisely the anti-linear objects that spoil the $\C$-linear structure. If there is no anti-linear object, a Galois-real fusion category reduces to an ordinary complex fusion category.

The tensor product preserves the Galois grading: if $a,b,c$ are three simples and there is a nonzero morphism $c\to a\otimes b$, or equivalently $c$ appears in the decomposition of $a\otimes b$ into simples, then $s(c)=s(a)s(b)$. Diagrammatically this is very clear from ~\Cref{fig:chargedjunction}. 

\begin{figure}
\centering
\begin{tikzpicture}[x=1cm,y=0.82cm,line cap=round,line join=round]

    \coordinate (a) at (0.00,3.00);
    \coordinate (b) at (1.20,3.00);
    \coordinate (v) at (0.60,1.86);
    \coordinate (c) at (0.60,0.12);

    \draw[fusion,onstrand=.72] (v) .. controls (0.15,2.05) and (0.00,2.55) .. (a);
    \draw[fusion,onstrand=.72] (v) .. controls (1.05,2.05) and (1.20,2.55) .. (b);
    \draw[fusion,onstrand=.60] (c) -- (v);

    \fill[myblue] (v) circle (0.11);

    \node at (0.00,3.30) {$a$};
    \node at (1.20,3.30) {$b$};
    \node at (0.56,-0.18) {$c$};

    \draw[myorange, line width=0.8pt, dash pattern=on 1.5pt off 1.5pt]
        (0.60,1.86) ellipse[x radius=1.18, y radius=0.78];

    \fill[myorange] (-0.58,1.86) circle (0.07);
    \fill[myorange] ( 1.78,1.86) circle (0.07);

    \node at (-0.86,1.82) {$z$};
    \node[anchor=west] at (1.92,1.82) {$z^{s(a)s(b)} = z^{s(c)}$};
\end{tikzpicture}
\caption{We can move a scalar $z$ from the LHS to  the RHS in two ways, yielding relation $s(a)s(b)=s(c)$, assuming the junction operator is nonzero.\label{fig:chargedjunction}}
\end{figure}

This implies that a Galois-real fusion category $\cc$ has a canonical $\Z_2^T:=\mathrm{Gal}(\C/\R)$-grading
\be
\cc=\cc_1\oplus \cc_T
\ee
such that tensor product of objects preserves the grading. Different orders of fusing simple objects are related as $F$-symbols similar to the case for an ordinary complex fusion category:
\be
\begin{tikzpicture}[x=1cm,y=0.82cm,line cap=round,line join=round]

    \begin{scope}[shift={(0,0)}]
        \coordinate (a) at (0.00,3.00);
        \coordinate (b) at (0.82,3.00);
        \coordinate (c) at (1.78,3.00);
        \coordinate (u) at (0.46,2.02);
        \coordinate (v) at (1.02,1.12);
        \coordinate (d) at (1.02,0.08);

        \draw[fusion,onstrand=.72] (u) .. controls (0.10,2.15) and (0.00,2.58) .. (a);
        \draw[fusion,onstrand=.72] (u) .. controls (0.76,2.15) and (0.82,2.58) .. (b);
        \draw[fusion,onstrand=.42] (v) .. controls (0.78,1.28) and (0.55,1.62) .. (u);
        \draw[fusion,onstrand=.72] (v) .. controls (1.68,1.50) and (1.80,2.52) .. (c);
        \draw[fusion,onstrand=.58] (d) -- (v);

        \fill[myblue] (u) circle (0.085);
        \fill[myblue] (v) circle (0.085);

        \node at (0.00,3.28) {$a$};
        \node at (0.82,3.28) {$b$};
        \node at (1.78,3.28) {$c$};
        \node at (1.02,-0.22) {$d$};
        \node at (0.28,1.40) {$e$};
    \end{scope}

    \node at (3.10,1.95) {$=$};

    \node at (4.35,1.97) {$\left(F^{abc}_{d}\right)^{e}_{f}$};
    \fill[myorange] (4.05,1.48) circle (0.055);

    \begin{scope}[shift={(5.10,0)}]
        \coordinate (a) at (0.00,3.00);
        \coordinate (b) at (0.86,3.00);
        \coordinate (c) at (1.68,3.00);
        \coordinate (u) at (1.22,2.02);
        \coordinate (v) at (0.78,1.08);
        \coordinate (d) at (0.78,0.08);

        \draw[fusion,onstrand=.72] (v) .. controls (0.12,1.50) and (0.00,2.46) .. (a);
        \draw[fusion,onstrand=.72] (u) .. controls (0.94,2.18) and (0.86,2.58) .. (b);
        \draw[fusion,onstrand=.72] (u) .. controls (1.62,2.18) and (1.68,2.58) .. (c);
        \draw[fusion,onstrand=.42] (v) .. controls (1.00,1.30) and (1.18,1.66) .. (u);
        \draw[fusion,onstrand=.58] (d) -- (v);

        \fill[myblue] (u) circle (0.085);
        \fill[myblue] (v) circle (0.085);

        \node at (0.00,3.28) {$a$};
        \node at (0.86,3.28) {$b$};
        \node at (1.68,3.28) {$c$};
        \node at (0.78,-0.22) {$d$};
        \node at (1.45,1.36) {$f$};
    \end{scope}

\end{tikzpicture}
\ee
Here we have suppressed the indices for the junction operators and summing over $f$ is assumed. Notice that since scalars are no longer transparent in a Galois-real fusion category, it matters where the scalars are ``located". We will make the convention that the $F$-symbol is located to the left of the  diagram -- as indicated by the orange dot.  The associativity constraints of a Galois-real fusion category can be expressed as the pentagon equation as for an ordinary fusion category, however the anti-linear objects modify the pentagon equation. Specifically, consider the following step of the pentagon equation: 
\be
\begin{tikzpicture}[x=1cm,y=0.82cm,line cap=round,line join=round]

    \begin{scope}[shift={(0,0)}]
        \coordinate (a) at (0.00,3.55);
        \coordinate (b) at (0.95,3.55);
        \coordinate (c) at (1.80,3.55);
        \coordinate (d) at (2.72,3.55);

        \coordinate (u) at (1.34,2.42); 
        \coordinate (v) at (1.77,1.53); 
        \coordinate (w) at (1.32,0.63); 
        \coordinate (e) at (1.32,0.00);

        \draw[fusion,onstrand=.80] (w) .. controls (0.55,0.98) and (0.03,2.15) .. (a);
        \draw[fusion,onstrand=.78] (u) .. controls (1.00,2.78) and (0.96,3.18) .. (b);
        \draw[fusion,onstrand=.78] (u) .. controls (1.67,2.78) and (1.78,3.18) .. (c);
        \draw[fusion,onstrand=.50] (v) .. controls (1.60,1.82) and (1.46,2.16) .. (u);
        \draw[fusion,onstrand=.77] (v) .. controls (2.34,1.82) and (2.70,2.48) .. (d);
        \draw[fusion,onstrand=.54] (w) .. controls (1.50,0.95) and (1.63,1.22) .. (v);
        \draw[fusion,onstrand=.60] (e) -- (w);

        \fill[myblue] (u) circle (0.085);
        \fill[myblue] (v) circle (0.085);
        \fill[myblue] (w) circle (0.085);

        \node at (0.02,3.87) {$a$};
        \node at (0.95,3.87) {$b$};
        \node at (1.80,3.87) {$c$};
        \node at (2.73,3.87) {$d$};
        \node at (1.32,-0.22) {$e$};

        \node at (1.3,1.86) {$j$};
        \node at (1.7,0.97) {$k$};
    \end{scope}


    \begin{scope}[shift={(5.5,0)}]
        \coordinate (a) at (0.00,3.55);
        \coordinate (b) at (0.95,3.55);
        \coordinate (c) at (1.80,3.55);
        \coordinate (d) at (2.70,3.55);

        \coordinate (u) at (2.24,2.54); 
        \coordinate (v) at (1.76,1.69); 
        \coordinate (w) at (1.28,0.76); 
        \coordinate (e) at (1.28,0.00);

        \draw[fusion,onstrand=.80] (w) .. controls (0.48,1.02) and (0.04,2.20) .. (a);
        \draw[fusion,onstrand=.76] (v) .. controls (1.02,2.10) and (0.96,2.92) .. (b);
        \draw[fusion,onstrand=.78] (u) .. controls (1.87,2.92) and (1.82,3.18) .. (c);
        \draw[fusion,onstrand=.78] (u) .. controls (2.54,2.92) and (2.69,3.18) .. (d);
        \draw[fusion,onstrand=.50] (v) .. controls (1.95,2.06) and (2.10,2.30) .. (u);
        \draw[fusion,onstrand=.54] (w) .. controls (1.47,1.02) and (1.60,1.34) .. (v);
        \draw[fusion,onstrand=.60] (e) -- (w);

        \fill[myblue] (u) circle (0.085);
        \fill[myblue] (v) circle (0.085);
        \fill[myblue] (w) circle (0.085);

        \node at (0.02,3.87) {$a$};
        \node at (0.95,3.87) {$b$};
        \node at (1.80,3.87) {$c$};
        \node at (2.71,3.87) {$d$};
        \node at (1.28,-0.22) {$e$};

        \node at (2.2,1.94) {$l$};
        \node at (1.76,1.08) {$k$};
    \end{scope}
 \fill[myorange] (5.5,0.76)  circle (0.055);
    \fill[myorange] (6.6,1.15) circle (0.055);
    \draw[orangepath=.56] (6.6,1.15) .. controls (6,1.2) .. (5.5,0.76);
    
      \node at (4.02,1.92) {$=$};

    \node[above] at (6.6,1.15){$\left(F^{bcd}_{k}\right)^{j}_{l}$};
    \node[below] at (5.5,0.76) {$\left(\left(F^{bcd}_{k}\right)^{j}_{l}\right)^{s(a)}$};

\end{tikzpicture}
\ee
The $F$-symbol comes from the associator $(b\otimes c)\otimes d\to b\otimes(c\otimes d)$, and by our convention it is located between the $a$-line and the $b$-line. Hence if we move it to the left of the entire diagram, the $F$-symbol acquires a complex conjugation if $a$ is an anti-linear object. This is the only move in the five moves of the pentagon equation that differs from that of an ordinary complex fusion category. In fact the modified pentagon equation appeared in multiple references~\cite{Bhardwaj:2016dtk, Barkeshli_2019}, but the correct mathematical structure, i.e., Galois-real fusion category,  was not identified.

Finally, a Galois-real fusion category is rigid, hence for every object $a$, there is a dual object $a^\vee$, with canonical morphisms $a\otimes a^\vee\to \lid,~\lid\to a\otimes a^\vee$ (below) satisfying the standard zig-zag identity:
\be
\begin{tikzpicture}[x=0.95cm,y=0.95cm,line cap=round,line join=round]
    \definecolor{mygray}{RGB}{155,155,155}

    \tikzset{
        fusion/.style={draw=myblue,line width=1.35pt},
        unitline/.style={draw=mygray,line width=1.15pt,dash pattern=on 3pt off 2.5pt},
        onstrand/.style={
            postaction={
                decorate,
                decoration={
                    markings,
                    mark=at position #1 with {\arrow{Latex[length=5pt,width=5pt]}}
                }
            }
        }
    }

    \begin{scope}[shift={(0,0)}]
        \coordinate (L) at (0.00,0.10);
        \coordinate (M) at (0.80,1.58);
        \coordinate (R) at (1.60,0.12);
        \coordinate (U) at (0.80,3.10);

        \draw[unitline] (M) -- (U);

        \draw[fusion,onstrand=.28]
            (L) .. controls (0.00,0.95) and (0.22,1.55) .. (M);
        \draw[fusion,onstrand=.72]
            (M) .. controls (1.38,1.55) and (1.60,0.95) .. (R);

        \fill[mygray] (M) circle (0.065);

        \node at (0.80,3.35) {$1$};
        \node at (-0.10,-0.22) {$a$};
        \node at (1.72,-0.22) {$a^{\vee}$};
    \end{scope}

    \begin{scope}[shift={(4.55,0)}]
        \coordinate (L) at (0.00,3.00);
        \coordinate (M) at (0.80,1.62);
        \coordinate (R) at (1.60,3.00);
        \coordinate (B) at (0.80,0.05);

        \draw[unitline] (B) -- (M);

        \draw[fusion,onstrand=.72]
            (M) .. controls (0.22,1.58) and (0.00,2.20) .. (L);
        \draw[fusion,onstrand=.28]
            (R) .. controls (1.60,2.20) and (1.38,1.58) .. (M);

        \fill[mygray] (M) circle (0.065);

        \node at (-0.08,3.32) {$a$};
        \node at (1.68,3.32) {$a^{\vee}$};
        \node at (0.80,-0.22) {$1$};
    \end{scope}
\end{tikzpicture}
\ee

\noindent{\bf Physical Interpretation.}
A Galois-real fusion category can arise in two situations. First it can be the category of symmetric defects in a time-reversal broken system, as we have seen in the $\Z_2^T$-SSB example. 
The reason that time-reversal breaking leads to Galois-real fusion category can be understood as follows. Since time-reversal is broken we have two conjugate vacua, and local operators take the form $(z\,\id_1, w\,\id_2)$ where $\id_{1/2}$ are the identity operators in the two vacua. Time-reversal acts as $(z,w)\mapsto (w^*, z^*)$, hence the time-reversal symmetric local operators are $(z, z^*)$ which still form a copy of $\C$. In this case the component $\cc_1$ consists of defects within each vacuum, and $\cc_T$ consists of domain walls between the two vacua.  Compare with the time-reversal unbroken case: there  is only one vacuum and the symmetric local operators are $z\, \id$ with $z\in \R$. This is why the category of symmetric defects in a time-reversal symmetric phase is $\R$-real.

Second, a {\bf Galois-real} fusion category can also arise as the abstract symmetry category for an anti-linear symmetry. This is similar to how a fusion category can arise as the abstract symmetry category for a linear symmetry. In this case the $\Z_2^T$-grading $\cc=\cc_1\oplus \cc_T$ has a very natural meaning: the component $\cc_1$ consists of all linear symmetry generators and $\cc_T$ consists of all anti-linear symmetry generators. If we Wick-rotate the relation~\Cref{pic_anti-linear}, the worldline of $a$ becomes a symmetry action on a time-slice, and the relation $za=az^{s(a)}$ is precisely the statement that $a$ is an anti-linear symmetry generator when 
$s= T$.

\smallskip    
\noindent{\bf Remark: $\R$-real categories cannot be symmetry categories. } We discussed above that Galois-real fusion categories are the natural mathematical object for describing anti-linear symmetries. The reader may wonder whether an $\R$-real fusion category can be regarded as a symmetry category for a physical system. 
We reject this interpretation as complex numbers should always be a symmetry of any quantum system. Allowing an $\R$-real fusion category to be a symmetry category would imply that we are actually doing quantum physics over $\R$, which is non-physical.

\subsection{Module Categories over Real Fusion Categories}\label{sec_module_cat}

The theory of module categories over real fusion categories will be crucial to the categorical formulation and classification of gapped phases. The basic structure of the theory of module categories does not depend on the ground field. In fact, module category can be defined for any monoidal category, and we direct the reader to~\cite{etingof2015tensor} for the definition. Roughly speaking a module category $\cM$ over a monoidal category $\cC$ consists of the data of a $\cC$-action:
\be
a\rhd m, \qquad ~a\in \cC, ~m\in \cM
\ee
and an associator 
\be
\alpha_{a,b,m}: \ a\rhd b\rhd m\to a\otimes b\rhd m,\quad a,b\in \cC,~m\in \cM
\ee
that satisfies a pentagon-like identity.
In this paper by a module category we always mean a finite semi-simple one, meaning that the underlying category of the module category is finite semi-simple. The key structural results that we will use about module categories are the following, which was proven in~\cite{douglas2018dualizabletensorcategories} for a general ground field. 
\begin{theorem}\label{thm_module}
    Let $\cC$ be a fusion category over $\R$ and $\cM$ a right $\cC$-module category, then there exists a separable algebra $A\in \cC$ such that $\cM=\Mod_\cC(A)$ and $\End_\cC(\cM)=\Bimod_\cC(A)$.  The algebra $A$ is determined by the module category $\cM$ up to Morita equivalence. 
\end{theorem}
{Here $\Mod_\cC(A)$ left $A$-modules over $\cC$, whereas $\Bimod_\cC(A)$ are bimodules.}
These structural results allow us to classify module categories by classifying separable algebras in the real fusion category up to Morita equivalence.
Following the same terminology as for module categories over $\C$, we will say that the real fusion category $\Bimod_\cC(A)$ is obtained by \textbf{condensing the algebra $A$ in $\cC$}. 
Module categories are used to define Morita (i.e. in this case, gauge) equivalence as follows. 
\begin{definition}
    Let $\cC, \cD$ be real fusion categories, they are said to be Morita equivalent if there exists a right $\cC$-module category $\cM$ such that $\cD\simeq \End_{\cC}(\cM)$. Or equivalently, if there exists an algebra $A\in \cC$ such that $\Bimod_\cC(A)\simeq \cD$.
\end{definition}
One thing we should emphasize is that Morita equivalence {\bf does not preserve} the types of real fusion categories. This means a Galois-real fusion category can be Morita equivalent to an $\R$-real one or a Galois-real one and similarly for an $\R$-real fusion category. 
Starting with a Galois-real fusion category $\cC$, one naturally asks what kind of algebras in $\cC$ give rise to $\R$-real fusion categories upon condensation and what kind of algebras in $\cC$ give rise to Galois-real ones. This question is relevant for the theory of gapped phases with anti-linear symmetry. We provide the answer in the form of the following theorem, which appears to be a new mathematical result.
\begin{theorem}\label{th_gauging}
    Let $\cC=\cC_1\oplus \cC_T$ be a Galois-real fusion category and $A\in \cC$ is an indecomposable separable algebra, then the Morita dual $\Bimod_\cC(A)$ is Galois-real if $A\in \cC_1$ and $\R$-real if $A\notin \cC_1$.
    \proof Since $A$ is indecomposable, we know the dual category $\Bimod_\cC(A)$ is a real fusion category. It then remains to calculate $\Omega(\Bimod_\cC(A))$ which is the same as $\End^\cC_{A-A}(A)$. The most general element in it takes the form $z\id_A$, this is an $A-A$ bimodule morphism precisely when it commutes with all the simple summands of $A$. If $A\in \cC_1$, there is no constraint on $z$ and $\Omega(\Bimod_\cC(A))\simeq \C$, if $A\notin \cC_1$ then $z$ has to be real and $\Omega(\Bimod_\cC(A))\simeq \R$. \qed 
\end{theorem}
We will now provide some examples of various types of real fusion categories. Subsequently we discuss and classify gapped phases for Galois-real fusion category symmetries, using module categories. 
\bigskip

\section{Examples of Real Fusion Categories}
\label{sec_example_rfc}

The goal of this section is to illustrate the previous general discussions with concrete examples. The examples will also be crucial for our later discussion on gapped phases with anti-linear symmetries. 

\subsection{The $\Vec_{G^T}^\omega$ Categories}

We start with categories associated with a group-like anti-linear symmetry $G^T$. Similar to the unitary internal symmetry case, there are $\Rep(G^T)$ and $\Vec_{G^T}$ categories associated with $G^T$. Here by $G^T$ we mean a group $G^T$ {together with a group homomorphism} 
\be 
s: \quad G^T\to \Z_2^T=\{1, T\}\,. 
\ee
We call an element $g\in G^T$ linear(anti-linear) if $s(g)=1(T)$. We will always assume the map $s$ is clear from the context and refer to $G^T$ as an anti-unitary symmetry.

\subsubsection{Definition of $\Vec_{G^T}^\omega$}

We first consider the category $\Vec_{G^T}$ and its twisted version $\Vec_{G^T}^\omega$ where $\omega$ is any class in the twisted cohomology group $H^3(G^T, U(1)_T)$. These are defined as follows: 
\begin{itemize}
    \item The simple objects are $\delta_g,~g\in G^T$, with fusion rule $\delta_g\otimes \delta_h=\delta_{gh}$.
    \item We declare $\End(\delta_g)=\C$, and $\delta_g$ satisfies
    \begin{equation}
        z\otimes \delta_g=\delta_g\otimes z^{s(g)},~z\in \C.
    \end{equation}
    This condition breaks $\C$-linearity whenever $s$ is not trivial.
\end{itemize}
Denote the associator by 
\begin{equation}
    \omega(g,h,k):~(\delta_g\otimes \delta_h)\otimes \delta_k\to \delta_g\otimes(\delta_h\otimes\delta_k)
\end{equation}
The pentagon equation gives the commutative diagram below:
\begin{widetext}
\begin{equation}
    \begin{tikzcd}
\bigl(((\delta_g\otimes \delta_h)\otimes \delta_k)\otimes \delta_l\bigr)
  \arrow[rr, "{\omega(gh,k,l)\mathrm{id}_{\delta_{ghkl}}}"]
  \arrow[d, "{\omega(g,h,k)\mathrm{id}_{\delta_{ghkl}}}"']
& &
\bigl((\delta_g\otimes \delta_h)\otimes (\delta_k\otimes \delta_l)\bigr)
  \arrow[d, "{\omega(g,h,kl)\mathrm{id}_{\delta_{ghkl}}}"]
\\
\bigl((\delta_g\otimes (\delta_h\otimes \delta_k))\otimes \delta_l\bigr)
  \arrow[dr, "{\omega(g,hk,l)\mathrm{id}_{\delta_{ghkl}}}"]
& &
\delta_g\otimes \bigl(\delta_h\otimes (\delta_k\otimes \delta_l)\bigr)\arrow[dl,"{\textcolor{red}{\mathrm{id}_{\delta_g}\otimes \omega(h,k,l)\mathrm{id}_{\delta_{hkl}}=\omega(h,k,l)^{s(g)}\mathrm{id}_{\delta_{ghkl}}}}"]\\
 & \bigl(\delta_g\otimes ((\delta_h\otimes \delta_k)\otimes \delta_l)\bigr)
  &
\end{tikzcd}
\end{equation}
    
\end{widetext}
where we have marked the key difference from the unitary $G$ case in red. Collecting coefficients we see that the pentagon equation demands $\omega$ to be a twisted cocycle in $H^3(G^T, U(1)_T)$. 
The following is a direct generalization of Theorem ~\Cref{th:VecZ2T}:
\begin{theorem}\label{th_VecGT}
    The category of symmetry defects in a fully $G^T$-SSB phase with $G^T$-anomaly $\omega \in H^3(G^T,U(1)_T)$ is $\Vec_{G^T}^\omega$.\qed 
\end{theorem}
Just like for a unitary $G$-symmetry, we will take $\Vec_{G^T}^\omega$ to mean the symmetry category for a group-like symmetry  $G^T$  with 't Hooft anomaly $\omega \in H^3(G^T,U(1)_T)$.

\subsubsection{Module Categories over $\Vec_{G^T}$}
\label{sec_module_GT}

Here we classify and describe module categories over $\Vec_{G^T}$. By~\Cref{thm_module}  this comes down to classifying (Morita classes of) indecomposable separable algebras in $\Vec_{G^T}$. The classification is similar to that for a $\Vec_G$ category, i.e. in terms of subgroups and cocycles. There are algebras $\C[K^T]^\omega\in \Vec_{G^T}$ where $K^T<G^T$ is a subgroup of $G^T$ with the inherited anti-linear structure, and $\omega \in H^2(K^T,\C^\times)$. The algebra $\C[K^T]^\omega$ is spanned as a $\C$-linear space by basis $\{\delta_g|~g\in K^T\}$, with product rule
\be
\delta_g\cdot \delta_h=\omega(g,h)\delta_{gh},  \quad ~z\delta_g=\delta_g z^{s(g)}.
\ee
The associated module category, denoted as $\cM(K^T,\omega)$, has underlying category $\Vec_{G^T/K^T}$ where $G^T/K^T$ is the coset, and the module structure is determined by $\omega$. This structure is completely analogous to that of module categories over $\Vec_G$ for a unitary $G$, hence we choose not to write down all the module category structures. We only mention that the action of $\delta_{k}\in \Vec_{G^T},~k\in K^T$ on the simple object $\delta_{K^T} \in \Vec_{G^T/K^T}$ labeled by trivial coset $K^T$ is 
\be
 \delta_{k}\rhd \delta_{K^T} =\delta_{K^T},~\delta_{k_1}\rhd \delta_{k_2}\rhd \delta_{K^T} \xrightarrow{\omega(k_1,k_2)} \delta_{k_1}\otimes \delta_{k_2}\rhd \delta_{K^T}.
\ee
The full module category structure is  determined by the above data, and the coherence condition of module category demands that $\omega$ is a twisted 2-cocycle of $K^T$.

We note that the algebra $\C[K^T]^\omega$ lies completely in the linear sector $\Vec_{G_0}$ precisely when $K^T$ is contained in the linear subgroup 
\be 
G_0=\ker(s)\,.
\ee
Therefore the Morita dual $\Bimod_{\Vec_{G^T}}(\C[K^T]^\omega)=\End_{\Vec_{G^T}}(\cM(K^T,\omega))$ is Galois-real precisely when $K^T\le G_0$ and $\R$-real otherwise.

\subsection{The $\Rep(G^T)$ Categories}
Next we consider the $\Rep(G^T)$ categories. By a representation of $G^T$ we mean a complex vector space $V$ together with actions $\rho_g: V\to V$ such that 
\begin{enumerate}
\item[(a)] $\rho_g$ is $\C$-linear if $s(g)=1$ and $\C$-anti-linear if $s(g)=T$. 
\item[(b)] $\rho_g\circ \rho_h=\rho_{gh}$. 
\end{enumerate}
The (finite dimensional) representations of an anti-linear $G^T$-symmetry naturally form a real fusion category $\Rep(G^T)$, which is always {\bf $\R$-real} as long as $G^T$ contains anti-linear elements. It is obtained from $\Vec_{G^T}$ by gauging the full group $G^T$, which is not contained in the $\ker(s)$ and thus yields indeed an $\R$-real category. In view of later application, we should note that  this means we should not be including this as  a symmetry category. This is quite in contrast to the unitary case, where this is a standard symmetry category.

The following is a direct generalization of~\Cref{th:RepZ2T} which we hope is by now clear without proof. 
\begin{theorem}\label{th_RepGT}
    The category of symmetric defects in a $(1+1)$d $G^T$-SPT is $\Rep(G^T)$.
\end{theorem}

Recall that for an ordinary $\Rep(G)$ category there is a useful dimension formula
\be
\sum_{\bm{R}} d_{\bm{R}}^2=|G| \,,
\ee
where $d_{\bm{R}}$ is the dimension of the representation $R$, and the sum goes over irreducible representations. For $\Rep(G^T)$ there is an analogous formula that we will make use of later~\cite{Fulton:2004uyc}. For an irreducible representation $\bm{R}$ of $G^T$, we define an indicator $\chi(\bm{R})=1/2/4$ if $\bm{R}$ is real/complex/quaternionic. Then we have 
\be\label{eq_dim_formula}
2\sum_{\bm{R}} \frac{d_{\bm{R}}^2}{\chi(\bm{R})}=|G^T|\,,
\ee
where $d_{\bm{R}}$ still stands for the \emph{complex} dimension of the representation $\bm{R}$. Notice if $G^T$ is purely linear, then $\chi({\bm{R}})=2$ for any ${\bm{R}}$ and we recover the usual formula.

Next we study various examples of $\Rep(G^T)$-categories, and identify some novel features absent in the $\C$-linear world. 
\begin{example}\label{eg_Z4T}
    Let $\Z_4^T$ be the group $\Z_4$ and the generator $T\in \Z_4$ is anti-linear. 
Thus we have 
\be
\ba
s: \quad \Z_4^T &\ \to\  \Z_2^T \cr 
T^i & \ \mapsto \ T^{i \,\text{mod} \,2}  \,.
\ea
\ee
    Then $U:=T^2$ is linear. The linear subgroup is $\ker (s)= \Z_2=\{1, U\}$. We will start by arguing that there is only one non-trivial irreducible representation using the dimension formula~\eqref{eq_dim_formula}. The trivial representation contributes $2\frac{d_\lid}{\chi(\lid)}=2$, hence there is at most one non-trivial irreducible representation $R$ such that $\frac{d_R^2}{\chi(R)}=1$. This leaves us two possibilities: either $d_R=1$ and $R$ is of real type, or $d_R=2$ and $R$ is of quaternion type. It is not difficult to see that any one dimensional representation of $\Z_4^T$ is equivalent to the trivial one, therefore the first possibility is ruled out. Hence the only non-trivial representation is of quaternion type. For completeness we provide an explicit form of the quaternionic representation. Write elements of $\C^2$ as $(z,w),~z,w\in \C$, define $T(z,w)=(-\overline{w}, \overline{z})$, then $T^2=-1$ and this realizes the non-trivial quaternionic representation of $\Z_4^T$. We conclude that there is only one non-trivial simple object in $\Rep(\Z_4^T)$, which we denote by $Q$. According to the formula~\eqref{eq_dual_object}, the tensor product $Q\otimes Q$ contains four copies of $\lid$. Dimension counting shows it contains nothing more. Hence we have the following fusion rule.
    \begin{align}
        Q\otimes Q\simeq \lid\oplus \lid\oplus \lid \oplus \lid=4\cdot \lid.
    \end{align}
\end{example}
 An interesting feature about $\Rep(\Z_4^T)$ is that although $\Z_4^T$ is an abelian group, its representation category $\Rep(\Z_4^T)$ contains non-invertible objects. This is to be contrasted with $\Rep(G)$ categories, which contain non-invertible objects if and only if $G$ is non-abelian. 

\subsubsection{Examples: $\Rep(A\times \Z_2^T)$}\label{sec_AxZ2}

    Let $A$ be an abelian group and $A\times \Z_2^T$ the direct product with $\Z_2^T$. We now compute the structure of 
    \be 
    \Rep(A\times \Z_2^T)\,.
    \ee 
    Let $R=(V,\rho)\in \Rep(A\times \Z_2^T)$, since $T\in \Z_2^T$ commutes with elements of $A$, we can choose a basis for $V$ so that $V\simeq \C^n$ and $T$ acts as complex conjugation: $T: \{z_i\}\mapsto \{\overline{{z_i}}\}$. This means if we restrict to the subgroup $A$, $(V,\rho|_A)$ is a representation of $V$ that is invariant under a complex conjugation on $V$. Conversely, suppose $R=(V,\rho_A)$ is a representation of $A$ such that $\rho_A$ commutes with a complex conjugation on $V$, then $R$ can be promoted to a representation of $A\times \Z_2^T$ by letting $T\in \Z_2^T$ act as complex conjugation. This argument shows that representations of $A\times \Z_2^T$ can be obtained by taking combinations of irreducible representations of $A$ that are invariant under complex conjugation(in some basis). Since $A$ is abelian, irreducible representations are valued in $\widehat{A}:=\Hom(A,U(1))$. Let $\gamma\in \widehat{A}$, then complex conjugation acts by $\gamma\mapsto \gamma^*=\gamma^{-1}$. Therefore a character $\gamma$ is invariant under complex conjugation if $\gamma^2=1$. On the other hand if $\gamma^2\neq 1$, then we may form the direct sum $\gamma\oplus \gamma^{-1}$, which is now invariant under complex conjugation. We conclude that there are two types of simple objects in $\Rep(A\times \Z_2^T)$: 
    \begin{itemize}
        \item $F_\gamma:=\gamma\in \widehat{A}$ if $\gamma^2=1$. In this case $F_\gamma$ is a real object in $\Rep(A\times \Z_2^T)$.
        \item $E_\gamma:=\gamma\oplus \gamma^{-1}$ if $\gamma^2\neq 1$.  In this case $E_\gamma$ is a complex object in $\Rep(A\times \Z_2^T)$.
    \end{itemize}
    It is direct to verify that the dimension formula~\eqref{eq_dim_formula} is satisfied: an object $F_\gamma$ contributes $2\frac{d_R^2}{\chi(R)}=2$, an object $E_\gamma$ contributes $2\frac{d_R^2}{\chi(R)}=4$. But $\gamma$ and $\gamma^{-1}$ give the same $E_\gamma$ object. Hence the total contribution is just $\sum_{\gamma \in \widehat{A}}2=2|A|=|A\times \Z_2^T|$. As for the fusion rule, the following are obvious
    \be\label{eq_fusion_rule_AxZ2}
    F_\gamma\otimes F_{\gamma'}=F_{\gamma\gamma'},~ F_\gamma \otimes E_{\gamma'}=E_{\gamma\gamma'}
    \ee
    Fusion rules for $E_\gamma\otimes E_{\gamma'}$ can be obtained using fusion rules of $\widehat{A}$ and is generically non-invertible. We do not provide a general formula here since it is not very illuminating. 

\begin{remark}
    Let $G_0:=\ker s< G^T$ be the subgroup of linear elements. It is always possible to describe $\Rep(G^T)$ in terms of $\Rep(G_0)$ together with a time-reversal action on $\Rep(G_0)$. By a $\Z_2^T$-action on a fusion category $\cc$ we mean an \emph{anti-linear} tensor equivalence $\cF_T: \cc\to \cc$, together with equivalence of tensor functors $\cF_T\circ \cF_T\simeq \id_\cc$. This philosophy is known as Galois descent in mathematics. For $G^T=A\times\Z_2^T$, we can say that $\Rep(A\times \Z_2^T)$ is the same as $\Rep(A)\simeq\widehat{A}$ together with a $\Z_2^T$-action on $\widehat{A}$. Here the action is simple to describe: $T: \gamma\mapsto \gamma^{-1}$. To recover $\Rep(A\times \Z_2^T)$, we take the (homotopy)\emph{fixed points} of the $\Z_2^T$-action. This is exactly what we did earlier: we took combinations of characters of $A$ that are invariant under the action $T: \gamma\mapsto \gamma^{-1}$.
\end{remark}

\begin{example}
    Let us consider a concrete example. Setting $A=\Z_4$, denote the generator of $\widehat{\Z_4}$ as $\gamma$. Then under complex conjugation $\gamma\mapsto \gamma^3$ and $\gamma^2$ is invariant. Therefore the simples in $\Rep(\Z_4\times \Z_2^T)$ are $\lid, E_\gamma, F_{\gamma^2}$, where $E_\gamma$ is complex, $\lid$ and $F_{\gamma^2}$ are real. We have fusion rules
    \begin{align}
         E_\gamma\otimes E_\gamma &=2\lid \oplus 2F\nonumber\\
         E_\gamma\otimes F_{\gamma^2}\simeq F_{\gamma^2}\otimes E_\gamma&=E_\gamma\nonumber\\
    F_\gamma\otimes F_\gamma&=\lid.
    \end{align}
    We note that the $\Rep(\Z_4\times \Z_2^T)$ again contains non-invertible simple object even though $\Z_4\times \Z_2^T$ is abelian. We also note that the fusion rule is very similar to that of a Tambara-Yamagami category $\TY(\Z_2)$. In fact Tambara-Yamagami categories over $\R$ have been studied in~\cite{plavnik2024TYreal}, and $\Rep(\Z_4\times \Z_2^T)$ is exactly a real Tambara-Yamagami category $\TY_\R(\Z_2)$.
\end{example}

\subsubsection{Examples: $\Rep(A\rtimes \Z_2^T)$}\label{eg_A_rtimes_Z2T}

    Let $A$ be an abelian group, denote by $A\rtimes \Z_2^T$ the semi-direct product with $T\in \Z_2^T$ acting on $A$ as inversion: $TaT=a^{-1},~\forall a\in A$. We now completely determine the structure of the real fusion category 
    \be 
    \Rep(A\rtimes \Z_2^T) \,.
    \ee 
    We claim that simple objects of $\Rep(A\rtimes \Z_2^T)$ are exactly the same as those of $\Rep(A)$. Let $\gamma\in \Rep(A)$ be a character of $A$, extend $\gamma$ to a representation of $A\rtimes \Z_2^T$ by letting $T$ act on $\C$ as complex conjugation. One checks $\gamma(TaT)=\gamma(a^{-1})=\gamma(a)^{-1}$ which is indeed $\gamma(T)\gamma(a)\gamma(T)=\gamma(a)^*=\gamma(a)^{-1}$.  Therefore any $\gamma\in \widehat{A}$  is automatically a representation of $A\rtimes \Z_2^T$. It is easy to see that these simple objects are of real type. In the dimension formula~\eqref{eq_dim_formula} these simple objects contribute $\sum_{\gamma \in \widehat{A}}2 \frac{1^2}{1}=|A\rtimes \Z_2^T|$. Therefore these are all the irreducible representations of $A\rtimes \Z_2^T$. Hence the simples of $\Rep(A\rtimes \Z_2^T)$ are labeled by $\gamma \in \widehat{A}$, with fusion rule identical to that of $\widehat{A}$. In particular, all simples of $\Rep(A\rtimes \Z_2^T)$ are {\bf invertible}. This also implies we have equivalence of real fusion categories $\Rep(A\rtimes \Z_2^T)=\Vec_{\R, \widehat{A}}$, where the RHS is the category of $\widehat{A}$-graded \emph{real} vector spaces.
    
\begin{remark}\label{rmk_Z2T_action_semidirect}
    We discussed earlier how $\Rep(G^T)$ can be understood as $\Rep(G_0)$ equipped with a $\Z_2^T$-action. In the case of $\Rep(A\rtimes \Z_2^T)$, the $\Z_2^T$-action on $\Rep(A)=\widehat{A}$ is very simple: On objects $T: \gamma\mapsto \gamma$. Since $\Rep(A)$ has trivial associators, this object-level action can be promoted to an $\Z_2^T$-action on the fusion category $\Rep(A)$. Since the action is trivial on objects, the fixed points are nothing but the objects in $\Rep(A)$. This explains why the fusion rule of $\Rep(A\rtimes \Z_2^T)$ is identical to that of $\Rep(A)$.
\end{remark}
\begin{example}
    Let us consider a concrete example, let 
    \be 
    D_4^{T, I}=\Z_4\rtimes \Z_2^T \,.
    \ee 
    As a plain group $D_4^{T, I}$ is the same as the dihedral group $D_4$. The extra superscript $I$ is meant to label its anti-linear structure, since there can be multiple homomorphisms from $D_4$ to $\Z_2^T$. Here the anti-linear structure is clear from the form $\Z_4\rtimes\Z_2^T$. Then by the general discussion above simples in $\Rep(D_4^{T,I})$ are labeled as $\lid,F, F^2,F^3$, with abelian fusion rule 
    \be F^i\otimes F^j=F^{i+j\,\text{mod}\, 4}\,.\ee
This is in contrast to the unitary case where this is fusion category, with non-invertible fusion. 
\end{example}

\begin{example}
We now consider the group $D_4$ with another anti-linear structure. Let 
\be 
D_4^{T,II}:=(\Z_2\times \Z_2)\rtimes \Z_2^T
\ee 
with $T$ acting on $\Z_2\times \Z_2$ as exchanging the two $\Z_2$ summands: $T(1,0)T=(0,1)$.     As a plain group $D_4^{T,II}$ is the same as the dihedral group $D_4$, but its anti-linear structure is different from $D_4^{T,I}$ introduced above. This results in $\Rep(D_4^{T,II)}$ behaving differently.  As discussed before we can understand $\Rep(D_4^{T, II})$ from $\Rep(\Z_2\times \Z_2)$ and a $\Z_2^T$-action. Denote by $\gamma_1, \gamma_2$ the generators of $\Rep(\Z_2\times \Z_2)\simeq \widehat{\Z_2}\times \widehat{\Z_2}$. Then time-reversal acts as $\gamma_1\leftrightarrow \gamma_2$. By Galois descend, an object of $\Rep(D_4^{T,II})$ is an object of $\Rep(\Z_2\times \Z_2)$ that is invariant under this time-reversal action. We have two obvious choices. One is $F:=\gamma_1\gamma_2$ which is a singlet hence a real object in $\Rep(D_4^{T,II})$, the other is $E:=\gamma_1\oplus \gamma_2$ which is a doublet hence a complex object in $\Rep(D_4^{T,II})$. These two simples together with the unit contribute $2(\frac{1^2}{1}+\frac{1^2}{1}+\frac{2^2}{2})=8$ in the dimensional formula~\eqref{eq_dim_formula}. Hence they are all the simples.  The fusion rule can be worked out using the fusion rule of $\Rep(\Z_2\times \Z_2)$.
\begin{align}
    &E\otimes E=2\oplus 2 F\nonumber\\
     &E\otimes F=F\otimes E=E\nonumber\\
     &F\otimes F=\lid\label{eq_fusion_rule_D4II}
\end{align}
Notice that this fusion rule is identical  to that of $\Rep(\Z_4\times \Z_2^T)$. However the real fusion categories $\Rep(\Z_4\times \Z_2^T)$,~$\Rep(D_4^{T,II})$ can not be equivalent as they complexify to different complex fusion categories. The difference lies in the associators. The fusion rule~\eqref{eq_fusion_rule_D4II} implies that the real fusion category $\Rep(D_4^{T, II})$ is also a real Tambara-Yamagami category $\TY_\R(\Z_2)$. One should not be surprised that there are multiple real Tambara-Yamagami categories with the same fusion rule, as this phenomena is also common for complex Tambara-Yamagami categories. According to the classification in~\cite{plavnik2024TYreal}, there are in total four $\R$-real fusion categories with the fusion rule~\eqref{eq_fusion_rule_D4II}, $\Rep(\Z_4\times \Z_2^T)$ are $\Rep(D_4^{T,II})$ two of them. The full associator data can be found in~\cite{plavnik2024TYreal}.
\end{example}


\subsubsection{Examples: $\Rep(G\times \Z_2^T)$}

    We discuss more the structure of real fusion categories of the form 
    \be 
    \Rep(G\times \Z_2^T)\,.
    \ee 
    According to Galois descend, $\Rep(G\times \Z_2^T)$ can be understood as $\Rep(G)$ equipped with a $\Z_2^T$-action. Here the action is straightfoward: $T: R\mapsto R^*$ where $R^*$ is the complex conjugate of $R$. If $R$ is not invariant under $T$, meaning that $R^*$ is not equivalent to $R$, then the invariant combination is $R\oplus R^*$, which descends to a complex simple object in $\Rep(G\times \Z_2^T)$. On the other hand if $R\simeq R^*$, then $R$ descends to either a real object or a quaternion object in $\Rep(G\times \Z_2^T)$. Which case is true can be determined by the Frobenius–Schur indicator, defined as 
\begin{align}
    \nu(R):=\frac{1}{|G|}\sum_{g\in G}\chi(g^2)\in\{-1,0,1\}
\end{align}
$\chi(g):=\Tr(\rho(g))$ is the character of the representation $(\rho, V)$. A standard result of representation theory is that the fate of a simple $R\in \Rep(G)$ in $\Rep(G\times \Z_2^T)$ is completely determined by its Frobenius–Schur indicator. 
\begin{itemize}
    \item If $\nu(R)=1$, then $R\simeq R^*$ and $R$ is well-defined over $\R$. This means we can choose a basis such that $R(g)$ are real matrices. Then $R$ becomes a real object in $\Rep(G\times \Z_2^T)$.
    \item If $\nu(R)=0$, then $R\neq R^*$. The sum $R\oplus R^*$ is a complex object in $\Rep(G\times \Z_2^T)$.
    \item If $\nu(R)=-1$, then $R\simeq R^*$ but $R$ is not well-defined over $\R$. This means there is no basis in which $R$ is written as real matrices. In this case although $R\simeq R^*$, to make a $\Z_2^T$-invariant object we still need to make a direct sum $R\oplus R^*$, which is a quaternion simple in $\Rep(G\times \Z_2^T)$.
\end{itemize}
Let us check the validity of the dimension formula~\eqref{eq_dim_formula}. A simple $R\in \Rep(G)$ with $\nu(R)=1$ contributes $2\frac{d_R^2}{1}=2d_R^2$; a simple $R\in \Rep(G)$ with $\nu(R)=0$ gives rise to the complex simple $R\oplus R^*\in \Rep(G\times \Z_2^T)$, which contributes $2\frac{(2d_R)^2}{2}=4d_R^2$; a simple $R\in \Rep(G)$ with $\nu(R)=-1$ gives rise to the quaternion simple $R\oplus R^*\in \Rep(G\times \Z_2^T)$, which contributes $2\frac{(2d_R)^2}{4}=2d_R^2$. In total the sum is
\begin{align}
    &\sum_{R\in \Rep(G),\nu(R)=1} 2d_R^2+\frac{1}{2}\sum_{ \nu(R)=0}4d_R^2+\sum_{\nu(R)=-1} 2d_R^2\nonumber\\
    &=2\sum_R d_R^2=2|G|=|G\times \Z_2^T|
\end{align}
where the $\frac{1}{2}$ in front of the second term is due to over-counting since $R\oplus R^*=R^*\oplus R$. 

\begin{example}
    We consider a concrete example. Let $Q_8:=\{\pm 1, \pm i,\pm j,\pm k\}\subset \bH$ be the quaternion group. We compute the structure of $\Rep(Q_8\times \Z_2^T)$. We start with the representations of $Q_8$. There are four 1-dimensional representations $\lid, I, J,K$,   determined by the requirement that $i$, $j$, and
$k$ respectively act trivially. These four 1-dimensional representations become four real simples in $\Rep(Q_8\times \Z_2^T)$. There is one more 2-dimensional irreducible representation of $Q_8$, denoted as $Q$. It can be written as 
\begin{align*}
    Q(-1)=\begin{pmatrix}
        -1&0\\
        0&-1
    \end{pmatrix},~Q(i)=\begin{pmatrix}
        i&0\\
        0& -i
    \end{pmatrix}, Q(j)=\begin{pmatrix}
        0&1\\
        -1&0
    \end{pmatrix}
\end{align*}
 One calculates $\nu(Q)=-1$, therefore $Q$ descends to a quaternion simple in $\Rep(Q_8\times \Z_2^T)$.
    The fusion rules are 
    \begin{align}
    &I\otimes Q=J\otimes Q=K\otimes Q=Q\otimes I=Q\otimes J=Q\otimes K=Q\nonumber\\
        &Q\otimes Q=4(\lid\oplus I\oplus J\oplus K)\nonumber\\
        & I\otimes I=J\otimes J=K\otimes K=\lid\nonumber\\
        & I\otimes J=K
    \end{align}
   This fusion rule resembles that of a complex Tambara-Yamagami category $\TY(\Z_2\times \Z_2)$. Indeed $\Rep(Q_8\times \Z_2^T)$ is an {\bf $\R$-real  $\Z_2\times \Z_2$-Tambara-Yamagami category}~\cite{plavnik2024TYreal}.
\end{example}
\subsubsection{Module categories over $\Rep(G^T)$}\label{sec_module_RepGT}
We now describe and classify module categories over a general $\Rep(G^T)$-category. Recall that an indecomposable module category over a usual complex fusion  $\Rep(G)$ category is labeled by a pair $(K, \psi)$ where $K<G$ is a subgroup and $\psi\in H^2(K,U(1))$ is a 2-cocycle. The corresponding module category is 
\be 
\Rep(K)^\psi:=\Mod(\C[K]^\psi)
\ee
that is, the category of projective representations of $K$ with projective phase $\psi$.  The indecomposable modules over $\Rep(G^T)$ has essentially the same structure. An indecomposable module is labeled by a pair $(K^T, \psi)$ where $K^T<G^T$ and $\psi$ is a twisted 2-cocycle of $K^T$. The corresponding module category is 
\be
\Rep(K^T)^\psi:=\Mod(\C[K^T]^\psi)
\ee
Later we will show that $\Rep(G^T)$ is Morita equivalent to $\Vec_{G^T}$, hence it is not surprising that modules over $\Rep(G^T)$ have the same labeling as those over $\Vec_{G^T}$.The Morita equivalence also implies that the module categories described above are a complete list. 

\subsection{Non-invertible Time-Reversal Symmetry}
\label{sec:TY}

So far we have only discussed invertible Galois-real fusion categories. The categories $\Vec_{G^T}^\omega$ describe all the invertible anti-linear symmetries in (1+1)d. For the purpose of studying generalized anti-linear symmetries, it is desirable to construct Galois-real fusion categories that contain non-invertible simples. One such construction is the Galois-real Tambara-Yamagami categories, which have been completely classified in~\cite{plavnik2024TYreal}. We provide a summary of their structure here and discuss some examples.

A Galois-real Tambara-Yamagami category 
\be 
{\TY_{\overline{\C}}(A)}
\ee 
has underlying category $\Vec_A\oplus \Vec$, where $A$ is an abelian group, and $\Vec$ is generated by a simple $m$. Denote the simples of $\Vec_A$ as $\delta_a,~a\in A$, the fusion rule is 
\be 
\ba
    \delta_a\otimes \delta_b &=\delta_{ab}\cr 
    \delta_a\otimes m&=m \cr 
    m\otimes m&=\bigoplus_{a\in A} \delta_a \,,
\ea
\ee
which is identical to that of a complex Tambara-Yamagami category $\TY(A)$. However we declare $m$ to be anti-linear. This affects the classification of Galois-real Tambara-Yamagami categories. Recall that a complex Tambara-Yamagami category $\TY(A)$ is determined by a non-degenerate symmetric bicharacter $q$ of $A$, and a Frobenius-Schur indicator $\chi=\pm$. On the other hand, according to the classification in~\cite{plavnik2024TYreal}, a Galois-real Tambara-Yamagami category ${\TY_{\overline{\C}}(A)}$  is determined solely by a non-degenerate \emph{skew-symmetric} bicharacter $q$ of $A$, where skew-symmetric means $q(a,b)=q(b,a)^*$. Later we will provide a SymTFT perspective on this classification. Let us look at some examples. 
\begin{example}\label{eg_TY}
    Take $A=\Z_2$, there is a unique skew-symmetric bicharacter of $\Z_2$, therefore there is a unique $\cTY(\Z_2)$. It has three simples $\lid, f, m$ with fusion rule
    \be
    f\otimes f=\lid,\ f\otimes m=m\otimes f=m,\ m\otimes m=\lid \oplus f.
    \ee
    Other than the fact that $m$ is anti-linear, the category $\cTY(\Z_2)$ is almost the same as a $\TY(\Z_2)$. However the classifications are different: there are two $\TY(\Z_2)$ categories but only one $\cTY(\Z_2)$ category. 
\end{example}

\begin{remark} {\bf Critical Ising Model.}
    It is well-known that the fusion category $\Ising=\TY(\Z_2,+)$ appears at the (1+1)d Ising critical point. One may wonder whether $\cTY(\Z_2)$ appears at some critical point. Indeed if one enriches the Ising critical point with a $\Z_2^T$-symmetry, then the full symmetry of the critical point enlarges to $\TY(\Z_2,+)\boxtimes \Vec_{\Z_2^T}$. Take the fusion sub-category generated by $\lid, f, \sigma\otimes T$, one obtains the complex-real fusion category $\cTY(\Z_2)$. This means the category  $\cTY(\Z_2)$ can be viewed as a sub-symmetry of the full symmetry of a $\Z_2^T$-symmetric Ising critical point. The non-invertible anti-linear generator $m$ can be viewed as the combination of the non-invertible Kramers-Wannier symmetry $\sigma$ and the invertible time-reversal symmetry $T$. 
\end{remark}

\begin{example}
    Take $A=\Z_4\times \Z_4=\langle x,y\rangle$, a skew-symmetric bicharacter $q$ has been given in~\cite{plavnik2024TYreal}:
    \be
    q(x,x)=1,~q(x,y)=i,~q(y,y)=-1 \,.
    \ee
    Thus there is a Galois-real Tambara-Yamagami category $\cTY(\Z_4,q)$. We will revisit this category when we discuss its SymTFT. 
\end{example}

The Galois-real Tambara-Yamagami categories contain non-invertible anti-linear simple objects. This means they represent non-invertible anti-linear symmetries. It seems some Galois-real Tambara-Yamagami categories can be embedded into Galois-real fusion categories of the form $\cc\rtimes \Z_2^T$ where $\cc$ is a complex fusion category. In this case one can view the Galois-real Tambara-Yamagami category as a sub-symmetry of a system with linear symmetry $\cc$ and time-reversal $\Z_2^T$. We do not know if all Galois-real Tambara-Yamagami categories can be embedded this way.

\subsection{The $\cC\rtimes \Z_2^T$ Categories and Time-Reversal SSB}\label{sec_semi_product}

Consider a gapped phase/TFT that is already enriched with a finite internal generalized symmetry $\cc$ (a fusion category), a naive idea of adding time-reversal to the system is to form a product $\cc\times \Z_2^T$. However there is in fact not a well-defined notion of taking product with $\Z_2^T$. This is because time-reversal \textbf{can not} act trivially on an internal symmetry category, since it at least needs to act as complex conjugation on scalars. Therefore ``adding time-reversal" can at best be done by taking a semi-direct product $\cc\rtimes \Z_2^T$, which means at the level of simple objects and fusion ring it is a semi-direct product of $\cc$ and $\Z_2^T$. In general there is no canonical way of forming the semi-direct product, as we need to specify how time-reversal acts on the lines  of $\cc$. 

Assume that an action of $\Z_2^T$ on $\cc$ is given, where by an action we mean a categorical action by $\C$-anti-linear monoidal functor:
\begin{itemize}
    \item There is a $\C$-\textbf{anti-linear} monoidal functor 
    \be
    T: \cc\to \cc.
    \ee
    \item Such that there is a monoidal natural isomorphism
    \be
    \alpha: T\circ T\Rightarrow \id_\cc
    \ee
    \item Such that $\alpha$ satisfies an associativity condition:
    \be
        (T\circ T\circ T\xrightarrow{\alpha T} T)= (T\circ T\circ T\xrightarrow{T(\alpha)} T)
    \ee
\end{itemize}

Then we can build a Galois-real fusion category $\cc\rtimes \Z_2^T$, with the following structure:
\begin{itemize}
    \item A simple object is a pair $(x,s)$ where $x\in \cc$ is a simple and $s\in \Z_2^T$. 
    \item There is no non-zero morphism between $(x,s),~(y,t)$ if $s\neq t$, and 
    \be
    \Hom((x,s),(y,s))=\Hom_\cc(x,y).
    \ee
    Composition of morphisms is given by composition in $\cc$.
    \item The fusion rule is 
    \be
    (x,s)\otimes (y,t)=(x\otimes s(y), st)
    \ee
    where $s(-)$ stands for the $s\in\Z_2^T$-action on objects.
    \item The associator is defined in the obvious way.
\end{itemize}
Because the action of $T\in \Z_2^T$ is $\C$-anti-linear on $\cc$, the category $\cc\rtimes \Z_2^T$ is canonically a Galois-real fusion category. This means the subcategory spanned by objects $(x,1)$ form a fusion category equivalent to $\cc$, and the objects of the form $(x,T)$ are Galois-nontrivial/anti-linear.

\begin{example}
    The Ising model~\eqref{eq_Ising} admits another time-reversal symmetry $U'_T:=K$ that simply acts as complex conjugation in the $Z_i$-basis. Then with respect to this time-reversal action the model is time-reversal symmetric everywhere in the phase diagram. This implies the critical point $(J=1/2)$ is also time-reversal symmetric. The critical point is well-known to enjoy a non-invertible Ising category symmetry $\Ising$. Indeed $\Ising$ is invariant under a time-reversal action that is identity on objects. Then apply the semi-direct product construction we obtain a Galois-real fusion category $\Ising\rtimes \Z_2^T$ which can be thought of as the full symmetry of a $\Z_2^T$-symmetric Ising critical point. The non-invertible time-reversal symmetry $\TY_{\ol{\C}}(\Z_2)$ can be viewed as a sub-category of $\Ising\rtimes \Z_2^T$ generated by $1, f, \sigma T$ where $\sigma$ is the non-invertible simple of $\Ising$ and $T\in \Z_2^T$.
\end{example}

We note that it is not always possible to ``add time-reversal" to an internal generalized symmetry. The above construction requires the existence of a categorical $\Z_2^T$-action on the internal generalized symmetry $\cc$, which may not always exist. For instance, the existence of an anti-linear action requires $\cc$ to be equivalent to its complex conjugation $\cc^*$. Consider for instance the fusion category $\Vec_{\Z_3}^\omega$ with a nontrivial cocycle $\omega\in H^3(\Z_3,U(1))=\Z_3$, then $(\Vec_{\Z_3}^\omega)^*\simeq \Vec_{\Z_3}^{\omega^{-1}}$, which is \textbf{not} equivalent to $\Vec_{\Z_3}^\omega$. This means a system with $\Vec_{\Z_3}^\omega$ symmetry can not be further enriched with $\Z_2^T$-symmetry. 

Just like a real vector space can be viewed as a complex one equipped with a $\Z_2^T$-action, an $\R$-real fusion category can be viewed as a complex fusion category equipped with a $\Z_2^T$-action. Hence the data needed to construct a semi-direct product category $\cc\rtimes \Z_2^T$ is the same as an $\R$-real fusion category. More concretely, one can construct a semi-direct product Galois-real fusion category as follows. 
\begin{itemize}
    \item Start with an $\R$-real fusion category $\cD_\R$. 
    \item The complexification $\cc:=\cD_\R\boxtimes_\R\Vec_\C$ is a complex fusion category equipped with a canonical $\Z_2^T$-action(complex conjugation on the $\Vec_\C$ factor). 
    \item Use the canonical $\Z_2^T$-action on $\cc$, construct the semi-direct product $\cc\rtimes \Z_2^T$.
\end{itemize}
In fact the constructed Galois-real fusion category $\cc\rtimes\Z_2^T$ is Morita equivalent to the $\R$-real one $\cD_\R$ we started with. We provide the proof in~\Cref{app_semi_direct}. The physical interpretation of this Morita equivalence is as follows. 

Consider a $(1+1)$d system/phase whose topological defects form a complex fusion category $\cc$. The phase could be either gapless(say a CFT) or gapped(say an SPT protected by some internal generalized symmetry). Here $\cc$ is not interpreted as a symmetry category but the category of IR symmetric defects in the phase, such as topological defects in a CFT, or topological defects in a gapped phases protected by some internal symmetry. Now assume that the system is time-reversal symmetric, then there must  exist a consistent categorical $\Z_2^T$-action on the complex fusion category $\cc$. Then the category of \textbf{time-reversal symmetric} defects would be made of objects in $\cc$ that are invariant under the $\Z_2^T$-action up to isomorphism. The mathematical theory of Galois descend~\cite{etingof2012descentformstensorcategories} says that taking the (homotopy) fixed points of $\cc$ under a $\Z_2^T$-action produces an $\R$-real fusion category $\cD_\R$, whose complexification recovers $\cc$. Now consider breaking $\Z_2^T$ spontaneously in this system(a domain wall between the two vacua), we arrive at a new phase whose category of symmetric defects is exactly $\cc\rtimes \Z_2^T$. This because when time-reversal is broken, the $T\in \Z_2^T$-action creates a nontrivial IR defect in the system that fuses with defects in $\cc$ according to the $\Z_2^T$-action on $\cc$. The process of going from an $\R$-real fusion category $\cD_\R$ to the Morita equivalent Galois-real fusion category $\cc\rtimes \Z_2^T$ describes the spontaneously breaking of time-reversal. 

A corollary of this Morita equivalence is that an $\R$-real fusion category is always Morita equivalent to a Galois-real one(by the semi-direct product construction). However the converse is not true in general. 

\begin{example}\label{eg_GT_semi_prod}
    Let $G^T$ be a group-like anti-linear symmetry with linear part being $G_0\le G^T$. Then $\Rep(G^T)$ always complexifies to $\Rep(G_0)$. Hence for every extension 
    \be
    1\to G_0\to G^T\to \Z_2^T \to 1
    \ee
    we can build a semi-direct product Galois-real fusion category $\Rep(G_0)\rtimes \Z_2^T$. The Morita equivalence $\Rep(G^T)\simeq_{\text{Morita}} \Rep(G_0)\rtimes \Z_2^T$ can be understood as follows. Start with a $G_0$-SPT phase, whose category of symmetric defects is $\Rep(G_0)$. We can further enrich the system with time-reversal by assigning a consistent $\Z_2^T$-action on the $G_0$-charges. Then the system can be equivalently viewed as being enriched with an extended $G^T$-symmetry, and the category of symmetric defects becomes $\Rep(G^T)$. Then we can spontaneously break the symmetry down to $G_0<G^T$, which means the quotient $G^T/G_0=\Z_2^T$ is broken. The category of symmetric defects in this $\Z_2^T$-broken phase is then $\Rep(G_0)\rtimes \Z_2^T$.
\end{example}

\section{Gapped Phases with Anti-Linear Symmetries}
\label{sec_GappedPhases}

With the physical and mathematical preparation in place, we now present a categorical framework for $(1+1)$d gapped phases with generalized anti-linear symmetries. 

\subsection{The Setup}\label{sec_setup}
Let us briefly recall the theory of $(1+1)$d gapped phases with internal categorical symmetry, or equivalently, the theory of $(1+1)$d extended TFTs with fusion categorical symmetry. See e.g.~\cite{Huang:2021zvu, Thorngren:2019fuc, Thorngren:2021yso}, for a more in-depth discussion and~\cite{Inamura:2021szw, Bhardwaj:2024wlr, Bhardwaj:2024kvy} for systematic lattice constructions. We will take the following definition.
\begin{definition} 
    \begin{itemize}
        \item A finite internal symmetry in $(1+1)$d is a fusion category. 
        \item Let $\cS$ be a finite internal symmetry, a gapped phase enriched with $\cS$ is a module category $\cM$ over $\cS$. 
    \end{itemize}
\end{definition}

From the topological field theory(TFT) perspective, a TFT with symmetry $\cS$ is one that can be coupled to ``$\cS$-background fields", where by a ``$\cS$-background field" we mean a mesh of $\cS$-defect network  on the underlying manifold. The module category $\cM$ has both a bulk interpretation and a boundary interpretation. The bulk interpretation is as follows: The module structure on $\cM$ may be written as a monoidal functor $\cS\to \End(\cM)$, which specifies how the symmetry defects map to the defects of the underlying TFT. This may be regarded as a UV-to-IR map, where the symmetry $\cS$ is preserved along the RG flow, but the IR fixed point has emergent symmetry $\End(\cM)$ that can be either larger or smaller than the UV symmetry $\cS$. The module structure $\cS\to \End(\cM)$ determines how the symmetry defects flow to the IR defects. From the boundary perspective, the underlying category of the module category is the category of boundary conditions of the underlying 2D TFT, and the module structure specifies how the symmetry generators act on the boundary conditions. 

In order to include anti-linear symmetries into this categorical framework, we must first identify the categorical model for anti-linear symmetries themselves. Once this is done, we then further specify how the abstract symmetry interact with the underlying gapped phase/TFT. 

To motivate the categorical description of anti-linear symmetries, let us first consider the simplest case of time-reversal symmetry $\Z_2^T$.  A gapped phase with time-reversal symmetry is described in the infrared by an unoriented TFT, namely a TFT defined on unoriented manifolds. Let $M$ be an  unoriented manifold on which the unoriented TFT lives, and $w_1\in H^1(M,\Z_2)$ be the 1st Stiefel-Whitney class of $M$. A codimension-1 locus $\Gamma$ Poincar\'e dual to $w_1$ can be viewed as the time-reversal defect. Namely, the unoriented theory on $M$ can be viewed as obtained from an oriented theory on $M\backslash \Gamma$ glued along a time-reversal defect $T$ located at $\Gamma$. The trivial defect $\lid$ and the time-reversal defect $T$ together generate the symmetry category $\Vec_{\Z_2^T}$. The anti-linear natural of the object $T\in \Vec_{\Z_2^T}$ encodes the fact that a $T$-defect implements local orientation/time-reversal.

More generally, a gapped phase with $G^T$-symmetry is described by a TFT defined on manifolds with $G^T$-background fields. A $G^T$-background can be viewed as a mesh of $G^T$-defect network, such that when we apply the map $s: G^T\to \Z_2^T$, the mesh reduces to a cycle that is Poincar\'e dual to $w_1$. I.e., if $A_G\in H^1(M,G^T)$ is the $G^T$-gauge field, then $s(A_G)=w_1$. 
Therefore the $G^T$-defects divide naturally into two types: orientation preserving ones and orientation reversing ones. Together they generate the symmetry category $\Vec_{G^T}$. 

Even more generally, in the spirit of ``symmetry=topological defect", a general (non-invertible) anti-linear symmetry should consist of defects that are either orientation-preserving or orientation-reversing. Together they naturally form a Galois-real fusion category $\cc_1\oplus \cc_T$. The defects in $\cc_T$ should really be thought of as domain walls between an oriented theory $\cT$ and its orientation reversal $\cT^{\mathrm{rev}}$. 
A field theory with symmetry $\cc_1\oplus \cc_T$ is one that can be defined on manifolds with mesh of defect network valued in $\cc_1\oplus \cc_T$, with the condition that if we apply the Galois indicator map $s:  \cc_1\oplus \cc_T\to \Z_2^T$, the defect network reduces to a cycle that is Poincar\'e dual to $w_1$ of the underlying manifold. See Fig~\Cref{fig_anti-linear_defect}.
\begin{figure}
    \centering
\begin{tikzpicture}[scale=0.5,x=1cm,y=1cm,line cap=round,line join=round]

  \coordinate (A) at (-3.0,-0.2);
  \coordinate (B) at (2.1,0.55);
  \coordinate (C) at (3.75,-1.75);

  \draw[blue!80!black, line width=1.1pt] (A) -- (-4.9,1.95);
  \draw[blue!80!black, line width=1.1pt] (B) -- (3.75,2.75);
  \draw[blue!80!black, line width=1.1pt] (C) -- (2.55,-3.45);

  \draw[red!90!black, line width=1.1pt] (-3.4,-3.0) -- (A);
  \draw[red!90!black, line width=1.1pt] (A) -- (B);
  \draw[red!90!black, line width=1.1pt] (B) -- (C);
  \draw[red!90!black, line width=1.1pt] (C) -- (5.7,-1.55);

  \node at (0.0,1.75) {\Large $\mathcal{T}^*$};
  \node at (-4.45,-0.65) {\Large  $\mathcal{T}^*$};
  \node at (4.95,0.05) {\Large  $\mathcal{T}^*$};
  \node at (0.9,-1.65) {\Large  $\mathcal{T}$};
  \node at (4.65,-2.9) {\Large  $\mathcal{T}$};
\end{tikzpicture}
    \caption{A defect network for a generalized anti-linear symmetry. The network is a diagram valued in the symmetry category $\cc_1\oplus \cc_T$. The lines in blue/red are liner/anti-linear defects. The theories on different sides of an anti-linear defect are complex conjugate of each other. 
    If one traces along the anti-linear defects, one gets a cycle (red line) Poincar\'e dual to $w_1(M)$.}
    \label{fig_anti-linear_defect}
\end{figure}

This concludes our argument that a generalized anti-linear symmetry should be modeled as a Galois-real fusion category $\cc_1\oplus \cc_T$. We make this into a definition. 
\begin{definition}
    A finite, potentially anti-linear, generalized symmetry in $(1+1)$d is a Galois-real fusion category $\cC=\cC_1\oplus \cC_T$. The component $\cC_1$ is the fusion category of linear symmetry generators, and $\cC_T$ is the category of anti-linear symmetry generators.
\end{definition}

Next we argue that a 2D gapped phase/TFT enriched with symmetry $\cc_1\oplus \cc_T$ should be defined as a module category over $\cc_1\oplus \cc_T$. We will provide a bulk argument and a boundary argument, following the bulk and boundary interpretation of module categories. From the bulk perspective, the symmetry defects in $\cc_1\oplus \cc_T$ should flow under RG to defects of the underlying gapped phase.  Let $\cM\in 2\Vec$  be a finite semi-simple category defining a 2D gapped phase. Then a linear defect $x\in \cc_1$ flows to a defect within the gapped phase defined by $\cM$, i.e., an object in $\End(\cM)$. If $y\in \cc_T$ is an anti-linear symmetry defect, then it should flow to a domain wall between the phase defined by $\cM$ and its complex conjugate defined by $\cM^*$, i.e., an object in $\mathrm{Func}(\cM, \cM^*)$. Together $\End(\cM)\oplus \mathrm{Func}(\cM, \cM^*)$ is nothing but $\End_\R(\cM)$ the category of $\R$-linear endofunctors of $\cM$. This is because any  $\R$-linear endofunctor of $\cM$ can be written as the direct sum of a $\C$-linear one and a $\C$-anti-linear one. We see that the action of the symmetry on the underlying gapped phase is described by a monoidal functor 
\be
\cc_1\oplus \cc_T \to \End_\R(\cM)
\ee
which is the same as a module category structure on $\cM$. 

We now discuss the boundary perspective. Again let $\cM\in 2\Vec$ define a 2D gapped phase. Then $\cM$ itself can be regarded as the category of boundary conditions for the gapped phase. A symmetry generator $a\in \cc_1\oplus \cc_T$ acts on the boundary as a boundary condition changing operator:
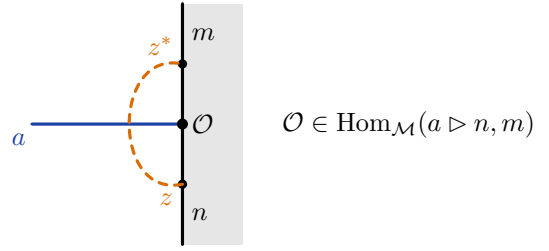
\begin{center} 
\begin{tikzpicture}[scale=0.8, line cap=round, line join=round]

  \colorlet{wallshade}{gray!20}
  \colorlet{scalarline}{orange!85!black}

  \fill[wallshade] (0,-2) rectangle (1,2);
  \draw[very thick] (0,-2) -- (0,2);

  \node[right] at (0, 1.5) {$m$};
  \node[right] at (0,-1.5) {$n$};

  \draw[myblue, very thick] (-2.5,0) -- (0,0);
  \node[myblue, below left] at (-2.45,0) {$a$};

  \fill (0,0) circle (2.6pt);

  \coordinate (zb)  at (0,-1);
  \coordinate (zst) at (0,1);

  \fill (zb)  circle (2.2pt);
  \fill (zst) circle (2.2pt);

  \draw[scalarline, dashed, very thick]
    (zb) .. controls (-1.18,-1.25) and (-1.18,+1.25) .. (zst);

  \node[scalarline, below left]      at (zb)  {$z$};
  \node[scalarline,above left] at (zst) {$z^{*}$};
 \node[right] at (1.5,0) {$\cO \in \Hom_{\cM}(a\rhd n,m)$};
 \node[right] at (0,0) {$\cO$};
\end{tikzpicture}
\end{center}
The action of $\cc=\cc_1\oplus \cc_T$ on the category of boundary conditions $\cM$ provides $\cM$ with a module category structure: $\cc\times \cM\to \cM$ that is only $\R$-linear: If $a$ is an anti-linear line, then a scalar on the boundary gets complex conjugated when it passes the junction. This means the functor 

\be
a\rhd -: \cM\to \cM
\ee
is $\C$-anti-linear when $a\in \cc_T$ is anti-linear.  This concludes our argument that a gapped phase with symmetry $\cc=\cc_1\oplus \cc_T$ should be defined a module category over $\cc$. We make this into a definition.

\begin{definition}\label{def_gapped_phase}
    Let $\cC=\cc_1\oplus \cc_T$ be a finite anti-linear symmetry.
    \begin{itemize}
        \item A $(1+1)$d gapped phase enriched with $\cC$ is a module category $\cM$ over $\cC$.
        \item The underlying gapped phase(with no symmetry) is $\cM\in 2\Vec$.
        \item The category of symmetric defects in the phase $\cM$ is $\End_\cC(\cM)$.
        \item The category of domain walls between two phases $\cM,\cN$ is $\mathrm{Func}_\cC(\cM,\cN)$.
    \end{itemize}
\end{definition}

Notice that although in the definition of module category we only require the underlying category to be $\R$-linear, the fact that the symmetry category $\cC$ is \emph{Galois-real} implies that there is an inclusion $\Vec\subset \cC$, hence any $\cC$-module is automatically a $\Vec$-module thus is automatically $\C$-linear. It is only the module structure on $\cM$ that is not $\C$-linear. This is why $\cM\in 2\Vec$ makes sense. 

The following two examples are two straightfoward families of gapped phases that can be defined for any symmetry category.

\begin{example}
    [SPT] An SPT for $\cc_1\oplus \cc_T$ is a module category $\cM$ such that $\cM\simeq \Vec$ as a plain category. This means the underlying gapped phase is trivial. Notice that the correct notion of ``fiber functor" for a Galois-real fusion category is now a monoidal functor $\cc_1\oplus \cc_T\to \End_\R(\Vec)=\Vec_{\Z_2^T}$.
\end{example}

\begin{example}
    [Complete SSB] For any symmetry $\cC$, there is a regular module category $\cC$. The category of symmetric defects in the phase defined by $\cC$ is $\End_\cC(\cC)\simeq \cC$. Hence the regular module category $\cC$ defines the complete $\cC$-SSB phase. The underlying gapped phase with no symmetry is $\End(\cC)$. Therefore the number of ground states is  the number of simple objects in the symmetry category $\cC$.  
\end{example}

For any symmetry category $\cc$, one can define a 2-category as follows. 
\begin{itemize}
    \item Objects are $(1+1)$d  gapped phases enriched with $\cc$. 
    \item 1-morphisms are domain walls between the gapped phases. 
    \item 2-morphisms are local operators on the domain walls.
\end{itemize}
The definition~\Cref{def_gapped_phase} can then be succinctly summarized as the statement that 
\begin{center}
    The 2-category of $(1+1)$d  gapped phases enriched with $\cc$ is $\Mod(\cc)$.
\end{center}
For instance, the 2-category of $(1+1)$d gapped phases with $\Z_2^T$ symmetry is $\Mod(\Vec_{\Z_2^T})\simeq 2\Vec_\R$ which has three simple objects: $\Vec_\R, \Vec_\C, \Vec_\bH$.
\subsection{Order Parameters and Drinfeld Center}\label{sec_op_center}
For an internal categorical symmetry $\cc$, the SymTFT is given by the Drinfeld center $\cZ(\cc)$, which can be thought of as the category of all local and non-local(string) order parameters of the symmetry $\cc$. For instance, when $\cc=\Vec_G$ is a group-like internal symmetry, then a simple object in $\cZ(\Vec_G)$ is labeled as $(\sigma, r)$ where $\sigma\subset G$ is a conjugacy class and $r\in \Rep(C_G(\sigma))$ is a representation of the centralizer of (some element in)$\sigma$. A pure charge $(e_G, r)\in \Rep(G)$ corresponds to a local order parameter of $G$, and a dyonic anyon $(\sigma ,r)$ corresponds to a string order parameter with flux sector $\sigma$ and charge decoration $r$. A $(1+1)$d gapped phase with symmetry $G$ is completely determined by the behavior of these local and non-local order parameters, which in the SymTFT language correspond to the condensation and confinement of anyons on the physical boundary. 

As we have discussed for time-reversal in~\Cref{sec_op_Z2T}, a pure time-reversal symmetry does not have a string order that can detect the Haldane phase. This is not special to time-reversal: any anti-linear symmetry generator does not have a string order parameter. We will argue that given this observation, the Drinfeld center $\cZ(\cc_1\oplus \cc_T)$ is still the correct category of all local and string order parameters of an anti-linear symmetry $\cc_1\oplus \cc_T$.  Drinfeld center is a notion that can be defined for any monoidal category, hence well-defined for real fusion categories, and in particular for Galois-real fusion categories. For a categorical anti-linear symmetry $\cc:=\cc_1\oplus \cc_T$, a simple object in $\cZ(\cc)$ is again given by a pair $(x,\gamma)$, where $x\in \cc$ and $\gamma: x\otimes y\simeq y\otimes x,~\forall y\in \cc$ is the half-braid. However, as shown in~\cite{sanford2025fusion}, the object $x$ is in fact strictly in $\cc_1$, otherwise there is no consistent half-braid $\gamma$. This is because, assuming a half-braid $\gamma$ does exists, then it gives natural isomorphism $x\otimes \lid \simeq \lid \otimes x$, which identifies the left and right actions by scalars, implying $x$ must be a linear(Galois-trivial) object $x\in \cc_1$. This precisely matches with the expectation that an anti-linear symmetry generator does not have a string order parameter: assume a string order parameter does exist, then it can be used to identify the action by scalars from two sides of the string order:
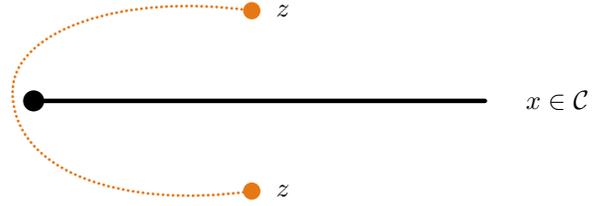
\begin{center}
\begin{tikzpicture}[x=1cm,y=1cm,line cap=round,line join=round]
    \definecolor{myorange}{RGB}{230,120,20}

    \coordinate (L) at (0,0);
    \coordinate (R) at (6,0);
    \coordinate (U) at (2.9,1.2);
    \coordinate (D) at (2.9,-1.2);

    \draw[line width=1.7pt] (L) -- (R);
    \fill (L) circle (0.14);

    \draw[draw=myorange, line width=1.05pt, dash pattern=on 0pt off 1.8pt]
        (U)
        .. controls (1.15,1.45) and (-0.28,0.92) .. (-0.28,0.12)
        .. controls (-0.28,-0.92) and (1.15,-1.45) .. (D);

    \fill[myorange] (U) circle (0.115);
    \fill[myorange] (D) circle (0.115);

    \node[right=0.20cm] at (U) {$z$};
    \node[right=0.20cm] at (D) {$z$};
    \node[right=0.42cm] at (R) {$x\in\mathcal{C}$};
\end{tikzpicture}
\end{center}
which implies $x$ is a linear symmetry generator. Hence the mathematical machinery of  Drinfeld center surprisingly still captures the correct physical picture. 

It is illuminating to consider the structure of $\cZ(\Vec_{G^T})$ for a general (non-anomalous) group-like symmetry $G^T$. A simple object in $\cZ(\Vec_{G^T})$ is a pair $(\sigma \subset G_0, r)$ where $\sigma$ is a conjugacy class strictly in the unitary subgroup $G_0$, and $r\in \Rep(C_G(\sigma)^T)$. We see that these simple objects still form the full list of local and non-local order parameters for the symmetry $G^T$. In particular, $\cZ(\Vec_{G^T})$ contains a copy of $\Rep(G^T)$, an $\R$-real fusion category.

Here we are not assuming any SymTFT interpretation of the Drinfeld center, simply viewing it as the category of 1D order parameters associated with $\cc$. The SymTFT interpretation will be discussed in~\Cref{sec_T_SymTFT}. However it is already clear from the discussion here that the Drinfeld center can not describe gapped phases that are not captured by string orders, such as the Haldane chain.
\subsection{Duality and Gauging}
One advantage of the categorical formulation of symmetries and phases is that it allows for straightfoward computation of duality and gauging. By that we mean that there is a well-developed Morita theory for (complex) fusion categories that capture exactly the physical operation of (generalized) gauging~\cite{Bhardwaj:2017xup}. 

\subsubsection{Gauging between Symmetry Categories}
Recall that for two complex fusion categories $\cc, \cD$, we say they are Morita equivalent if there is some algebra $A\in \cc$ such that $\cD\simeq \Bimod_\cc(A)$ are equivalent as fusion categories. In this case we say $\cD$ is obtained by condensing $A$ in $\cc$. Equivalently, $\cc, \cD$ are Morita equivalent if there is a $\cc$-module category $\cM$ such that $\cD\simeq \End_\cc(\cM)$. If $\cc, \cD$ are symmetry categories for two internal generalized symmetries, then being Morita equivalent has the consequence that the 2-categories of gapped phases with symmetry $\cc$ and $\cD$: $\Mod(\cc),~\Mod(\cD)$ are equivalent. This means we can relabel phases enriched with symmetry $\cc$ as phases enriched with symmetry $\cD$ in such a way that does not affect any physical observables. 

The Morita theory works well over $\R$ as well, see our discussion in ~\Cref{sec_module_cat}. Hence we say two anti-linear symmetry categories $\cc, \cD$ are dual to each other if they are Morita equivalent as real fusion categories. This means there exists some algebra $A\in \cc$ such that $\cD\simeq \Bimod_\cc(A)$. Notice that since both $\cc, \cD$ are Galois-real, the algebra $A$ is necessarily in $\cc_1$ according to~\Cref{th_gauging}. This is the mathematical manifestation of the fact that one cannot literally gauge time-reversal, or any anti-linear symmetry generator. The gauging always happens inside the linear part $\cc_1$. Then if $\cc, \cD$ are two dual anti-linear symmetries, the 2-category of gapped phases: $\Mod(\cc),~\Mod(\cD)$ are equivalent.

\subsubsection{Gauging between Gapped Phases}

Although the symmetry categories are always Galois-real, the category of symmetric defects in a gapped phase with symmetry $\cc$ can be either Galois-real or $\R$-real. For instance, with $\Z_2^T$-symmetry, the category of symmetric defects is $\Vec_\R$ in the $\Z_2^T$-SPT phases.  Given an anti-linear symmetry $\cc$, the category of symmetric defects in a gapped phase defined by a module category $\cM$ over $\cc$ is $\End_\cM(\cc)$. By definition, $\End_\cM(\cc)$ is Morita equivalent to $\cc$. Hence the categories of symmetric defects in all gapped phases with the same symmetry $\cc$ are Morita equivalent, and are Morita equivalent to the symmetry category $\cc$ itself. 

\subsection{$\Z_2^T$ Gapped Phases Revisited}
We discussed the three $\Z_2^T$ gapped phases in ~\Cref{sec_Z2TPhases} including their module category descriptions. One thing that was not discussed there is the underlying gapped phase. Let us confirm here that the three phases have the expected underlying gapped phase. 

The symmetric trivial and the Haldane phase are both described by module categories whose underlying category is $\Vec$. Hence the underlying gapped phase for these two phases is $\End(\Vec)=\Vec$, the trivial gapped phase with one ground state, as expected.  On the other hand the $\Z_2^T$-SSB phase is described by the regular module category $\Vec_{\Z_2^T}$. As a plain category it is equivalent to $\Vec\oplus \Vec$, hence the underlying gapped phase is $\End(\Vec\oplus \Vec)$ which is a gapped phase with two ground states. The two ground states are permuted by $\Z_2^T$-action:
\be
T\rhd \lid=T,~T\rhd T=\lid
\ee
confirming that the module category $\Vec_{\Z_2^T}$ describes the $\Z_2^T$-SSB phase.

\subsection{General Structure of $G^T$ Gapped Phases}\label{sec_GT_general}
We now consider $(1+1)$d gapped phases with a general group-like $G^T$-symmetry and with no anomaly. The symmetry category is $\Vec_{G^T}$ and $G^T$-gapped phases are module categories over $\Vec_{G^T}$. We have classified modules over $\Vec_{G^T}$ in~\Cref{sec_module_GT}. The indecomposable modules are denoted as $\cM(K^T, \omega)$ where $K^T<G^T$ and $\omega$ is a twisted 2-cocycle for $K^T$. It is clear that this matches with the standard physical classification of $(1+1)$d gapped phases with $G^T$-symmetry. The subgroup $K^T$ stands for the unbroken symmetry in the phase, and $\omega$ labels the $K^T$-SPT phase in any of the vacua. From the general discussion in~\Cref{sec_setup}, the number of ground states in a phase defined by $\cM(K^T, \omega)$ is the number of simple objects in $\cM(K^T, \omega)$, which is $|G^T/K^T|$. 

According to the general discussion in~\Cref{sec_setup}, the category of symmetric defects in the phase $(K^T,\omega)$ is given by the Morita dual real fusion category 
\be
\Bimod_{\Vec_{G^T}}(\C[K^T]^\omega).
\ee
Then by~\Cref{th_gauging}, this real fusion category is Galois-real exactly when $K^T\le G_0$ is completely unitary and $\R$-real when $K^T$ contains any anti-unitary element. This also means the Morita dual categories $\Bimod_{\Vec_{G^T}}(\C[K^T]^\omega)$ for $K^T\le G_0$ can be viewed as symmetries dual to $\Vec_{G^T}$. 

For two gapped $G^T$-phases $(K^T, \omega), (H^T, \psi)$, the category of domain walls between them is 
\begin{align}\label{eq_bimodule}
   &\mathrm{Func}_{\Vec_{G^T}}(\cM(K^T, \omega), \cM(H^T, \psi))\nonumber\\
   &=\Bimod_{\Vec_{G^T}}(\C[K^T]^\omega, \C[H^T]^\psi) \nonumber\\
   &=\bigoplus_{[g]\in K^T\backslash G^T/ H^T}\Rep^{\alpha_{g}}(K^T\cap g H^Tg^{-1})
\end{align}
where $\alpha_g$ is $\omega\cdot g(\psi^{-1})$ restricted to $K^T\cap g H^T g^{-1}$ and $g(\psi^{-1})(g_1,g_2)=\psi(gg_1g^{-1}, gg_2g^{-1})^{-1}$. The derivation of the above formula is completely analogous to the unitary $G$ case, except the various categories are now $\R$-linear.

In the remainder of this section we study several concrete examples of gapped phases with anti-linear symmetries.

\subsection{$\Z_4^T$-phases}
We now consider gapped phases with $\Z_4^T$-symmetry. Here $\Z_4^T=\{T|T^4=1\}$ and the generator $T$ is anti-linear. There are two motivations for considering this example: (a). This is an example that contains all the possibilities of a $G^T$-phase: symmetric trivial, symmetric SPT, partial SSB, and complete SSB. (b). Through examining this example, we will establish a duality of symmetries: 
\begin{theorem}\label{th_Z4T_dual}
    The symmetry categories $\Vec_{\Z_4^T}$, $\Vec_{\Z_2\times \Z_2^T}^\alpha$ are Morita equivalent. 
\end{theorem}
Here $\alpha$ is a nontrivial 3-cocycle given in~\Cref{eq_Z2Z2Tanomaly}. We notice this is similar to the situation with an internal $\Z_4$-symmetry, where it is Morita equivalent to an anomalous $\Z_2\times \Z_2$-symmetry. 

By the discussion above there are in total four gapped $\Z_4^T$-phases. 
    \begin{itemize}
        \item The $\Z_4^T$-symmetric phases. Since $H^2(\Z_4^T, U(1)_T)=\Z_2$, there are two $\Z_2^T$-symmetric gapped phase. We notice that the nontrivial SPT can be thought of as coming from 
        \be
        H^1(\Z_2^T, H^1(\Z_2, U(1)))
        \ee
        of the $E_2$-page of the Lyndon–Hochschild–Serre spectral sequence
$H^p(\Z_2^T, H^q(\Z_2, U(1)))\Rightarrow H^{p+q}(\Z_4^T, U(1))$. One checks that this term survives to a nontrivial element of $H^2(\Z_4^T, U(1))$ by standard spectral sequence calculation.
        \item $\Z_4^T$ is spontaneously broken to $\Z_2=\{1, T^2\}$. Since $H^2(\Z_2,U(1))=0$, there is no further SPT choice. 
        \item The complete $\Z_4^T$-SSB phase.
    \end{itemize}
Next we consider the categories of symmetric defects in each of these four phases. According to the discussion in ~\Cref{sec_setup}, all the categories of symmetric defects should be Morita equivalent to $\Vec_{\Z_4^T}$.

\subsubsection{$\Z_4^T$-symmetric trivial phase}
According to~\Cref{th_RepGT} the category of symmetric defects in any $G^T$-symmetric gapped phase is $\Rep(G^T)$.
Hence the category of symmetric defects in this phase is $\Rep(\Z_4^T)$. Alternatively, this phase is defined by the module category $\cM(\Z_4^T, 1)\simeq\Vec$ over $\Vec_{\Z_4^T}$, and the category of symmetric defects is 
\be
\End_{\Vec_{\Z_4^T}}(\Vec)\simeq \Rep(\Z_4^T)\,.
\ee
This category was discussed in~Example~\Cref{eg_Z4T}.  Recall from the discussion there that there is a unique non-trivial simple object $Q$ of  type satisfying 
\be
Q\otimes Q=4\cdot \lid \,.
\ee
Physically a quaternion type defect is a Kramers doublet under time-reversal, hence the excitations/defects in this phase are Kramers doublets. In particular, the excitations in this phase are \textbf{non-abelian}, which means, for example, there will be a four-fold degeneracy when a pair of $\Z_4^T$-charges is present. 

A generic local operator in this phase takes the form $z\id,~z\in \C$, and the generator $T\in \Z_4^T$ acts on it as $z\id\to z^*\id$. Therefore the only symmetric local operators are $x\id,~x\in \R$, which form a copy of $\R$.
\subsubsection{Nontrivial $\Z_4^T$-SPT}\label{sec_Z4T_SPT}
Since this is a $\Z_4^T$-symmetric gapped phase, by~\Cref{th_RepGT} the category of symmetric defects is $\Rep(\Z_4^T)$. The nontrivial $\Z_4^T$-SPT is characterized by the nontrivial element $\psi\in H^2(\Z_4^T, U(1))\simeq \Z_2$. By the discussion in~\Cref{sec_GT_general}, this means there is a nontrivial twisted algebra 
\be
\C[\Z_4^T]^\psi, \psi \in H^2(\Z_4^T, U(1))
\ee
which is generated as a $\C$-linear space by $1, T, T^2, T^3$, with product rule $z T=Tz^*,~\forall z\in \C$ and $T^4=-1$. For later discussion, we now determine the type of this algebra. The center of the algebra is
\be
Z(\C[\Z_4^T]^\psi)=\R[T^2]
\ee
and since $(T^2)^2=-1$, the center $Z(\C[\Z_4^T]^\psi)$ may be identified with $\C$. Therefore the algebra $\C[\Z_4^T]^\psi$ is in fact a central $\C$-algebra. By the classification of such algebras, we have $\C[\Z_4^T]^\eta\simeq \mathrm{Mat}_n(\C)$ for some non-negative integer $n$. Dimension counting gives $n=2$. We conclude that $\C[\Z_4^T]^\psi=\mathrm{Mat}_2(\C)$. 
\subsubsection{$\Z_4^T$-SSB}
According to \Cref{th_VecGT} the category of symmetric defects in any complete $G^T$-SSB phase is $\Vec_{G^T}$. Hence the category of symmetric defects in this phase is $\Vec_{\Z_4^T}$, whose simple objects are $\Z_4^T$-domain walls. 

Denote by $\cO_k,~k=0,1,2,3$ the local operator that projects onto the $k$th vacuum, and $\cD_T$ the domain wall created by $T\in \Z_4^T$. Then the domain wall acts on local operators as 
\be
\cD_T(z_1\cO_1+z_2\cO_2+z_3\cO_3+z_4\cO_4)=z_4^*\cO_1+z_1^*\cO_2+z_2^*\cO_3+z_3^*\cO_4.
\ee
The only $\Z_4^T$-symmetric local operators are $z\cO_1+z^* \cO_2+z\cO_3+z^*\cO_4$, which form a copy of $\C$. 

\subsubsection{$\Z_4^T \to \Z_2$ partial SSB}

This partial SSB phase is defined by the module category $\Vec_{\Z_4^T/\Z_2}\simeq \Vec_{\Z_2^T}$, or equivalently the algebra $\C[\Z_2]\in \Vec_{\Z_4^T}$ where $\Z_2=\{1, T^2\}<\Z_4^T$. By the general discussion before, the category of symmetric defects in this phase $\Bimod_{\Vec_{\Z_4^T}}(\C[\Z_2])$, is Galois-real. In order to determine the structure of this Galois-real fusion category, we take a shortcut by applying the ``breaking time-reversal" construction in ~\Cref{sec_semi_product}. Since the broken symmetry here is $\Z_4^T/\Z_2=\Z_2^T$, we see that the category of symmetric defects should be of the semi-direct product type discussed in ~\Cref{sec_semi_product}. Following example~\Cref{eg_GT_semi_prod}, we know that the category of symmetric defects takes the form 
\be
\Rep(\Z_2)\rtimes \Z_2^T\,,
\ee
where the $\Z_2^T$-action is the one that produces $\Rep(\Z_4^T)$ upon taking fixed points. In fact there are only two possible ways $\Z_2^T$ can act on $\Rep(\Z_2)$: 
\begin{itemize}
    \item  $T$ acts as a $\C$-anti-linear monoidal equivalence $\Rep(\Z_2)\to \Rep(\Z_2)$. In order to preserve the fusion rule, $T$ has to act trivially on objects, and act as complex conjugation on scalars. 
    \item There is freedom to choose the natural isomorphism $\alpha: T\circ T\Rightarrow \id$. This natural isomorphism is determined by its component at the nontrivial object $q\in \Rep(\Z_2)$, $\eta_q: T^2(q)=q\to q$. This being a monoidal natural transformation means we have condition $\alpha_{q\otimes q}=\eta_1^2=1$, hence $\eta_q=\pm$. 
\end{itemize}
With the choice $\eta_q=+1$, the homotopy fixed point category is $\Rep(\Z_2\times \Z_2^T)$; with the choice $\eta_q=-1$, the homotopy fixed point category  is $\Rep(\Z_4^T)$. Therefore we can describe $\Rep(\Z_4^T)$ as the complex fusion category $\Rep(\Z_2)$ together with a $\Z_2^T$-action that has a nontrivial associator $\eta_q=-\id_q: q\to q$. With this action understood, we can work out the associator of the category $\Rep(\Z_2)\rtimes \Z_2^T$ as follows. Let $(a_j,s_j)\in \Rep(\Z_2)\rtimes \Z_2^T$ where $a_j=1$ or $q\in \Rep(\Z_2)$ and $s_j\in \Z_2^T$, $j=1,2,3$. Then the associator 
\be
\alpha_{(a_1,s_1),(a_2,s_2),(a_3,s_3)} 
\ee
is given as follows. By definition of tensor product in semi-direct product categories(\ref{sec_semi_product}), we have 
\begin{widetext}
    \begin{align*}
        ((a_1,s_1)\otimes(a_2,s_2))\otimes (a_3,s_3)&=(((a_1\otimes s_1(a_2))\otimes s_1s_2(a_3),s_1s_2s_2)\\
        (a_1,s_1)\otimes ((a_2,s_2)\otimes (a_3,s_3))&=((a_1\otimes s_1(a_2\otimes s_2(a_3)),s_1s_2s_3)
    \end{align*}
The associator is then given by the natural moves
\be
((a_1\otimes s_1(a_2))\otimes s_1s_2(a_3)\xrightarrow{(-1)^{s_1s_2a_3}}((a_1\otimes s_1(a_2))\otimes s_1\circ s_2(a_3)\xrightarrow{1} a_1\otimes (s_1(a_2)\otimes s_1\circ s_2(a_3))\xrightarrow{1} a_1\otimes s_1(a_2\otimes s_2(a_3))
\ee
\end{widetext}
The first step produces a $-1$ precisely when $s_1=s_2=T \in \Z_2^T$ and $a_3=q$ is the nontrivial charge, since this step uses the natural transformation $\eta: T\circ T\Rightarrow \id$ and $\eta_q=-1$. The second step is identity because the associator in the $\Rep(\Z_2)$-category is trivial, the third step is identity because the $\Z_2^T$-action has trivial monoidal structure: $T(a\otimes b)\xrightarrow{1} T(a)\otimes T(b)$.  Put everything together we find that the category of symmetric defects in this partial SSB phase is 
\begin{itemize}
    \item A Galois-real fusion category of the form $\Rep(\Z_2)\rtimes \Z_2^T$. 
    \item Write simple objects as $(a,s), a\in \Rep(\Z_2)\simeq \Z_2=\{0,1\},~s\in \Z_2^T\simeq \{0,1\}$. The associator/F-symbol is 
    \be\label{eq_Z2Z2Tanomaly}
    \alpha((a_1,s_2),(a_2,s_2),(a_3,s_3))=(-1)^{s_1s_2a_3} \,.
    \ee
\end{itemize}

We see that this category can be renamed as $\Vec_{\Z_2\times \Z_2^T}^\alpha$ with $\alpha$ given as above. Notice that this cocycle $\alpha$ is a nontrivial one since by inspection it is exactly the generator of 
\be
H^1(\Z_2,H^2(\Z_2^T, U(1)_T))\le H^3(\Z_2\times \Z_2^T,U(1)_T) \,.
\ee
In the language of decorated domain wall construction, the $1\in \Z_2$-domain wall is now decorated with a Kramers doublet. Recall that the category of symmetric defects is always Morita equivalent to the symmetry category. Hence we have proven~\Cref{th_Z4T_dual}.
At the physics level of rigor, this result has appeared in various places. However we believe this is the first time a completely rigorous proof is given.

 \subsubsection{The Categories of Domain Walls and the full 2-Category of Gapped Phases}
Finally we discuss the domain wall categories between the four gapped phases above, and the full 2-category of $\Z_4^T$-gapped phases.

Recall the general discussion in ~\Cref{sec_setup}: if $\cM, \cN$ are two module categories over a symmetry category $\cc$, then the category of domain walls between then is $\mathrm{Func}_\cc(\cM, \cN)$. In particular, the domain wall between any phase defined by $\cM$ and the full SSB phase defined by $\cc$ is 
\be
\mathrm{Func}_\cc(\cc,\cM)\simeq \cM\,.
\ee
Therefore the domain wall categories where one side of the domain wall is the complete SSB phase are easy to determine--they are the module categories themselves. Since $\Vec_{\Z_4^T}$ is Morita equivalent to $\Rep(\Z_4^T)$, we can also ``change basis" and label gapped  $\Z_4^T$-phases as $\Rep(\Z_4^T)$-modules, the regular module $\Rep(\Z_4^T)$ then corresponds to the $\Z_4^T$-symmetric trivial phase. Then we can utilize the classification of $\Rep(G^T)$-modules in~\Cref{sec_module_RepGT} to calculate the domain wall categories where one side of the domain wall is the $\Z_4^T$-symmetric trivial phase--they are just the module categories(over $\Rep(\Z_4^T)$ themselves.  Therefore our strategy is as follows. We first determine all the categories of domain walls where one side of the domain wall is the complete SSB phase using the $\Vec_{\Z_4^T}$-basis, then we switch to $\Rep(\Z_4^T)$-basis and determine all the categories of domain walls where  one side of the domain wall is the $\Z_4^T$-symmetric trivial phase. Finally we will deal with the category of domain walls between the trivial phase and the SPT using formula~\Cref{eq_bimodule}

Using the $\Vec_{\Z_4^T}$-module description of phases, we have 
\begin{itemize}
    \item The domain wall category between the complete $\Z_4^T$-SSB phase and the $\Z_4^T$-trivial phase is $\Vec$. This implies there is a unique simple domain wall between these two phases.
     \item The domain wall category between the complete $\Z_4^T$-SSB phase and the $\Z_4^T$-SPT phase is also $\Vec$. 
    \item The domain wall category between the complete $\Z_4^T$-SSB phase and the $\Z_4^T\to \Z_2$ partial SSB phase is $\Vec_{\Z_2^T}$. This implies there are two simple domain walls between the two phases.
\end{itemize}

We now switch basis to $\Rep(\Z_4^T)$-modules. 
\begin{itemize}
    \item The $\Z_4^T$-SPT is defined by the module $\Rep(\Z_4^T)^\psi=\Mod(\C[\Z_4^T]^\psi)$. We showed in~\Cref{sec_Z4T_SPT} that $\C[\Z_4^T]^\psi\simeq \mathrm{Mat}_2(\C)$ as $\R$-algebras, since $\mathrm{Mat}_2(\C)$ is further Morita equivalent to $\C$ as $\R$-algebras, we see that the underlying category of the module category $\Mod(\C[\Z_4^T]^\psi)$ is simply $\Mod(\C)=\Vec_\C$. This implies the category of domain walls between the $\Z_4^T$-trivial phase and the $\Z_4^T$-SPT is $\Vec$.
    \item The $\Z_4^T\to \Z_2$ partial SSB phase is defined by the module $\Rep(\Z_2)$ over $\Rep(\Z_4^T)$. Hence the domain wall between the $\Z_4^T$-trivial phase and the partial SSB phase is $\Rep(\Z_2)$.
\end{itemize}

Finally we are left with the category of domain walls between the $\Z_4^T$-SPT and the $\Z_4^T\to \Z_2$ partial SSB phase. To determine this category, we apply formula~\Cref{eq_bimodule}. Take $K^T=\Z_4^T,~ K^T=\Z_2$, we see that 
\be
\mathrm{Func}_{\Vec_{\Z_4^T}}(\cM(\Z_4^T,\omega),\cM(\Z_2,1))=\Rep(\Z_4^T\cap \Z_2)=\Rep(\Z_2).
\ee

We may now summarize the full structure of $(1+1)$d gapped $\Z_4^T$-phases with the diagram in figure~\ref{fig:Z4TPhases}.

\subsection{$S_3^T =\Z_3 \rtimes \Z_2^T$}
We now consider a non-abelian group-like anti-linear symmetry $S_3^T$. We take 
\be
S_3=\{r,T| r^3=T^2=TrTr=1\}
\ee
and take the generator $s$ to be anti-linear and the generator $r$ to be linear. This means we can write $S_3^T=\Z_3\rtimes \Z_2^T$ with $T$ acting on $\Z_3$ as inversion.  One reason for considering this example is that it will reveal a surprising dual symmetry that has no counterpart in the $\C$-linear world. Namely we prove the following by analyzing this example.
\begin{theorem}\label{th_S3T_dual}
    The symmetry categories $\Vec_{\Z_3\rtimes \Z_2^T}$,~$\Vec_{\Z_3\times \Z_2^T}$ are Morita equivalent. 
\end{theorem}
Notice that both symmetries are non-anomalous, yet the group structures are different. In particular $\Z_3\times \Z_2^T$ is abelian while $\Z_3\rtimes \Z_2^T$ is non-abelian.  Dualities of this kind are impossible with internal $G$-symmetries. More generally, non-anomalous $\Z_n\times \Z_2^T$-symmetry is dual to non-anomalous $\Z_n\rtimes \Z_2^T$ symmetry, where $T\in \Z_2^T$ acts on $\Z_2$ as inversion.

By the classification in~\Cref{sec_GT_general} there are the following six gapped phases. 
\begin{itemize}
    \item $S_3^T$-symmetric. We have the decomposition 
    \begin{align}
         & H^2(\Z_3\rtimes \Z_2^T, U(1)_T)\nonumber\\
    &=H^2(\Z_2^T, U(1)_T)\oplus H^1(\Z_2^T, H^1(\Z_3, U(1))\oplus H^2(\Z_3,U(1)) \cr 
    &= \Z_2 \,,
    \end{align}
    where only the first term is nonzero, which is the group of pure $\Z_2^T$-SPTs. Therefore there are two $S_3^T$-SPTs, and the nontrivial one is essentially the Haldane chain protected by the subgroup $\Z_2^T$. We label the two phases as $(S_3^T, 1)$ and $(S_3^T, \omega)$ with $\omega$ being the nontrivial 2-cocyel of $S_3^T$.
    \item $S_3^T\to \Z_2^T$ partial SSB. Since $H^2(\Z_2^T, U(1))=\Z_2$, we can further put a Haldane chain onto the vacua. The two phases are denoted as $(\Z_2^T, 1),~(\Z_2^T, \omega)$.
    \item $S_3^T\to \Z_3$ partial SSB. Since $H^2(\Z_3,U(1))=0$, there is no further SPT choices. This phase is denoted as $(\Z_3,1)$.
    \item $S_3^T$ complete SSB, denoted as $(1,1)$.
\end{itemize}
\subsubsection{$S_3^T$-symmetric trivial phase}
The category of symmetric defects is $\Rep(S_3^T)$. Since $S_3=\Z_3\rtimes \Z_2^T$, from the general discussion in~\Cref{eg_A_rtimes_Z2T} we know that $\Rep(S_3^T)=\Vec_{\R, \Z_3}$, all simples are invertible and of real type.  Denote the generating $S_3^T$-charge by $q\in \Rep(S_3^T)$, then $q\otimes q\otimes q=\lid$.  Since $q$ is a real object, the excitations in this phase are time-reversal singlets. 

\subsubsection{Nontrivial $S_3$-SPT}
This is a $S_3^T$-symmetric gapped phase, hence the category of symmetric defects is $\Rep(S_3^T)$ as before. The nontrivial SPT is characterized by the nontrivial twisted algebra 
\be
\C[S_3^T]^\psi,~\psi\in H^2(S_3^T, U(1))
\ee
For later discussion we determine the structure of this $\R$-algebra. Since this SPT is essentially the Haldane chain protected by the subgroup $\Z_2^T$, the product rules of the algebra $\C[S_3^T]^\psi$ are 
\be
r^3=1, T^2=-1, TrT=r^2, zr=rz ,  z T=Tz^*, \forall z\in \C
\ee
Let $\omega_3=e^{2\pi i/3}$, define the following idempotents
\be
e_0=\frac{1+r+r^2}{3},~e_1=\frac{1+\omega_3 r+\omega_3^2 r^2}{3},~e_2=\frac{1+\omega_3^2 r+\omega_3 r}{3}
\ee
One checks $e_i^2=1, e_ie_j=\delta_{ij}e_i, \sum_i e_i=1$. Therefore we have decomposition 
\be
\C[S_3^T]^\psi=\C[S_3^T]^\psi(e_0+e_1+e_2)
\ee
Since $r e_i \propto e_i$, each block is generated by $\C e_i$ and $Te_i$. Since $Te_i=e_i T$, we have $(Te_i)^2=-e_i$. Therefore each block is isomorphism to the algebra $\C[\Z_2^T]^\eta=\bH$. We conclude 
\be
\C[S_3^T]^\psi=\bH\oplus \bH\oplus \bH.
\ee
\subsubsection{$(\Z_3,1)$ partial SSB}
This phase breaks the symmetry down to $\Z_3$. It may be characterized by the module category $\cM(\Z_3, 1)$, whose underlying category is $\Vec_{S_3^T/\Z_3}=\Vec_{\Z_2^T}$. Alternatively, it is described by the algebra 
\be
\C[\Z_3]\in \Vec_{S_3^T}
\ee
Since $\Z_3$ is unitary, this category of symmetric defects is a Galois-real fusion category. To determine its structure, we apply the semi-direct product construction~\Cref{eg_GT_semi_prod}. Since the broken symmetry here is $S_3^T/\Z_3=\Z_2^T$, we know that the category of symmetric defects is 
\be
\Rep(\Z_3)\rtimes \Z_2^T
\ee
where the $\Z_2^T$-action is the one that produces $\Rep(S_3^T)$ upon taking fixed points. According to the discussion in~\Cref{rmk_Z2T_action_semidirect}, this is exactly the trivial action: $T(q)=q$ for $q\in \Rep(\Z_3)$ the elementary charge of $\Z_3$. The natural transformation $T\circ T\Rightarrow \id$ is also trivial. We see that this semi-direct product category is in fact a product
\be
\Rep(\Z_3)\rtimes \Z_2^T=\Vec_{\Z_3\times \Z_2^T}
\ee
where we renamed $\Rep(\Z_3)$ as $\Vec_{\Z_3}$. We have proved~\Cref{th_S3T_dual}. Since $\Rep(G^T)$ is always Morita equivalent to $\Vec_{G^T}$, this immediately implies there is one more Morita equivalent real fusion category $\Rep(\Z_3\times \Z_2^T)=\Rep_\R(\Z_3)$, so we have the following chain of Morita equivalences
\begin{widetext}
    \be \label{eq_S3T_Morita}
\Vec_{\Z_3\rtimes \Z_2^T}\simeq_{\text{Morita}} \Rep(\Z_3\rtimes \Z_2^T)\simeq_{\text{Morita}} \Vec_{\Z_3\times \Z_2}\simeq_{\text{Morita}}\Rep(\Z_3\times \Z_2^T)=\Rep_\R(\Z_3)
\ee
\end{widetext}
By number counting these are all the real fusion categories that are Morita equivalent to $\Vec_{\Z_3\rtimes \Z_2^T}$, in other words, all the real fusion categories that will appear as category of symmetric defects in $(1+1)$d gapped $S_3^T$-phases.
\subsubsection{$(\Z_2^T,1)$ partial SSB+trivial SPT}
This phase breaks the symmetry down to $\Z_2^T$ and the remaining $\Z_2^T$ is in a trivial SPT phase. It may be described by the module category $\cM(\Z_2^T,1)$ over $\Vec_{S_3^T}$, whose underlying category is $\Vec_{S_3^T/\Z_2^T}=\Vec_{\Z_3}$.  Alternatively it is described by the algebra $\C[\Z_2^T]$ in $\Vec_{S_3^T}$. The category of symmetric defects in this phase is therefore 
\be
\Bimod_{\Vec_{S_3^T}}(\C[\Z_2^T]) \,.
\ee
Since $\Z_2^T$ is anti-unitary, according to the general discussion in~\Cref{sec_GT_general}, this real fusion category is $\R$-real. We first determine the underlying category of this real fusion category using formula~\Cref{eq_bimodule}. Take $K^T=H^T=\Z_2^T$, we have double coset decomposition
\be
S_3^T=\{1,T\}\cup \{r, r^2, rT, r^2T \}
\ee
and $\Z_2^T \cap r\Z_2^T r^{-1}=\{1\}$.
Hence 
\be
\Bimod_{\Vec_{S_3^T}}(\C[\Z_2^T])=\Rep(\Z_2^T)\oplus \Rep(1)=\Vec_\R\oplus \Vec_\C
\ee
The only real fusion category in the list~\Cref{eq_S3T_Morita} that has this underlying category is $\Rep(\Z_3\times \Z_2^T)=\Rep_\R(\Z_3)$. Indeed, according to the general discussion in~\Cref{sec_AxZ2} and~\Cref{eq_fusion_rule_AxZ2} the  real fusion $\Rep_\R(\Z_3)$ has two simples $1, Q$ with fusion rule
\be
Q\otimes Q=2 \times \lid \oplus Q
\ee
with $\lid$ being a real object and $Q$ being a complex object. Hence the excitations in this phase are time-reversal doublets.
\subsubsection{$(\Z_2^T,1)$ partial SSB+ Haldane chain}
This phase breaks the symmetry down to $\Z_2^T$ and the remaining $\Z_2^T$ is in the Haldane phase. It is described by the twisted algebra $\C[\Z_2^T]^\eta=\bH\in \Vec_{S_3^T}$. Stacking with SPT does not change the category of symmetric defects, hence 
\be 
\Bimod_{\Vec_{S_3^T}}(\C[\Z_2^T])\simeq \Rep_\R(\Z_3)
\ee
as before.

\subsubsection{$(1,1)$ complete SSB}
The complete SSB phase is given by the regular module $\Vec_{S_3^T}$, and the category of symmetric defects is $\Vec_{S_3^T}$. 
\subsubsection{The Domain Wall Categories for the 2-category of $S_3^T$-Phases}
Lastly we work out the categories of domain walls between the above six phases. The computation is straightforward thanks to the formula~\eqref{eq_bimodule}. We present the result as the table \ref{tab:S3Phases}. 






The chain of Morita equivalences~\Cref{eq_S3T_Morita} can be understood physically as the following phase transitions. 

\begin{center}
    \begin{tikzcd}
    \Rep(S_3^T)\arrow[rr, "\text{Break}~\Z_3"]\arrow[dd, "\text{Break}~\Z_2^T"']& & \Rep_\R(\Z_3)\arrow[dd, "\text{Break}~\Z_2^T"] \\
    &&
    \\    
    \Vec_{\Z_3\times \Z_2^T}\arrow[rr, "\text{Break}~\Z_3"']& & \Vec_{S_3^T}
\end{tikzcd}
\end{center}

\subsection{Gauging Finite Subgroups}

The effect of gauging a normal subgroup $N\lhd G$ for a finite internal $G$-symmetry was calculated in~\cite{Tachikawa:2017gyf, Bhardwaj:2017xup} using the Morita theory of fusion categories. As an application of our categorical framework for anti-linear symmetries, we prove the following  result about gauging subgroups in Appendix~\ref{app_gauge_subgroup}, which generalizes the Morita equivalences~\Cref{th_Z4T_dual} and~\Cref{th_S3T_dual}.
\begin{theorem}\label{thm_gauge_subgroup}
    Let $(G^T, s: G^T\to \Z_2^T)$ be an anti-unitary group-like symmetry, and $N\lhd G_0<G^T$ be a unitary central subgroup. Denote by $K^T:=G^T/N$ the quotient. Then there is Morita equivalence 
    \be
\Vec_{G^T}\simeq_{\text{Morita}} \Vec_{\widehat{N}\rtimes K^T}^\omega
    \ee
    where $G^T/N$ acts on $\widehat{N}:=\Hom(N, U(1))$ as inversion via $s$, and 
    \be
     \omega((k_1,\gamma_1),(k_2,\gamma_2),(k_3,\gamma_3))=\gamma_1(e_2(k_2,k_3))
    \ee

    where $e_2\in H^2(K^T, N)$ is the extension class of the extension $N\to G^T\to K^T$. 
\end{theorem}
Compared with the result of gauging central subgroup for a unitary group, the main new feature here is the semi-direct product $\widehat{N}\rtimes K^T$.  The anomaly $\omega\in H^3(G^T/N\ltimes \widehat{N},U(1)_T)$ lives in the subgroup \[H^2(G^T/N, H^1(\widehat{N},U(1)_T))\simeq H^2(G^T/N, N).\] Notice that, since $G^T/N$ acts both on $\widehat{N}$ and the coefficient $U(1)_T$ as inversion, its action on $H^1(\widehat{N},U(1))\simeq N$ is trivial. Then $\omega$ can be viewed as the extension class $e_2\in H^2(G^T/N,N)$ itself. Hence $\omega$ is nontrivial if and only if $N\to G^T\to G^T/N$ is a non-split extension. 

\section{$\Z_2^T$-Enriched Topological Orders and Their Boundaries}\label{sec_T_SymTFT}

The SymTFT has been a powerful means of classification of phases, gapped and gapless, for categorical symmetries. In this section we discuss how such a SymTFT approach needs to be modified in the presence of time-reversal symmetry to consider $\Z_2^T$-symmetry-enriched SymTFTs instead.

In~\Cref{sec_not_gauging_T} we argue that there is a fundamental obstruction for gauging time-reversal, or any anti-linear symmetry, within the framework of standard quantum theory. This will motivate the use of a $\Z_2^T$-enriched SymTFT for anti-linear symmetries as in \cite{Pace:2025hpb}. 
Then in~\Cref{sec_T_SET} we provide a systematic discussion of  $(2+1)$d $\Z_2^T$-enriched non-chiral topological orders (SETs), including a complete bulk-boundary relation. 
\subsection{T-Enrichment versus Gauging}\label{sec_not_gauging_T}
For a $(d+1)$D internal finite $G$-symmetry, the standard SymTFT paradigm identifies the SymTFT as a $(d+2)$D $G$-gauge theory. Following this principle one tries to  define the SymTFT for an anti-linear symmetry, e.g. the  time-reversal $\Z_2^T$, as the $(d+2)$D ``$\Z_2^T$-gauge theory". However, there is a fundamental obstruction for gauging time-reversal. 
Coupling a field theory to time-reversal background is really putting the theory on an unoriented manifold $M$, whose first Stiefel–Whitney class $w_1(M)$ can be thought of as the time-reversal background field. However, unlike the background gauge field for an internal symmetry, the time-reversal background field $w_1(M)$ is determined by the underlying manifold $M$ and we are not allowed to vary it. Hence the only meaningful way of dynamically gauging time-reversal is to sum over spacetime topologies, which forces us into the realm of quantum gravity. Since we would like to stay in the realm of (topological) quantum field theory, we will not take this approach for the SymTFT for $\Z_2^T$, or any anti-linear symmetry. Instead, we argue that it suffices to take an unoriented TFT, or in condensed matter language, a $\Z_2^T$-enriched topological order, as our SymTFT for an anti-linear symmetry. We will show that for a $\Z_2^T$-enriched non-chiral topological order, its $\Z_2^T$-SSB boundaries are described by Galois-real fusion categories, hence give rise to anti-linear symmetry categories.

\subsection{T-Enriched Non-chiral Topological Order}
\label{sec_T_SET}

In this subsection we will give a self-contained discussion of unoriented 3d TFTs, or equivalently time-reversal enriched (non-chiral) topological orders (TOs). Some of the results here, including the classification of boundary conditions and conversely the reconstruction of the $T$-enriched bulk from a gapped boundary, are new.

We start by reviewing the theory of $(2+1)$D non-chiral topological orders(framed extended TFTs) from the perspective of higher condensation~\cite{Gaiotto:2019xmp}. The advantage of the higher condensation perspective is that passing to the unoriented world($\Z_2^T$-enriched world) is simply a base change $\C\to \R$.

\subsubsection{Review of 3d Non-chiral TOs}

The category of $(0+1)$d gapped phases(with no symmetry) is simply $\Vec$: a $(0+1)$d gapped phase is determined by its (finite dimensional) ground state space, and a ``domain wall" between two gapped phases is simply a linear map between ground state spaces. Then we can build $(1+1)$d gapped phases by placing the $(0+1)$d gapped phases on a chain and couple them. Iterating this process we can build $(n+1)$d gapped phases from $n$d gapped phases by layering and coupling. The mathematical tool describing this iterated coupled layer construction is condensation completion~\cite{Gaiotto:2019xmp}: if $\cc$ is the (symmetric monoidal) $n$-category of $n$d gapped phases, then the $(n+1)$d category of $(n+1)$d gapped phases that can be obtained in this way is $\Sigma \cc:=\mathrm{Kar}(B\cc)$, which is called the condensation completion of $B\cc$. This construction can be iterated to yield $\Sigma^k\cc$. 

Now apply the iterated condensation completion construction to $\Vec$, we obtain $n\Vec:=\Sigma^{n-1}\Vec$, which describes all $n$d gapped phases that can be obtained by the layering and coupling construction. In particular, for $n=3$ this gives us $3\Vec$ which is precisely the 3-category of 
\begin{itemize}
    \item Non-chiral $(2+1)$d topological orders;
    \item Domain walls between them;
    \item Junctions between domain walls;
    \item Local operators on the junctions.
\end{itemize}

Since $2\Vec$ is the 2-category of finite semi-simple ($\C$-linear) categories, $3\Vec=\Sigma (2\Vec)$ can be modeled as the Morita 3-category of 
\begin{itemize}
    \item (Multi-)Fusion categories(over $\C$);
    \item Bimodule categories;
    \item Bimodule functors;
    \item Bimodule natural transformations,
\end{itemize}
which fully agrees with the physical picture above. In particular, fix a fusion category $\cc\in 3\Vec$, the 2-category of bulk operators in the corresponding topological order is 
\be
\End_{3\Vec}(\cc)=\Bimod(\cc),
\ee
and the (braided fusion) 1-category of bulk operators of codimension-2 and higher is 
\be
\Omega \End_{3\Vec}(\cc)=\End_{\cc-\cc}(\cc)=\cZ(\cc)
\ee
exactly the Drinfeld center of $\cc$. A gapped boundary for the topological order defined by $\cc\in 3\Vec$ is a 1-morphism in $\Hom_{3\Vec}(\cc, \Vec)$, which is nothing but a $\cc$-module category. This recovers the Kitaev-Kong module category classification of gapped boundaries~\cite{kitaev2012models, kong2017boundary}\footnote{ The Lagrangian algebra classification~\cite{kong2014anyon} can be obtained by the full-center construction~\cite{davydov2010centre}, which is completely equivalent to the module category classification.}.

\subsubsection{The Theory of Non-chiral 3d $\Z_2^T$-SETs}
We now move on to $\Z_2^T$-SETs. A complete theory of $\Z_2^T$-SET is missing in literature, but for non-chiral topological orders we can apply the iterated condensation construction. The starting point of the iteration is the category $(0+1)$D gapped phases with $\Z_2^T$-symmetry. This is clearly $\Rep(\Z_2^T)\simeq \Vec_\R$. Namely a   $(0+1)$D gapped phase with $\Z_2^T$-symmetry is simply a quantum mechanical system invariant under an anti-unitary $\Z_2^T$-symmetry. The (finite-dimensional) ground state space carries a $\Z_2^T$-action, hence is labeled by a representation of $\Z_2^T$. Apply the condensation machinery we obtain $n\Vec_\R$ which describes a subset of $n$d gapped phases with $\Z_2^T$-symmetry(namely those that can be constructed by iterated layering and coupling). 

In particular for $n=2$ we obtain $2\Vec_\R$ the 2-category of finite semi-simple $\R$-linear categories. This has exactly three simple objects:
\begin{itemize}
    \item $\Vec_\R$, describing the trivial $\Z_2^T$-phase. 
    \item $\Vec_\bH$, describing the Haldane phase.
    \item $\Vec_\C$, describing the $\Z_2^T$-SSB phase.
\end{itemize}
This agrees with our discussion in~\Cref{sec_Z2TPhases}.
Moving on to $n=3$ we obtain $3\Vec_\R$, which physically should describe the 3-category of 
\begin{itemize}
     \item Non-chiral $(2+1)$d topological orders enriched with $\Z_2^T$;
    \item  $\Z_2^T$-symmetric domain walls between them;
    \item $\Z_2^T$-symmetric junctions between domain walls;
    \item  $\Z_2^T$-symmetric local operators on the junctions.
\end{itemize}

Since $2\Vec_\R$ is the 2-category of finite semi-simple $\R$-linear categories, we see that $3\Vec_\R=\Sigma 2\Vec_\R$ is nothing but the 3-category of 

\begin{itemize}
    \item (Multi-)Fusion categories over $\R$;
    \item $\R$-linear bimodule categories;
    \item $\R$-linear bimodule functors;
    \item Bimodule natural transformations.
\end{itemize}
Since the $\Z_2^T$-SETs captured by $3\Vec_\R$ are obtained by coupled layer construction, we expect the underlying topological order to be non-chiral. 
The full bulk-boundary relation of the non-chiral $\Z_2^T$-SETs can then by read off from the description of the 3-category above. Let $\cc\in 3\Vec_\R$ be a real fusion category, then the 2-category of bulk $\Z_2^T$-symmetric operators is 
\be
\End_{3\Vec_\R}(\cc)=\Bimod_\R(\cc)
\ee
where $\Bimod_\R(-)$ stands for $\R$-linear bimodules. The (braided fusion) 1-category of bulk $\Z_2^T$-symmetric operators is then 
\be
\Omega \End_{3\Vec_\R}(\cc)=\End_{\cc-\cc}(\cc)=\cZ(\cc)
\ee

This shows, almost tautologically, that the ``bulk=center" paradigm still holds for $\Z_2^T$-SETs, assuming we work with $\Z_2^T$-symmetric operators throughout.
\begin{remark}
    A subtlety here is that without time-reversal enrichment (or any symmetry), the Drinfeld center description $\cZ(\cc)$ for a non-chiral topological is equivalent to the bimodule 2-category description $\Bimod(\cc)$, because the latter is simply $\Sigma \cZ(\cc)$. However with time-reversal enrichment and work over $\R$, this is no longer true in general. Namely there could be codimension-1 defects in the bulk of a non-chiral $\Z_2^T$-SET that does not arise as a condensation defect of codimension-2 defects. This happens when the $\Z_2^T$-SET is a $\Z_2^T$-SSB order~\footnote{For instance, take $\cc=\Vec_\C$, then $\cZ(\Vec_\C)=\Vec_\C$ while $\Bimod_\R(\Vec_\C)=\End_\R(2\Vec_\C)=2\Vec_\C\oplus 2\Vec_\C$ is not connected. This is the trivial $(2+1)$d $\Z_2^T$-SSB order}.
    As long as the input real fusion category $\cc$ is not a complex fusion category(i.e., either it is $\R$-real, or Galois-real with nonzero $\cc_T$ component), then $\Bimod_\R(\cc)=\Sigma \cZ(\cc)$ remains true~\footnote{Proof of this will appear elsewhere.} and there is no harm in saying ``the bulk is completely determined by the Drinfeld center". Therefore in what follows we consider two types of input real fusion categories for a $\Z_2^T$-SET: either an $\R$-real fusion category, or a Galois-real fusion category $\cc_1\oplus \cc_T$ with nonzero $\cc_T$. We will then simply say the Drinfeld center $\cZ(\cc)$ is the $\Z_2^T$-SET.
\end{remark}

Given a real fusion category $\cc\in 3\Vec_\R$, a $\Z_2^T$-symmetric boundary condition for the $\Z_2^T$-SET defined by $\cc$ is simply a morphism in $\Hom_{3\Vec_\R}(\Vec_\R, \cc)$, which is a module category over $\cc$. Therefore the $\Z_2^T$-symmetric boundaries~\footnote{By a $\Z_2^T$-symmetric boundary we mean one that does not break the enriching $\Z_2^T$ explicitly, but we allow for spontaneously breaking $\Z_2^T$.} are classified by module categories over the input real fusion category. In particular, for every input real fusion category $\cc$, there is a canonical (Dirichlet) boundary given by the regular module $\cc$ itself, and the real fusion category of $\Z_2^T$-symmetric defects on the canonical boundary is 
\be
\End_{\cc}(\cc)=\cC \,.
\ee
This shows that every real fusion category can arise as the category of $\Z_2^T$-symmetric topological defects on a $\Z_2^T$-symmetric boundary of a non-chiral $\Z_2^T$-SET. In particular, this implies every anti-linear symmetry category(Galois-real fusion category) can arise on the boundary of a $\Z_2^T$-SET. This lays the foundation of our later  application of the SymTFT to anti-linear symmetries. 

\begin{remark}
    The above discussion shows that for a real fusion category $\cc$, $\Z_2^T$-symmetric gapped boundaries of $\cZ(\cc)$ are classified by module categories over $\cc$. One may wonder what is the role of anyon condensation in describing these gapped boundaries. Unfortunately anyon condensation, in particular Lagrangian algebras, no longer provide a complete classification of $\Z_2^T$-symmetric gapped boundaries. This already happens when the input real fusion category is trivial, $\Vec_\R$. Then the bulk $\Z_2^T$-SET is given by $\cZ(\Vec_\R)=\Vec_\R$. A $\Z_2^T$-symmetric gapped boundary of the trivial SET is nothing but a 2d $\Z_2^T$-gapped phase. Indeed the module category classification says gapped boundaries are given by modules over $\Vec_\R$, and indeed there are three modules over $\Vec_\R$. They can be described by the algebras $\R, \C, \bH \in \Vec_\R$. However, applying the full-center, we obtain condensable algebras $Z(\R)=Z(\bH)=\R$, which shows from the perspective of bulk anyon condensation, the trivial($\R$) boundary is indistinguishable from the Haldane phase boundary($\bH$). Of course, this comes from the physical fact that there is no string (or local) order parameter for the Haldane phase.
\end{remark}

The $\Z_2^T$-symmetric boundaries of a $\Z_2^T$-SET thus come in two types, $\R$-real and Galois-real. A Galois-real type gapped boundary is really one where the enriching symmetry $\Z_2^T$ spontaneously breaks. 
Indeed, to see the underlying gapped boundary of the underlying topological order, one needs to complexify. Since the scalars in a Galois-real fusion category already form $\C$, further complexify it will lead to scalars forming 
\be
\C\otimes_\R \C=\C\oplus \C \,.
\ee
This implies the underlying gapped boundary is in fact multi-fusion with $\lid=\lid_+\oplus \lid_-$. The two ground states are permuted by $\Z_2^T$-action, hence there is still only one $\Z_2^T$-symmetric ground state. The Galois-real fusion category $\cc=\cc_1\oplus \cc_T$ on the boundary should be viewed as the category of either defects within each vacuum, or domain walls between the two vacua. 
\begin{example}
Consider $\cc=\Vec_\R$, then the bulk $\Z_2^T$-SET is trivial: $\cZ(\Vec_\R)=\Vec_\R$. If we take the module category $\Vec$ over $\Vec_\R$, then the category of $\Z_2^T$-symmetric defects on the boundary defined by $\Vec$ is 
    \be
    \End_{\Vec_\R}(\Vec)=\Vec_{\Z_2^T} \,.
    \ee
To obtain the underlying gapped boundary, we complexify to  
\be     \Vec_{\Z_2^T}\boxtimes_\R \Vec=\End_\C(\Vec\oplus \Vec)=\begin{pmatrix}
     \Vec& \Vec\\
         \Vec& \Vec
     \end{pmatrix} \,,
    \ee
which is now a multi-fusion category.  This means the module category $\Vec$ defines the $\Z_2^T$-SSB boundary, which is not a simple boundary for the underlying trivial topological order, but simple as a $\Z_2^T$-symmetric boundary of the $\Z_2^T$-SET. 
\end{example}
\begin{remark}
Gapped boundaries of $\Z_2^T$-SETs that spontaneously break $\Z_2^T$, such as the one discussed above, were not considered in~\cite{Pace:2025hpb}. However, in order to obtain an anti-linear symmetry category (such as $\Vec_{\Z_2^T}$) as actual category of defects on a boundary, one has to spontaneously break $\Z_2^T$. In our later application to the SymTFT quiche, the symmetry boundary is always Galois-real, hence breaks the enriching $\Z_2^T$ spontaneously. 
\end{remark}

\subsection{Alternative Description of $\Z_2^T$-SETs}
\label{sec_alternative}
Previously we discussed the description of non-chiral 3d $\Z_2^T$-SETs in terms of the Drinfeld center $\cZ(\cc)$ of a real fusion category $\cc\in 3\Vec_\R$. A lighter and equivalent description of a $\Z_2^T$-SET is as an underlying topological order $\cB$ together with an anti-unitary $\Z_2^T$-action. Here by an anti-unitary $\Z_2^T$-action we mean
\begin{itemize}
    \item An \emph{anti-unitary} braided equivalence , $\rho_T: \cB\to \cB$.
    \item A braided natural isomorphism $\eta: \rho_T\circ \rho_T\Rightarrow \id$. 
    \item Associativity of $\eta$.
\end{itemize}
In order for the action to define a $\Z_2^T$-SET we also require certain anomaly class of the action 
$\cO(\rho_T)\in H^4(\Z_2^T, U(1)_T)$ is trivial. Notice that for a $G$-SET with internal $G$-symmetry, the data involves not only a non-anomalous $G$-action but also a choice in a torsor for $H^3(G, U(1))$. Thanks to $H^3(\Z_2^T, U(1)_T)=0$, we can take $\Z_2^T$-SET to be simply a topological order with a non-anomalous $\Z_2^T$-action. 

If the input real fusion category is Galois-real $\cc=\cc_1\oplus \cc_T$, then the underlying topological order of $\cZ(\cc)$ is simply $\cZ(\cc_1)$(see~\Cref{sec_center_of_gr} below) and there is a straightfoward way of obtaining the $\Z_2^T$-action on the topological order $\cZ(\cc_1)$ from the input $\cc_1\oplus \cc_T$. See~\Cref{fig:DTaction}.  The component $\cc_T$ is canonically a $\cc_1-\cc_1^*$ bimodule via the monoidal structure of $\cc_1\oplus \cc_T$ . A standard construction in fusion category theory is that for two fusion categories $\cc, \cD$, an invertible $\cc-\cD$ bimodule $\cM$ induces a braided equivalence on the Drinfeld center $\alpha_\cM:\cZ(\cc)\to \cZ(\cD)$. Applying this construction to $\cc_T$, we obtain a ($\C$-linear) braided equivalence $\alpha_{\cc_T}: \cZ(\cc_1)\to \cZ(\cc_1^*)=\cZ(\cc_1)^*$. This is the same as an \emph{anti-linear} braided automorphism of $\cZ(\cc_1)$. Furthermore, the monoidal structure of $\cc_1\oplus \cc_T$ leads to an equivalence of $\cc_1$-bimodules $\cc_T \boxtimes_{\cc_1} \cc_T\simeq \cc_1$, which translates to a braided natural isomorphism $\alpha_{\cc_T}\circ \alpha_{\cc_T}\Rightarrow \id$. Finally, the coherence of tensor product in $\cc_1\oplus \cc_T$ implies that the above natural isomorphism is associative. This shows the Galois-real fusion category $\cc_1\oplus \cc_T$ defines a canonical anti-linear $\Z_2^T$-action on the topological order $\cZ(\cc_1)$.  

\begin{figure}
$$
\begin{tikzpicture}
 \begin{scope}[shift={(0,0)}] 
\draw [black]
(0,0) -- (0,4) -- (2,5) -- (5,5) -- (5,1) -- (3,0)--(0,0);
\draw [black, thick,  fill=white,opacity=1]
(0,0) -- (0, 4) -- (2, 5) -- (2,1) -- (0,0);
\draw [thick, fill=black, opacity=0.2]
(0,0) -- (0, 4) -- (2, 5) -- (2,1) -- (0,0);
\draw [cyan, thick, fill=cyan, opacity=0.5]
(0,2) -- (3, 2) -- (5, 3) -- (2, 3) -- (0,2);
\draw [blue, line width = 0.05cm] (0,2) -- (2, 3)  ; 
\node at  (3,3.5)  {$\fZ(\cc_1)$} ;
\node at  (3,1.5)  {$\fZ(\cc_1^*)$} ;
\node at (4, 3.2) {$a$}; 
\draw [->, thick] (4,2.9) -- (4,1.8);
\node at (4, 1.6) {$\rho_T(a)$}; 
\draw[dashed] (0,0) -- (3,0);
\draw[dashed] (0,4) -- (3,4);
\draw[dashed] (2,5) -- (5,5);
\draw[dashed] (2,1) -- (5,1);
\draw [black, dashed]
(3,0) -- (3, 4) -- (5, 5) -- (5,1) -- (3,0);
 \node[above] at (1.1, 3) {$\cC_1$};
 \node at (1.7,3.5){$|\uparrow\rangle$};
 \node[above] at (2.3, 2.2) {$D_T$};
 \node[above] at (1.1, 1) {$\cc_1^*$};
 \node at (1.7, 1.5){$|\downarrow\rangle$};
 \node[blue, left] at (0, 2) {$\cC_T$};
\end{scope}
\end{tikzpicture}
$$
\caption{T-enriched SymTFT quiche for the Galois-real category $\cC_1 \oplus \cc_T$: The fusion category $\cc_1$ ($\cc_1^*$) is the canonical boundary of the Levin-Wen topological order $\fZ(\cc_1)$ ($\fZ(\cc_1^*)$). The $\Z_2^T$ (anti-unitary) automorphism defect in the bulk of the SymTFT ends in $\cc_T$,  which is an invertible $\cc_1$-$\cc_1^*$ bimodule. The boundary of the $D_T$ surface is a domain wall between two vacua on the boundary (denoted $|\uparrow\rangle$ and $|\downarrow\rangle$).}\label{fig:DTaction}
\end{figure}

The fusion category $\cc_1$($\cc_1^*$) is the canonical boundary of the Levin-Wen topological order $\cZ(\cc_1)$($\cZ(\cc_1^*)$), and a domain wall between $\cc_1$-region and $\cc_1^*$-region on the boundary canonically extends to a domain wall $D_T$ in the bulk between $\cZ(\cc_1)$ and $\cZ(\cc_1^*)$. 
The bulk domain wall $D_T$ is invertible precisely when the category of defects on the boundary wall, $\cc_T$, is an invertible $\cc-\cc^*$ bimodule. When $\cc_T$ is invertible, the bulk domain-wall $D_T$ implements a braided equivalence on the Drinfeld center: $\rho_T: \cZ(\cc)\to \cZ(\cc^*)$ by transporting anyons across the wall. 
The fact that $\cc_T$ has order-2 as a bimodule category($\cc_T \boxtimes_{\cc_1}\cc_T\simeq \cc_1$) implies that the bulk invertible wall $D_T$ has order-2: $D_T\otimes D_T=1$. 
The natural transformation $\rho_T\circ \rho_T\Rightarrow \id$ can be understood as \emph{symmetry fractionalization} for the $D_T$ surface. Namely, fusing two $D_T$ surface may leave behind a nontrivial defect on the $\cc_1$-boundary, see ~\Cref{fig:SymTFTFrac}.

\begin{center}
\begin{figure*}[t]
$$
  \begin{tikzpicture}
 \begin{scope}[shift={(0,0)}] 
\draw [black]
(0,0) -- (0,4) -- (2,5) -- (5,5) -- (5,1) -- (3,0)--(0,0);
\draw [black, thick,  fill=white,opacity=1]
(0,0) -- (0, 4) -- (2, 5) -- (2,1) -- (0,0);
\draw [thick, fill=black, opacity=0.2]
(0,0) -- (0, 4) -- (2, 5) -- (2,1) -- (0,0);
\draw [cyan, thick, fill=cyan, opacity=0.5]
(0,2.5) -- (3, 2.5) -- (5, 3.5) -- (2, 3.5) -- (0,2.5);
\draw [cyan, thick, fill=cyan, opacity=0.5]
(0,1.5) -- (3, 1.5) -- (5, 2.5) -- (2, 2.5) -- (0,1.5);
\draw [blue, line width = 0.05cm] (0,2.5) -- (2, 3.5)  ;
\draw [blue, line width = 0.05cm] (0,1.5) -- (2, 2.5)  ;
\draw[dashed] (0,0) -- (3,0);
\draw[dashed] (0,4) -- (3,4);
\draw[dashed] (2,5) -- (5,5);
\draw[dashed] (2,1) -- (5,1);
\draw [black, dashed]
(3,0) -- (3, 4) -- (5, 5) -- (5,1) -- (3,0);
 \node[above] at (1.1, 3.2) {$\cC_1$};
  \node[above] at (1.1, 2.1) {$\cc_1^*$};
  \node[above] at (1.1, 1) {$\cc_1$};
 \node[above] at (2.3, 2.7) {$D_T$};
  \node[above] at (2.3, 1.7) {$D_T$};
 \node[blue, left] at (0, 2.5) {$x\in\cC_T$};
  \node[blue, left] at (0, 1.5) {$y\in \cC_T^*$};
  \node at (6.5,2.5){$\Longrightarrow$};
\end{scope}
\end{tikzpicture}
     \begin{tikzpicture}
\begin{scope}[shift={(2,0)}]
\draw[black]
(0,0) -- (0,4) -- (2,5) -- (5,5) -- (5,1) -- (3,0) -- (0,0);
\draw[black, thick, fill=white, opacity=1]
(0,0) -- (0,4) -- (2,5) -- (2,1) -- (0,0);
\draw[thick, fill=black, opacity=0.2]
(0,0) -- (0,4) -- (2,5) -- (2,1) -- (0,0);
\draw[dashed] (0,0) -- (3,0);
\draw[dashed] (0,4) -- (3,4);
\draw[dashed] (2,5) -- (5,5);
\draw[dashed] (2,1) -- (5,1);
\draw[black, dashed]
(3,0) -- (3,4) -- (5,5) -- (5,1) -- (3,0);
\draw[red, line width=0.04cm] (0,2) -- (2,3);
\node at (1.0,2) {$\mathcal{C}_1$};
\node at (3.65,2.65) {$\fZ(\mathcal{C}_1)$};
\node[red, left] at (0,2) {$a\in \cc_1$};
\end{scope}
\end{tikzpicture}
$$
\caption{SymTFT picture showing the fractionalization for the $D_T$ surface: fusing the two $D_T$ surfaces, can leave behind on the the symmetry boundary a non-trivial line $a \in \cC_1$ }\label{fig:SymTFTFrac}
\end{figure*}
\end{center}

\noindent
{\bf Symmetry Fractionalization.}
We also note that for an abelian topological order $\cB$, the symmetry fractionalization for a $\Z_2^T$-action can be recorded as an assignment:
\be
\eta_a=\pm \,,
\ee
for every anyon $a\in \cB$ that is invariant under the $T$-action, $T(a)=a$. The assignment is required to satisfy $\eta_a\eta_b=\eta_{ab}$. Whenever $\eta_a=-1$, we say the anyon $a$ is a Kramers doublet, and denote the symmetry fractionalization pattern as $a_T$. Physically the phase $\eta_a$ is the phase obtained by local action of $T^2$ on $a$. 

The argument above can be applied in reverse: any non-anomalous $\Z_2^T$-action on the bulk topological order $\cZ(\cc_1)$ give rise to an anti-linear component $\cc_T$ such that $\cc_1\oplus \cc_T$ is a Galois-real fusion category extending $\cc_1$. The anti-linear component $\cc_T$ is the category of defects living on the boundary of the bulk automorphism domain wall $D_T$.

\begin{example}
    Let us consider the toric code topological order $\cZ(\Vec_{\Z_2})$ and we fix the canonical Dirichlet boundary condition with category of defects $\Vec_{\Z_2}$. According to the above discussion, there should be a 1-1 correspondence between extensions of $\Vec_{\Z_2}$ to a Galois-real fusion category and non-anomalous $\Z_2^T$-actions on the toric code.
    Indeed we can identify the following extensions of $\Vec_{\Z_2}$:
    \begin{enumerate}
        \item $\Vec_{\Z_2\times \Z_2^T}$
        \item $\Vec_{\Z_4^T}$
        \item $\Vec_{\Z_2\times \Z_2^T}^\omega$
        \item $\TY_{\overline{\C}}(\Z_2)$ (see \Cref{sec:TY}).
    \end{enumerate}
    Accordingly, we identify the following $\Z_2^T$-SET structures on the toric code:
    \begin{enumerate}
        \item This is the trivial SET, where $T$ acts as identity on the anyons, and there is no symmetry fractionalization. 
        \item This is the $e_Tm$ SET, which means $\eta_e=-1$, i.e., the $e$ anyon is a Kramers doublet. Equivalently, acting with $T^2$ on an open region leaves behind a loop of $m$ on the boundary. 
        \item This is the $em_T$ SET, i.e. $\eta_m=-1$, and the $e$ anyon is a Kramers doublet. Equivalently, acting with $T^2$ on an open region leaves behind a loop of $e$ on the boundary.
        \item The $e-m$-exchange SET, where $T: e\leftrightarrow m$ as the $e-m$ duality symmetry.
    \end{enumerate}
Hence we seem to have indeed a 1-1 correspondence. In the next section we will perform a more careful analysis to confirm that the 1-1 correspondence agrees with the physical picture we discussed above. Notice there is another $\Z_2^T$-action: 
$e_Tm_T$ where both $e$ and $m$ are Kramers doublets. However it is well-known this symmetry fractionalization pattern is anomalous~\cite{Barkeshli_2019} thus does not define a genuine SET. 
Hence it does not correspond to any Galois-real fusion category extending $\Vec_{\Z_2}$.
\end{example}

\subsection{The Drinfeld Center of Galois-real Fusion Categories}\label{sec_center_of_gr}
In this subsection we prove a technical result confirming the equivalence of the two descriptions of non-chiral $\Z_2^T$-SETs discussed above. We show that, given any Galois-real fusion category $\cc=\cc_1\oplus \cc_T$, the $\Z_2^T$-SET $\cZ(\cc)$ has underlying topological order $\cZ(\cc_1)$, with a $\Z_2^T$-action determined by $\cc_T$ in the way described in the previous subsection.

Take an anti-linear symmetry category $\cc=\cc_1\oplus \cc_T$. According to~\cite{sanford2025fusion} the Drinfeld center $\cZ(\cc_1\oplus \cc_T)$ is a modular tensor category(MTC) over $\R$ with 
\be
\End_{\cZ(\cc)}(\lid)=\R\,,
\ee
whenever $\cc_T\neq 0$. Being a real MTC, $\cZ(\cc)$ can be viewed as the fixed point category of a complex MTC with respect to a $\Z_2^T$-action. In other words $\cZ(\cc)$ can be viewed as an ordianry topological order $\cZ(\cc)^\C$(complexification), together with certain $\Z_2^T$-action. It turns out $\cZ(\cc)^\C$ is exactly $\cZ(\cc_1)$ the SymTFT for the internal symmetry $\cc_1$, and the $\Z_2^T$-action is exactly the one described in the previous subsection(the one induced by the bimodule $\cc_T$).  Hence the Drinfeld center $\cZ(\cc_1\oplus \cc_T)$ contains the full information of the underlying topological order $\cZ(\cc_1)$ and the anti-linear $\Z_2^T$-action.

This is summarized as the following technical theorem. 
\begin{theorem}\label{th_center}
  Let $\cc=\cc_1\oplus \cc_T$ be a Galois-real fusion category. Then $\cZ(\cc)$ is the (homotopy) fixed point category of the $\Z_2^T$-action on $\cZ(\cc_1)$ defined by the invertible $\cc_1-\cc_1^*$ bimodule $\cc_T$. 
  \proof
  We need to recall some fusion category theory~\cite{etingof2015tensor}. If $\cE,\cD$ are two fusion categories, then an invertible $\cE-\cD$ bimodule $\cM$ induces a braided equivalence $\alpha_\cM:\cZ(\cE)\simeq \cZ(\cD)$ defined as follows.  Let $L_\cM: \cZ(\cE)\to \End_{\cE-\cD}(\cM)$ be the functor defined as $L_\cM((x,\gamma))(m)=x\rhd m$(with $\cC$-module structure provided by the half-braid $\gamma$), and $R_\cM: \cZ(\cD)\to \End_{\cE-\cD}(\cM)$ be the functor defined bas $R_\cM((x,\gamma))(m)=m\lhd x$ (with $\cD$-module structure  provided by the half-braid $\gamma$). Then $L_\cM, R_\cM$ are tensor equivalences and the composition $\alpha_\cM:=R_\cM^{-1}\circ L_\cM: \cZ(\cE)\to \cZ(\cD)$ is the braided equivalence induced by $\cM$. In other words, the braided equivalence is uniquely determined by the relation $x\rhd m\simeq m\lhd \alpha_\cM(x),~\forall (x,\gamma)\in \cZ(\cE), m\in \cM $. 
  
  Now let $(x,\gamma)\in \cZ(\cc)$ be an object where $x\in \cc=\cc_1\oplus \cc_T$ and $\gamma: x\otimes y\simeq y\otimes x,~\forall y\in \cc$ is the half-braid. According to~\cite{sanford2025fusion}, the object $x$ is in fact strictly in the linear component $\cc_1$.  Denote by the restriction of $\gamma$ to $\cc_1/\cc_T$ by $\gamma_1/\gamma_T$. Then $(x,\gamma_1)$ can be viewed as an object in $\cZ(\cc_1)$. Then $\gamma_T: x\otimes y\simeq y\otimes x,~\forall y\in \cc_T$ can be viewed as an equivalence of bimodule functors $L_{\cc_T}(x,\gamma_1)\simeq R_{\cc_T}(x,\gamma_1)$. Since we have $L_{\cc_T}\simeq R_{\cc_T}\circ \alpha_{\cc_T}$, the information carried by $\gamma_T$ is the same as an equivalence $R_{\cc_T}\circ \alpha_{\cc_T}(x,\gamma_1)\simeq R_{\cc_T}(x,\gamma_1)$, which is then the same as an isomorphism $\alpha_{\cc_T}(x,\gamma_1)\simeq (x,\gamma_1)$. Therefore an object $(x,\gamma)\in \cZ(\cc)$ is the same as an object $(x,\gamma_1)\in \cZ(\cc_1)$ together with an equivariant structure $\alpha_{\cc_T}(x,\gamma_1)\simeq (x,\gamma_1)$. In other words $(x,\gamma)$ is the same as an object in the homotopy fixed point category $\cZ(\cc_1)^{h\Z_2^T}$ with respect to the $\Z_2^T$-action defined by $\cc_T$. \qed
\end{theorem}


\section{SymTFT Quiches for Anti-Linear Symmetries}\label{sec_quiche}

We now apply the theory of non-chiral $\Z_2^T$-SETs developed in the previous section to formulate the $\mathbb{Z}_2^T$-enriched SymTFT construction for anti-linear symmetries. We begin by explaining the general perspective: how to derive the anti-linear symmetry category from the underlying SymTFT $\mathcal{Z}(\mathcal{C}_1)$ for the linear sector $\mathcal{C}_1$, together with the $\Z_2^T$-enrichment data. We then examine several examples, including $\Vec_{\Z_4^T}$ and its Morita dual $\Vec_{\Z_2\times \Z_2^T}^\alpha$, $\Vec_{A\times \Z_2^T}$ and its Morita dual $\Vec_{A\rtimes \Z_2^T}$, as well as Galois-real Tambara–Yamagami categories.

In this work, we focus on the SymTFT quiche, i.e. the SymTFT with  a single, gapped, symmetry boundary. One difficulty of developing a full sandwich picture for gapped phases with general anti-linear symmetries is the presence of spontaneous time-reversal breaking on the symmetry boundary. As also noticed in~\cite{Pace:2025hpb}, the symmetry on a $T$-symmetric gapped boundary always takes the form $\cc\rtimes \Z_2^T$ \Cref{sec_semi_product}, which from our perspective is too restrictive as a symmetry boundary. We will present a self-consistent description of all Morita equivalent symmetry categories, from the point of view of the $\Z_2^T$-enriched SymTFT. 
However in order to "close the sandwiches" this requires some more detailed analysis of the $\Z_2^T$-SSBing that happens in this framework. 

To summarize: the focus here is on the realization of the symmetry defects on the boundary of the $\Z_2^T$-enriched SymTFT quiches, and to prove gauge-equivalences between symmetry categories. We will comment on how to utlize this to reproduce the classification of gapped phases we presented in \Cref{sec_GappedPhases} in the conclusions.

\subsection{T-Enriched SymTFT Quiches}

Let us recall what a SymTFT quiche comprises: 
we start with a symmetry category $\cC$ and construct the SymTFT, which realizes this category on a gapped boundary condition -- without yet including a physical boundary, schematically:
\be\label{ComplexQuiche}
\begin{tikzpicture}
\begin{scope}[shift={(0,0)}]
\draw [teal,  fill=teal, opacity =0.7] 
(0,0) -- (0,2) -- (2,2) -- (2,0) -- (0,0) ; 
\draw [white] (0,0) -- (0,2) -- (2,2) -- (2,0) -- (0,0)  ; 
\draw [very thick] (0,0) -- (0,2) ;
\node at  (1,1)  {$\fZ(\cC)$} ;
\node[above] at (0,2) {$\Bsym_{\cC}$}; 
\end{scope}
\end{tikzpicture}
\ee
This in particular allows a systematic characterization of all symmetry categories $\cC$ and $\cC'$, that share the same SymTFT, and are thus Morita equivalent, i.e. related by gauging, and thus share their centers $\cZ (\cC) \simeq \cZ (\cC')$. 

For an anti-linear symmetry $\cc=\cc_1\oplus \cc_T$, the SymTFT is now a $\Z_2^T$-enriched topological order, with a gapped boundary whose defects -- both orientation preserving as well reversing --  together form $\cc$. We will denote by $\fZ(\cC)^{\rho}$ the $\Z_2^T$-enriched SymTFT, where $\fZ(\cc)$ is the underlying SymTFT for the linear fusion categorical symmetry $\cc$ and $\rho$ is the $\Z_2^T$-enrichment data. We take a $\Z_2^T$-SET as a topological order with $\Z_2^T$-action, as explained in~\Cref{sec_alternative}, this is equivalent to the formulation of $\Z_2^T$-SET from the $3\Vec_\R$ perspective. 
The SymTFT $\fZ(\cc)$ has the standard Dirichlet boundary with category of defects $\cc$. The $\Z_2^T$-enrichment data provides an invertible domain wall $D_T$ in the interior of the SymTFT, see  \Cref{fig:SymTFTQuiche}.

\begin{figure}
$$
       \begin{tikzpicture}
 \begin{scope}[shift={(0,0)}] 
\draw [black]
(0,0) -- (0,4) -- (2,5) -- (5,5) -- (5,1) -- (3,0)--(0,0);
\draw [black, thick,  fill=white,opacity=1]
(0,0) -- (0, 4) -- (2, 5) -- (2,1) -- (0,0);
\draw [thick, fill=black, opacity=0.2]
(0,0) -- (0, 4) -- (2, 5) -- (2,1) -- (0,0);
\draw [cyan, thick, fill=cyan, opacity=0.5]
(0,2) -- (3, 2) -- (5, 3) -- (2, 3) -- (0,2);
\draw [blue, line width = 0.05cm] (0,2) -- (2, 3)  ; 
\node at (1.2,5) {$\Bsym$};
\node at  (2.8,3.5)  {$\fZ(\cC_1)^{\rho}$} ;
\node at  (2.8,1.5)  {$\fZ(\cC_1^*)^{\rho}$} ;
\draw[dashed] (0,0) -- (3,0);
\draw[dashed] (0,4) -- (3,4);
\draw[dashed] (2,5) -- (5,5);
\draw[dashed] (2,1) -- (5,1);
\draw [black, dashed]
(3,0) -- (3, 4) -- (5, 5) -- (5,1) -- (3,0);
 \node[above] at (1.1, 3) {$\cC_1$};
 \node[above] at (2.3, 2.2) {$D_T$};
 \node[above] at (1.1, 1) {$\cC_1^*$};
 \node[blue, left] at (0, 2) {$\cC_T$};
\end{scope}
\end{tikzpicture}
$$
\caption{SymTFT quiche with symmetry boundary and the bulk as a T-enriched SymTFT, specified by the $\Z_2^T$-action of $\rho$ on the SymTFT of $\cC$. \label{fig:SymTFTQuiche}}
\end{figure}

The surface $D_T$ then ends on the symmetry boundary as a nontrivial defect between $\cc_1$ and $\cc_1^*$, with potentially multiple boundary conditions. It is always understood that the boundary has two vacua permuted by the $D_T$-action. If $\cc_1$ is the category of topological defects in one vacua, while $\cc_1^*$ is the category of topological defects in the other vacua. The collection of all boundary conditions for the $D_T$ surface form a $\cc_1-\cc_1^*$ bimodule category $\cc_T$. Fix any boundary condition $T\in \cc_T$, other boundary conditions are always obtained by fusing with orientation-preserving defects in $\cc_1$ or $\cc_1^*$. We will denote by $a\rhd x, ~a\in \cc_1,~ x\in \cc_T$ fusing with $a$ from above , and $x\lhd a$ for fusing with $a$ from below. A defect in $\cc_1$(or $\cc_1^*$) will be called a linear defect, and a domain wall between the two vacua on the boundary(where $D_T$ ends) is called an anti-linear defect.

Collecting linear and anti-linear defects, we obtain the direct sum category 
\be 
\cc:=\cc_1\oplus \cc_T \,.
\ee
The coherence of the bulk $\Z_2^T$-action provides $\cc$ a tensor structure, making it into a Galois-real fusion category. For instance, since the bulk surface has order 2, $D_T\otimes D_T=\id$, fusing any orientation reversing defects $x, y\in \cc_T$ gives an orientation preserving one $x\otimes y\in \cc_1$. The associative of the $\Z_2^T$-action implies the tensor product is associative, and the non-anomalous condition for the $\Z_2^T$-action makes sure the (modified) pentagon equation is satisfied. 

We emphasis that the symmetry category in our setup is \emph{determined} by the Dirichlet boundary $\cc_1$ of the usual SymTFT $\fZ(\cc_1)$, and a given $\Z_2^T$-action on the SymTFT $\fZ(\cc_1)$. It is also possible to specify the symmetry category first, and then derive the $\Z_2^T$-action on $\fZ(\cc_1)$ (following the proof of~\Cref{th_center}). We find the former approach more intuitive and easier to implement in practice.

\subsection{T-Enriched Quiche for $\Z_2^T$}

We now consider the T-enriched quiche for pure time-reversal symmetry $\Z_2^T$-symmetry, i.e., the symmetry category $\Vec_{\Z_2^T}$. According to our general discussion above, the underlying  SymTFT is the trivial topological order $\Vec$ with canonical boundary $\Vec$. There is only one way of enriching the trivial topological order, namely the $\Z_2^T$-action on $\Vec$ is 
\begin{itemize}
    \item $\rho_T: \Vec\to \Vec^*$ is the functor that is identity on objects and complex conjugation on hom-spaces. 
    \item The natural transformation $\rho_T\circ \rho_T\Rightarrow \id$ is the identity $\id_{\id}$.
\end{itemize}
According to the general prescription discussed above, the symmetry category consists of linear defects on the canonical boundary of $\Vec$, and anti-linear defects that live on the boundary of the bulk $D_T$-surface defined by the $T$-action. In the current case there is a unique linear symmetry defect $\lid\in \Vec$, and a unique anti-linear symmetry defect $T$. Together they generate the symmetry category $\Vec_{\Z_2^T}$: 
\be
       \begin{tikzpicture}
 \begin{scope}[shift={(0,0)}] 
\draw [black]
(0,0) -- (0,4) -- (2,5) -- (5,5) -- (5,1) -- (3,0)--(0,0);
\draw [black, thick,  fill=white,opacity=1]
(0,0) -- (0, 4) -- (2, 5) -- (2,1) -- (0,0);
\draw [thick, fill=black, opacity=0.2]
(0,0) -- (0, 4) -- (2, 5) -- (2,1) -- (0,0);
\draw [cyan, thick, fill=cyan, opacity=0.5]
(0,2) -- (3, 2) -- (5, 3) -- (2, 3) -- (0,2);
\draw [blue, line width = 0.05cm] (0,2) -- (2, 3)  ; 
\node at (1.2,5) {$\Bsym$};
\node at  (2.8,3.5)  {$\fZ(\Vec)^{\rho}$} ;
\node at  (2.8,1.5)  {$\fZ(\Vec^*)^{\rho}$} ;
\draw[dashed] (0,0) -- (3,0);
\draw[dashed] (0,4) -- (3,4);
\draw[dashed] (2,5) -- (5,5);
\draw[dashed] (2,1) -- (5,1);
\draw [black, dashed]
(3,0) -- (3, 4) -- (5, 5) -- (5,1) -- (3,0);
 \node[above] at (1.1, 3) {$\Vec$};
 \node[above] at (2.4, 2.2) {$D_T$};
 \node[above] at (1.1, 1) {$\Vec^*$};
 \node[blue, left] at (0, 2) {$\Vec$};
\end{scope}
\end{tikzpicture}
\ee
The boundary $T$ of the bulk $D_T$ surface is a domain wall between the two (trivial) vacua on the boundary. 
\subsection{T-Enriched Quiche for $\Vec_{\Z_4^T}$ and $\Vec_{\Z_2\times \Z_2^T}^\omega$}
We now construct SymTFT quiches for two dual symmetry categories: $\Vec_{\Z_4^T}$ and $\Vec_{\Z_2\times\Z_2^T}^\alpha$ with the anomaly $\alpha$ given in~\Cref{eq_Z2Z2Tanomaly}. Previously we have shown they are dual by analyzing the category of symmetric defects in different $\Z_4^T$-phases. The SymTFT discussion here provides a simple proof: the two symmetry categories are just different boundary conditions of the same $\Z_2^T$-enriched SymTFT. 

We will start by constructing the quiche for $\Vec_{\Z_4^T}$. Here the linear symmetries form the symmetry category $\Vec_{\Z_2}$. Hence the SymTFT for $\Z_4^T$ will be a $\Z_2^T$-enriched version of the topological order $\fZ(\Vec_{\Z_2})$, i.e., the toric code. We will follow standard convention that the simple anyons in the theory $\fZ(\Vec_{\Z_2})$ are denoted as $1, e,m, f$. According to the general prescription, any $\Z_2^T$-enrichment on the toric code automatically leads to an anti-linear extension of $\Vec_{\Z_2}$. Here in order to realize $\Vec_{\Z_4^T}$ we will consider the so-called $e_Tm$-SET. This means our $\Z_2^T$-action $\rho$ is defined as follows. 
\begin{itemize}
    \item $\rho_T: \cZ(\Vec_{\Z_2})\to \cZ(\Vec_{\Z_2}^*)$ is identity on objects and complex conjugation on hom-spaces. 
    \item The natural transformation $\eta: \rho_T\circ \rho_T\Rightarrow \id$ has components $\eta_e=-1, \eta_m=1, \eta_f=-1$. Physically this means the $e$-anyon is now a Kramers doublet.
\end{itemize}

Following the general prescription, the SymTFT quiche starts with the ordinary SymTFT $\fZ(\Vec_{\Z_2})$ in the bulk with the canonical $\Vec_{\Z_2}$ boundary on the symmetry boundary, which we take to be the $e$-condensed boundary. The anti-linear component of the symmetry category is then determined by the above $\Z_2^T$-enrichment data. The bulk $D_T$ surface acts trivially on anyons. Let $T$ be any given boundary condition of the $D_T$-surface on $\fB^\sym$. This boundary condition gives an object $T\in \cc_T$. Other boundary conditions are then obtained by fusing with linear symmetry defects in $\Vec_{\Z_2}$. Take the nontrivial object $m\in \Vec_{\Z_2}$, we obtain another anti-linear defect $mT:=m\otimes T$: 
\be
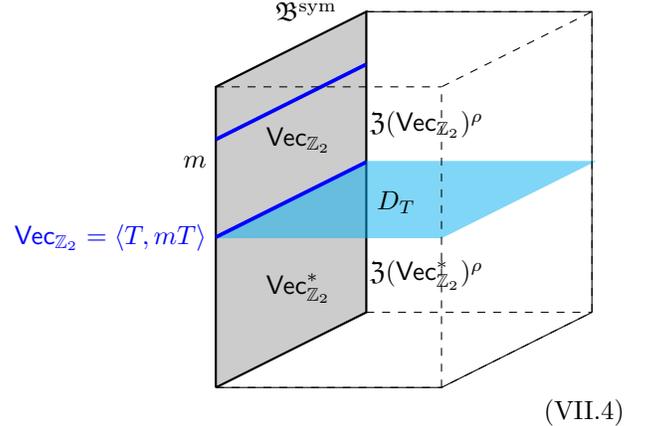
       \begin{tikzpicture}
 \begin{scope}[shift={(0,0)}] 
\draw [black]
(0,0) -- (0,4) -- (2,5) -- (5,5) -- (5,1) -- (3,0)--(0,0);
\draw [black, thick,  fill=white,opacity=1]
(0,0) -- (0, 4) -- (2, 5) -- (2,1) -- (0,0);
\draw [thick, fill=black, opacity=0.2]
(0,0) -- (0, 4) -- (2, 5) -- (2,1) -- (0,0);
\draw [cyan, thick, fill=cyan, opacity=0.5]
(0,2) -- (3, 2) -- (5, 3) -- (2, 3) -- (0,2);
\draw [blue, line width = 0.05cm] (0,2) -- (2, 3)  ; 
\draw [blue, line width = 0.05cm] (0,3.3) -- (2, 4.3)  ;
\node[left] at (0,3){$m$};
\node at (1.2,5) {$\Bsym$};
\node at  (2.8,3.5)  {$\fZ(\Vec_{\Z_2})^{\rho}$} ;
\node at  (2.8,1.5)  {$\fZ(\Vec_{\Z_2}^*)^{\rho}$} ;
\draw[dashed] (0,0) -- (3,0);
\draw[dashed] (0,4) -- (3,4);
\draw[dashed] (2,5) -- (5,5);
\draw[dashed] (2,1) -- (5,1);
\draw [black, dashed]
(3,0) -- (3, 4) -- (5, 5) -- (5,1) -- (3,0);
 \node[above] at (1.1, 3) {$\Vec_{\Z_2}$};
 \node[above] at (2.4, 2.2) {$D_T$};
 \node[above] at (1.1, 1) {$\Vec_{\Z_2}^*$};
 \node[blue, left] at (0, 2) {$\Vec_{\Z_2}=\langle T, mT\rangle$};
\end{scope}
\end{tikzpicture}
\ee
Therefore the full symmetry category takes the form $\Vec_{\Z_2}\oplus \Vec_{T\Z_2}$. The natural transformation $\eta$ defined above leads to symmetry fractionalization of the $D_T$-surface: fusing two $D_T$-surfaces with the same boundary condition leaves behind an $m$-string on the symmetry boundary. Therefore we have fusion rule:
\be
T\otimes T=Tm\otimes Tm=m \,.
\ee
Hence the symmetry category is $\Vec_{\Z_4^T}$ with $\Z_4^T$ generated by $T$(or $Tm$).

Next we keep the bulk $\Z_2^T$-SET to be the $e_Tm$-SET, but change the symmetry boundary to be the $m$-condensed one. The linear part of the symmetry is still $\Z_2$, but now generated by the $e$-anyon on the boundary. Following the same reasoning as above, we see that the anti-linear symmetries are generated by a boundary condition of the $D_T$ surface $T$, and the fusion $eT:=e\otimes T$. Together the full symmetry category is still $\Vec_{\Z_2}\oplus \Vec_{T\Z_2}$ as a plain category. To determine the monoidal structure, note that fusing two $D_T$ surfaces with the same boundary condition still leaves an $m$-string behind, however an $m$-string on the symmetry boundary is now trivial as $m$ is condensed there. Therefore the fusion rule becomes 
\be
T\otimes T=eT\otimes eT=\lid\,.
\ee
 This means the symmetry defined by the $e_Tm$-SET and the $m$-condensed boundary is $\Z_2\times \Z_2^T$ but with possibly an anomaly. Indeed, in the $e_Tm$-SET the $e$-anyon is a Kramers doublet, hence the symmetry generator of the linear symmetry $\Z_2$ is now a Kramers doublet under $\Z_2^T$. This is precisely the characteristic of the anomaly $\alpha$. We conclude that symmetry category is now $\Vec_{\Z_2\times \Z_2^T}^\alpha$. Since this is related to $\Vec_{\Z_4^T}$ by a change of symmetry boundary, we conclude that they are Morita equivalent. This completes the SymTFT proof of the Morita equivalence we proved in~\Cref{th_Z4T_dual}.

\smallskip
\noindent{\bf Field-theoretical Perspective.} 
For this particular example a simple field-theoretical derivation of the two symmetry categories is possible, which will help confirm the mixed anomaly for the $\Z_2\times \Z_2^T$ symmetry. The $\Z_2^T$-enriched SymTFT $e_Tm$ can be described by the following Lagrangian:
 \be
\cL_{e_Tm}=\frac{1}{2} a\cup \delta b+ \frac{1}{2} a\cup w_1^2 \,,
 \ee
 where $a,b\in C^1(M,\Z)$ are the gauge fields. Then integrating out $a$ we obtain $\delta b=w_1^2$, which says $b$ and $w_1$ together form a $\Z_4^T$-connection.  On the $e$-condensed boundary, field $a$ has Dirichlet boundary condition, the field $b$ can take arbitrary value that satisfies $\delta b=w_1^2$. Hence it defines a $\Z_4^T$-connection on $\fB^\sym$ together with $w_1(T\fB^\sym)$. On the $m$-condensed boundary the field $b$ has Dirichlet boundary condition, the field $a$ generates a $\Z_2$-symmetry. However the term $\frac{1}{2} a\cup w_1^2$ says there is a mixed anomaly between the $\Z_2$ generated by $a$ and $\Z_2^T$.

\subsection{T-Enriched Quiche for $A\times \Z_2^T$ and  $A\rtimes\Z_2^T$}
In this subsubsection, we study SymTFT quiche for a family of symmetries of the form $A\rtimes \Z_2^T$ and $A\times \Z_2^T$. In~\Cref{th_S3T_dual} we proved that $\Vec_{\Z_3\rtimes \Z_2^T}$ is Morita equivalent to $\Vec_{\Z_3\times \Z_2^T}$ by analyzing the categories of symmetric defects in different gapped $S_3^T$-phases. The SymTFT discussion here provides a generalization of this result. We will show that for any abelian group $A$, $\Vec_{A\times \Z_2^T}$ is Morita equivalent to $\Vec_{A\times \Z_2^T}$. This Morita equivalence is nothing but a consequence of charge-flux duality of $(2+1)$d abelian gauge theories.

We will start by constructing the SymTFT quiche for the symmetry $\Vec_{A\times \Z_2}$. The linear part of the symmetry is $\Vec_A$, hence the underlying SymTFT is the (untwisted) $A$-gauge theory $\cZ(\Vec_A)$.
The group of simple anyons in $\cZ(\Vec_A)$ is $A\times \widehat{A}$, so we will denote a simple anyon as $(a,\chi)$ where $a\in A$ and $\chi \in \widehat{A}$. 
The braiding of this theory is given by 
\be
B_{(a_1,\chi_1),(a_2,\chi_2)}=\chi_1(a_2)\chi_2(a_1)\,.
\ee
We then define a $\Z_2^T$-enrichment on this theory as follows. 
\begin{itemize}
    \item $\rho_T$ acts on anyons as $a\mapsto a,~\chi\mapsto \chi^{-1}$.
    \item The natural transformation $\eta: \rho_T\circ \rho_T \Rightarrow \id$ is trivial.
\end{itemize}
Crucially $\rho_T$ defined above is indeed an \emph{anti-linear} automorphism of $\cZ(\Vec_A)$. An anti-linear automorphism must \emph{reverse} braiding phases. Indeed $B_{\rho_T(\chi),\rho_T(a)}=B_{\chi^{-1},a}=\chi(a)^*$. On the other hand, the identity action $a\mapsto a, \chi\mapsto \chi$ is in general \emph{not} an anti-linear braided automorphism. In general there is no such thing as  ``the trivial" $\Z_2^T$-action on a topological order.

We then verify that this choice of $\Z_2^T$-enrichment does give rise to the symmetry category $\Vec_{A\times \Z_2^T}$:
\be
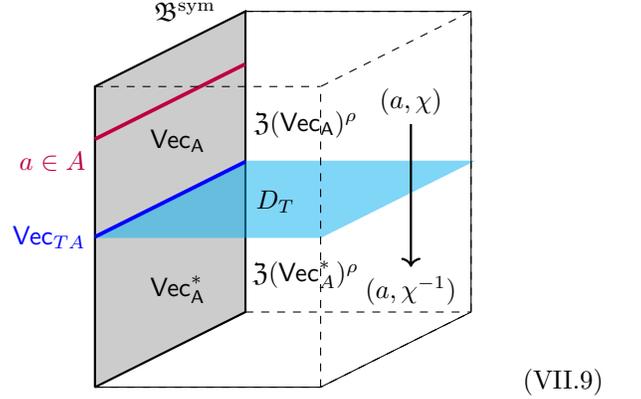
       \begin{tikzpicture}
 \begin{scope}[shift={(0,0)}] 
\draw [black]
(0,0) -- (0,4) -- (2,5) -- (5,5) -- (5,1) -- (3,0)--(0,0);
\draw [black, thick,  fill=white,opacity=1]
(0,0) -- (0, 4) -- (2, 5) -- (2,1) -- (0,0);
\draw [thick, fill=black, opacity=0.2]
(0,0) -- (0, 4) -- (2, 5) -- (2,1) -- (0,0);
\draw [cyan, thick, fill=cyan, opacity=0.5]
(0,2) -- (3, 2) -- (5, 3) -- (2, 3) -- (0,2);
\draw [blue, line width = 0.05cm] (0,2) -- (2, 3)  ; 
\draw [purple, line width = 0.05cm] (0,3.3) -- (2, 4.3)  ;
\node[left,purple] at (0,3){$a\in A$};
\node at (1.2,5) {$\Bsym$};
\node at  (2.8,3.5)  {$\fZ(\Vec_{\A})^{\rho}$} ;
\node at  (2.8,1.5)  {$\fZ(\Vec_{A}^*)^{\rho}$} ;
\draw[dashed] (0,0) -- (3,0);
\draw[dashed] (0,4) -- (3,4);
\draw[dashed] (2,5) -- (5,5);
\draw[dashed] (2,1) -- (5,1);
\draw [black, dashed]
(3,0) -- (3, 4) -- (5, 5) -- (5,1) -- (3,0);
 \node[above] at (1.1, 3) {$\Vec_{\A}$};
 \node[above] at (2.4, 2.2) {$D_T$};
 \node[above] at (1.1, 1) {$\Vec_{\A}^*$};
 \node[blue, left] at (0, 2) {$\Vec_{TA}$};
 \draw[black, line width = 0.03cm, ->] (4.2,3.5) -- (4.2,1.6);
 \node[above]  at (4.2,3.5){$(a,\chi)$};
 \node[below] at (4.2,1.6){$(a,\chi^{-1})$};
\end{scope}
\end{tikzpicture}
\ee
We take the symmetry boundary of the SymTFT $\cZ(\Vec_A)$ to be the one that condenses all charges. The linear symmetry is then generated by fluxes $a\in \A$. Since the bulk $D_T$ surface acts trivially on fluxes. We see that the anti-linear symmetry generators simply form another copy of $A$: pick any boundary condition fo $D_T$, call that $T$, then fusing with the linear symmetry defect $a\in A$ we obtain another boundary condition $aT:=a\otimes T$.  The full symmetry category is therefore $\Vec_A\oplus \Vec_{TA}$ as a plain category. Since the $D_T$-surface acts trivially on $A$, we have fusion rule:
\be 
aT\otimes bT=ab \in A \,.
\ee
Hence the full symmetry category is $\Vec_{A\times \Z_2^T}$ as claimed. We note that there is no anomaly here since there is no symmetry fractionalization for the $D_T$-surface.

Next we keep the bulk $\Z_2^T$-SET to be the one specified above, but change the symmetry boundary to be the one where all fluxes condense:
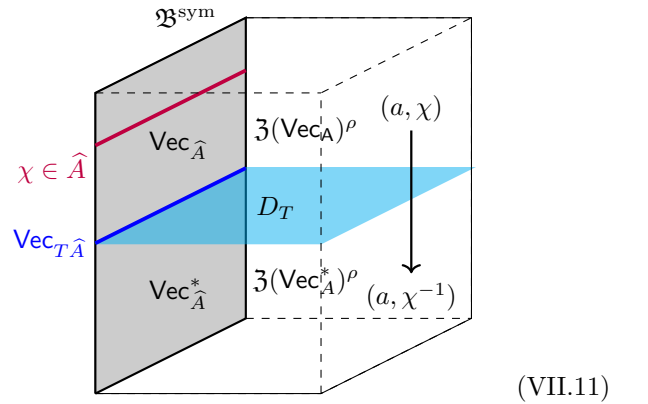
\be
       \begin{tikzpicture}
 \begin{scope}[shift={(0,0)}] 
\draw [black]
(0,0) -- (0,4) -- (2,5) -- (5,5) -- (5,1) -- (3,0)--(0,0);
\draw [black, thick,  fill=white,opacity=1]
(0,0) -- (0, 4) -- (2, 5) -- (2,1) -- (0,0);
\draw [thick, fill=black, opacity=0.2]
(0,0) -- (0, 4) -- (2, 5) -- (2,1) -- (0,0);
\draw [cyan, thick, fill=cyan, opacity=0.5]
(0,2) -- (3, 2) -- (5, 3) -- (2, 3) -- (0,2);
\draw [blue, line width = 0.05cm] (0,2) -- (2, 3)  ; 
\draw [purple, line width = 0.05cm] (0,3.3) -- (2, 4.3)  ;
\node[left,purple] at (0,3){$\chi\in \widehat{A}$};
\node at (1.2,5) {$\Bsym$};
\node at  (2.8,3.5)  {$\fZ(\Vec_{\A})^{\rho}$} ;
\node at  (2.8,1.5)  {$\fZ(\Vec_{A}^*)^{\rho}$} ;
\draw[dashed] (0,0) -- (3,0);
\draw[dashed] (0,4) -- (3,4);
\draw[dashed] (2,5) -- (5,5);
\draw[dashed] (2,1) -- (5,1);
\draw [black, dashed]
(3,0) -- (3, 4) -- (5, 5) -- (5,1) -- (3,0);
 \node[above] at (1.1, 3) {$\Vec_{\widehat{A}}$};
 \node[above] at (2.4, 2.2) {$D_T$};
 \node[above] at (1.1, 1) {$\Vec_{\widehat{A}}^*$};
 \node[blue, left] at (0, 2) {$\Vec_{T\widehat{A}}$};
 \draw[black, line width = 0.03cm, ->] (4.2,3.5) -- (4.2,1.6);
 \node[above]  at (4.2,3.5){$(a,\chi)$};
 \node[below] at (4.2,1.6){$(a,\chi^{-1})$};
\end{scope}
\end{tikzpicture}
\ee
The linear part of the symmetry is now generated by the confined charges $\chi\in \widehat{A}$ and form the fusion category $\Vec_{\widehat{A}}$. We next determine the anti-linear symmetry generators and their fusion rules. Start with any boundary condition $T$ of the bulk $D_T$ surface, other boundary conditions are obtained by fusing with linear symmetry defects $\chi \in \widehat{A}$. Notice that passing a linear symmetry defect $\chi$ across the bulk $D_T$ surface changes it to $\chi^{-1}$, which is still nontrivial on the symmetry boundary. Therefore all the fusion $\chi T:=\chi\otimes T$ are distinct defects. Due to the nontrivial action of $D_T$ on charges, we need to fix a convention for the definition of $\chi T$: we define it to be $T$ fused with a $\chi$ defect \emph{above} $T$. If instead we fuse with a $\chi$-defect from below $T$, we would get $T\otimes \chi\simeq \chi^{-1} T$.  The nontrivial action of $D_T$ on the charges leads to a different fusion rule: 
\be
\chi T\otimes \chi' T= \chi\chi'^{-1}  \,.
\ee
Hence the full symmetry category is now $\Vec_{\widehat{A}\rtimes \Z_2^T}$, with $\Z_2^T$ acting on $\widehat{A}$ as inversion. Since this is related to $\Vec_{A\times \Z_2^T}$ by a change of symmetry boundary, we have proven the Morita equivalence 
\be
\Vec_{A\times \Z_2^T}\simeq_{\text{Morita}} \Vec_{A\rtimes \Z_2^T} \,,
\ee
where we also used the (non-canonical) isomorphism $A\simeq \widehat{A}$.

\subsection{T-Enriched Quiche for Galois-real TY-categories}
Here we will discuss the $\Z_2^T$-enriched SymTFT quiche for all the Galois-real Tambara-Yamagami (TY) categories. We will show that the SymTFT perspective provides s simple derivation of the classification of Galois-real TY categories. 

Let us recall the classification. An ordinary complex TY category is determined by an abelian group $A$, a symmetric non-degenerate bicharacter $q: A\times A\to U(1)$, and a Frobenius–Schur indicator $\nu=\pm$. The corresponding complex TY category is denoted as $\TY(A,q,\nu)$. A Galois-real TY category is a Galois-real fusion category of the form $\Vec_A\oplus \Vec$ with the generator $T\in \Vec$ being anti-linear. It satisfies the TY fusion rule 
\be\label{eq_TY_fusion}
T\otimes T=\bigoplus_{a\in A}a \,.
\ee
According to~\cite{plavnik2024TYreal}, a Galois-real TY category is determined by an abelian group $A$ together with a skew-symmetric non-degenerate bicharacter $q: A\times A\to U(1)$. The corresponding Galois-real TY category  is denoted as $\TY_{\overline{\C}}(A,q)$. Here skew-symmetric means $q(a,b)=q(b,a)^*$. Compared with usual complex TY categories, the main difference is the change from symmetric bicharacter to skew-symmetric, and the missing of the  Frobenius–Schur indicator. As we shall see, the enriched SymTFT provides an intuitive understanding for these differences.

To construct the $\Z_2^T$-enriched SymTFT for a Galois real TY category $\TY_{\overline{\C}}(A,q)$, we start with the underlying SymTFT $\fZ(\Vec_A)$ for the linear part of the symmetry $\Vec_A$. The symmetry boundary of the underlying SymTFT $\fZ(\Vec_A)$ is the one that condenses all charges. The linear symmetry is then generated by the fluxes $a\in A$. Let $D_T$ be the bulk time-reversal domain wall between $\fZ(\Vec_A)$ and $\fZ(\Vec_A^*)$. In order for the anti-linear component of the symmetry $\cc_T$ to have only one simple object, we argue that the $D_T$ surface must act as $e-m$ duality or  ``charge-flux" exchange, i.e. it maps $A$ to $\widehat{A}$ and $\widehat{A}$ to $A$.

We start with a given boundary condition $T$ of the $D_T$ surface, other boundary conditions are obtained by fusing with linear symmetry defects $a\in A$. In order for the fusion $a\otimes T$ to be equivalent to $T$, $a$ must be locally trivial near the $T$-defect. Since $a$ itself is not trivial on the symmetry boundary, the only way it can be locally trivial is that it becomes a condensed anyon when transported across the bulk $D_T$ surface. This means $D_T$ must map all fluxes to charges. Since $D_T$ has order 2, this automatically implies it maps all charges back to fluxes.

Now assume the $\Z_2^T$-action on simple anyons of $\cZ(\Vec_A)$ is given by
\be 
\rho_T:\quad \left\{ 
\ba 
     A\to \widehat{A}, \quad  & a\mapsto q(a) \cr 
    \widehat{A}\to A,  \quad  & \chi\mapsto \widehat{q}(\chi) \,.
\ea\right.
\ee
The quiche take the following form
\be
       \begin{tikzpicture}
 \begin{scope}[shift={(0,0)}] 
\draw [black]
(0,0) -- (0,4) -- (2,5) -- (5,5) -- (5,1) -- (3,0)--(0,0);
\draw [black, thick,  fill=white,opacity=1]
(0,0) -- (0, 4) -- (2, 5) -- (2,1) -- (0,0);
\draw [thick, fill=black, opacity=0.2]
(0,0) -- (0, 4) -- (2, 5) -- (2,1) -- (0,0);
\draw [cyan, thick, fill=cyan, opacity=0.5]
(0,2) -- (3, 2) -- (5, 3) -- (2, 3) -- (0,2);
\draw [blue, line width = 0.05cm] (0,2) -- (2, 3)  ; 
\draw [purple, line width = 0.05cm] (0,3.3) -- (2, 4.3)  ;
\node[left,purple] at (0,3){$a\in A$};
\node at (1.2,5) {$\Bsym$};
\node at  (2.8,3.5)  {$\fZ(\Vec_{\A})^{\rho}$} ;
\node at  (2.8,1.5)  {$\fZ(\Vec_{A}^*)^{\rho}$} ;
\draw[dashed] (0,0) -- (3,0);
\draw[dashed] (0,4) -- (3,4);
\draw[dashed] (2,5) -- (5,5);
\draw[dashed] (2,1) -- (5,1);
\draw [black, dashed]
(3,0) -- (3, 4) -- (5, 5) -- (5,1) -- (3,0);
 \node[above] at (1.1, 3) {$\Vec_{A}$};
 \node[above] at (2.4, 2.2) {$D_T$};
 \node[above] at (1.1, 1) {$\Vec_A^*$};
 \node[blue, left] at (0, 2) {$\Vec_T$};
 \draw[black, line width = 0.03cm, ->] (4.2,3.5) -- (4.2,1.6);
 \node[above]  at (4.2,3.5){$(a,\chi)$};
 \node[below] at (4.2,1.6){$(\widehat{q}(\chi),q(a))$};
\end{scope}
\end{tikzpicture}
\ee

Then $q: A\to \widehat{A}$ may be regarded as a non-degenerate bicharacter $A\times A\to U(1)$: $q(a,b):=q(a)(b)$. For $\rho_T$ to be an \emph{anti-linear} braided automorphism, we must have 
\be
B_{\rho_T(a),\rho_T(\chi)}=B_{a,\chi}^* \,.
\ee
The LHS is $B_{q(a), \widehat{q}(\chi)}=q(a)(\widehat{q}(\chi))$ while the RHS is $\chi(a)^*$. Therefore we have relation $q(a)(\widehat{q}(\chi))=\chi(a)^*$. Now take $\chi=q(b)$ for some $b\in A$, we obtain
\be
q(a)(\widehat{q}(q(b)))=q(b)(a)^* \,.
\ee
However, the action $\rho_T$ must have order 2, hence $\widehat{q}(q(b))=b$. Therefore the LHS reduces to $q(a)(b)$. We arrive at relation 
\be
q(a)(b)=q(b)(a)^* \,,
\ee
which is nothing but the condition that $q$ viewed as a bicharacter of $A$ is skew-symmetric. We conclude that a $\Z_2^T$-enrichment on $\cZ(\Vec_A)$ that gives rise to a Galois-real TY symmetry  category on the $\Vec_A$-boundary is determined by a non-degenerate skew-symmetric bicharacter $q$ of $A$.  Note that for the usual complex TY categories, the Frobenius-Schur indicator can be viewed as the freedom of stacking with a $(2+1)$d $\Z_2$-SPT in the bulk before gauging a charge-flux exchange $\Z_2$-symmetry. However,  here $H^3(\Z_2^T,U(1)_T)=0$ hence there is no more freedom. This explains the absence of a Frobenius-Schur indicator  in the classification of Galois-real TY categories.

\section{Conclusions and Outlook}
\label{sec:Conclusions}

We have put forward a theory for categorical symmetries including time-reversal action $\Z_2^T$, and made a case that the natural mathematical habitat for such symmetries is real fusion categories. In fact, symmetry categories with time-reversal have to be so-called Galois-real fusion categories, as opposed to $\R$-real ones (where the endomorphisms of the identity form a real vector space). Generalized gauging (or Morita equivalence between categories) in the world of real fusion categories does not always follow the intuition that we may have formed working with complex fusion categories, and we provided some key examples to illustrate this. E.g. the symmetries $A\times \Z_2^T $ and $A\rtimes \Z_2^T$ are gauge related and even the representation category of an abelian group $\Z_n^T$ can be non-invertible. 
We furthermore set up a  $T$-enriched SymTFT framework for such real fusion categories, which in particular gives an elegant way to derive many of the Morita equivalence relations. 

Although the general theory we put forward is applicable in any dimensions, our examples focus on $(1+1)$d real fusion categories, and finite many simple objects. It would be interesting to extend our analysis to a full study of the phases using the SymTFT and to symmetries in higher dimensions, including continuous symmetries, such as 
$\Z_4^T, U(1)\rtimes \Z_2^T, U(1)\times \Z_2^T$ that play a role in topological insulator/superconductors. 
One challenge for developing the SymTFT sandwich construction for gapped phases is the $\Z_2^T$-symmetry breaking on the symmetry boundary. This could be avoided by considering instead $\Z_2^T$-symmetric (non-SSB) boundaries as in \cite{Pace:2025hpb}, which may lead to a complete sandwich construction for gapped phases with symmetries of the form  $\cc\rtimes \Z_2^T$, up to pure $\Z_2^T$-SPTs. Alternatively, one could interpret the compactification of the sandwich with a $\Z_2^T$-SSB boundary  not as the gapped phase itself but a doubled version where the two vacua are interpreted as the system itself and its complex conjugate. 

In higher dimensions, we believe the generalization of our categorical formulation of symmetries and gapped phases is conceptually straightforward. The main challenge comes from the mathematical side, where the theory of fusion higher categories over  non-algebraically-closed fields are not well-developed -- see~\cite{D_coppet_2025} for a discussion of fusion 2-categories over non-algebraically-closed fields. We do mention here that the $(2+1)$d time-reversal anomaly $\omega \in H^4(\Z_2^T, U(1)_T)$ can be described by the Galois-real fusion 2-category $2\Vec_{\Z_2^T}^\omega$~\cite{D_coppet_2025}, and generalizing it to $2\Vec_{G^T}^\omega$ seems straightforward. This suggests that Galois-real fusion higher categories are still the correct language for describing generalized anti-linear symmetries in higher dimensions.

\begin{acknowledgments}
We thank  Andrea Antinucci, Lea Bottini, Lakshya Bhardwaj, Christian Copetti, Nick Jones, Ho Tat Lam, Sal Pace, Shinsei Ryu, Sean Sanford, Shu-Heng Shao, Apoorv Tiwari, Weicheng Ye, Zhi-Hao Zhang for discussions. We thank Lea Bottini and Nick Jones for coordinating submission of their paper \cite{Bottini:2026zga} with us. SSN thanks Lakshya Bhardwaj, Lea Bottini, Nick Jones, and Apoorv Tiwari for discussions at a preliminary stage of these developments.
Some of the computations in the paper were obtained with the help of ChatGPT-5.4 Thinking.
This work is  supported by the UKRI Frontier Research Grant, underwriting the ERC Advanced Grant ``Generalized Symmetries in Quantum Field Theory and Quantum Gravity”. 
\end{acknowledgments}

\appendix

\section{More on the Semi-direct Product Categories $\cc\rtimes \Z_2^T$}
\label{app_semi_direct}

Here we will establish the Morita equivalence between an $\R$-real fusion category $\cD_\R$ and the corresponding semi-direct product Galois-real fusion category $\cc\rtimes \Z_2^T$, where $\cc:=\cD_\R\boxtimes_\R\Vec_\C$ is the complexification of $\cD$. The physical meaning of this Morita equivalence is discussed in~\Cref{sec_semi_product}: if $\cc$ is a (complex) fusion category of topological defects in a 2D theory, and the 2D theory is invariant under a $\Z_2^T$-action, then the $\Z_2^T$-invariant topological defects form a real form $\cD_\R$ of $\cc$, and spontaneously breaking $\Z_2^T$ results in a theory with topological defects forming $\cc\rtimes \Z_2^T$. We formulate this as the following theorem. 

\begin{theorem}\label{thm_semi_direct}
    Let $\cD_\R$ be an $\R$-real fusion category, and $\cc:=\cD_\R\boxtimes_\R\Vec_\C$ its complexification. The canonical $\Z_2^T$-action on $\Vec_\C$ gives $\cc$ a natural monoidal $\Z_2^T$-action. Form the semi-direct product $\cc\rtimes \Z_2^T$, then we have Morita equivalence of real fusion categories:
    \be
    \cD_\R\simeq_{\text{Morita}} \cc\rtimes \Z_2^T\,.
    \ee
    \proof Define a left $\cc\rtimes \Z_2^T$-module category structure on $\cc$ as follows, for simple objects $(x,s)\in \cc\rtimes \Z_2^T,~y\in \cc$, define
    \be
     (x,s)\rhd y:= x\otimes s(y)\,.
    \ee
    Let $F\in \End_{\cc\rtimes \Z_2^T}(\cc)$. Since $\cc\subset \cc\rtimes \Z_2^T$ is a fusion sub-category, the functor $F$ takes the form $F(-)=-\otimes z$ for some $z\in \cc$. The module structure on $F$ provides natural isomorphism 
    \be
     F((x,s)\rhd y)=(x\otimes s(y))\otimes z \simeq (x,s)\rhd F(y)=x\otimes s(y\otimes z)\,.
    \ee
    By naturality it suffices to set $x=y=\lid$, then $F$ gives isomorphism $z\simeq s(z)$. Therefore the module functor $F$ is the same as an object $z\in \cc$ together with an equivariant structure with respect to the $\Z_2^T$-action on $\cc$. Hence $\End_{\cc\rtimes \Z_2^T}(\cc)\simeq \cc^{h\Z_2^T}=\cD_\R$.\qed 
\end{theorem}

\section{Proof of Gauging Finite Subgroup}\label{app_gauge_subgroup}

Let $N\lhd G_0<G^T$ be a central subgroup and $K^T:=G^T/N$. Let $\C[N]\in \Vec_{G^T}$ be the  algebra supported on $N$ with the obvious (untwisted) algebra structure. We claim 
\be
\Bimod_{\Vec_{G^T}}(\C[N])\simeq \Vec_{\widehat{N}\rtimes K^T}^\omega
\ee
is the desired Morita dual category.
 To start, we notice that since $N$ is abelian any indecomposable bimodule over $\C[N]$ must have underlying object of the form $\C[kN]$ for some $k\in G^T/N=:K^T$. Fix a section $\sigma: K^T\to G^T$, and define
\[
e_2(k_1,k_2):=\sigma(k_1)\sigma(k_2)\sigma(k_1k_2)^{-1}
\]
to be the extension class of the extension $N\to G^T\to K^T$. Choose a basis $\{[k,n], n\in N\}$ for $\C[k,N]$. Similarly choose a basis $\{[n],n\in N\}$ for $\C[N]$. The left module structure is then defined on the basis by
\begin{align}
    [m]\rhd [k,n]=l(m,n)[k,mn],\quad l(m,n)\in \C^\times.
\end{align}
The left module structure requires that the scalars $l(m,n)$ satisfy
\begin{align}
    l(m_1,m_2n)l(m_2,n)=l(m_1m_2,n).
\end{align}
Setting $n=1$, we obtain
\begin{align}
    l(m_1,m_2)=\frac{l(m_1m_2,1)}{l(m_2,1)}.
\end{align}
Denote the corresponding left module by $\C[kN]_l$. There is an equivalence of left modules
\[
\C[kN]_l\simeq \C[kN]_{l'}
\]
if and only if there is a morphism $\eta: \C[kN]\to \C[kN]$, $\eta([k,n])=\eta(n)[k,n]$, such that
\begin{align}
    l'(m,n)\eta(n)=\eta(mn)l(m,n).
\end{align}
Setting $\eta(n)=\frac{1}{l(n,1)}$ and $l'=1$, we see that $l$ is equivalent to $l'=1$. Hence there is only one left module structure on $\C[kN]$ up to equivalence. We can then fix $l=1$. A right module structure on $\C[k,N]$ is similarly defined by scalars $r(n,m)\in \C^\times$. Again $r$ is completely determined by its values on $r(1,m)$:
\begin{align}\label{eq_r}
    r(m_1,m_2)=\frac{r(1,m_1m_2)}{r(1,m_2)}.
\end{align}
For $\C[kN]_r$ to be a bimodule we require
\begin{align}
    ([m_1]\rhd [s,n])\lhd [m_2]= [m_1]\rhd ([s,n]\lhd [m_2]),
\end{align}
which translates to the equation
\begin{align}
    r(m_1n,m_2)=r(n,m_2).
\end{align}
Setting $n=1$, we obtain $r(m_1,m_2)=r(1,m_2)$. Combining with~\Cref{eq_r}, we obtain
\begin{align}
    r(1,m_1m_2)=r(1,m_1)r(1,m_2).
\end{align}
Hence $r(1,-)$ is a character of $N$. Conversely, it can be checked that for every character $\gamma\in \wh{N}$, $r(m_1,m_2)=\gamma(m_2)$ defines a bimodule structure. Denote the corresponding bimodule by $\C[kN]^\gamma$. Two bimodules $\C[kN]^\gamma$ and $\C[kN]^{\gamma'}$ are equivalent if and only if the characters $\gamma,\gamma'$ are the same. Now we have found a complete list of all simple objects in $\Bimod_{\Vec_{G^T}}(\C[N])$: as a set they form $K^T\times\wh{N}$. We next compute the monoidal structure.

Since the left module structure on each $\C[kN]^\gamma$ is trivial, we have
\begin{align}
    \C[k_1N]^{\gamma_1}\ot_{\C[N]}\C[k_2N]^{\gamma_2}\simeq \C[k_1k_2N]
\end{align}
at the level of objects. There is a universal balanced morphism
\begin{align}
   &\beta: \C[k_1N]^{\gamma_1}\ot \C[k_2N]^{\gamma_2}\to\nonumber \\
   &\C[k_1N]^{\gamma_1}\ot_{\C[N]}\C[k_2N]^{\gamma_2}\simeq \C[k_1k_2N]
\end{align}
which on the basis takes the form
\begin{align}
& [k_1,n_1]\diamond [k_2,n_2]:=\beta([k_1,n_1]\ot[k_2,n_2])\nonumber\\
 &=\beta(n_1,n_2)[k_1k_2,n_1n_2e_2(k_1,k_2)],\quad \beta(n_1,n_2)\in \C^\times.
\end{align}
The fact that $\beta$ is balanced means that we have the identity
\begin{align}
   ([k_1,n_1]\lhd [m])\diamond [k_2,n_2]=[k_1,n_1]\diamond ([m]\rhd [k_2,n_2]),
\end{align}
which translates to the equation
\begin{align}
    \gamma_1(m)\beta(n_1m,n_2)=\beta(n_1,mn_2).
\end{align}
Setting $n_1=1$, we obtain
\begin{align}\label{eq_gamma_beta}
    \gamma_1(m)\beta(m,n_2)=\beta(1,mn_2).
\end{align}
We now compute the left and right module structures on
\[
\C[k_1N]^{\gamma_1}\ot_{\C[N]}\C[k_2N]^{\gamma_2}\simeq \C[k_1k_2N].
\]
Assume
\[
\C[k_1N]^{\gamma_1}\ot_{\C[N]}\C[k_2N]^{\gamma_2}\simeq \C[k_1k_2N]^{\gamma_3}
\]
as bimodules for some $\gamma_3\in \wh{N}$. Recall that the left module structure is defined by the diagram
\begin{widetext}
\begin{equation}
    \begin{tikzcd}
        \C[N]\ot \C[k_1N]^{\gamma_1}\ot \C[k_2N]^{\gamma_2}\arrow[r,"\rhd"]\arrow[d,"\beta"]& \C[k_1N]^{\gamma_1}\ot \C[k_2N]^{\gamma_2}\arrow[d,"\beta"]\\
          \C[N]\ot \C[k_1N]^{\gamma_1}\ot_{\C[N]} \C[k_2N]^{\gamma_2}\arrow[r,dashed,"\rhd"]& \C[k_1N]^{\gamma_1}\ot_{\C[N]} \C[k_2N]^{\gamma_2}
    \end{tikzcd}\label{eq_left_module}
\end{equation}
\end{widetext}
Pick
\[
[n]\ot [k_1,m_1]\ot[k_2,m_2]\in \C[N]\ot \C[k_1N]^{\gamma_1}\ot \C[k_2N]^{\gamma_2}.
\]
The diagram expresses
\begin{align}
    [n]\rhd([k_1,m_1]\diamond[k_2,m_2])=([n]\rhd [k_1,m_1])\diamond [k_2,m_2],
\end{align}
which leads to
\begin{align}
    &\beta(m_1,m_2)[n]\rhd [k_1k_2,m_1m_2e_2(k_1,k_2)]\nonumber\\&
    =\beta(nm_1,m_2)[k_1k_2,nm_1m_2e_2(k_1,k_2)].
\end{align}
Since the left module structure on $\C[k_1k_2N]^{\gamma_3}$ is trivial, this gives
\begin{align}
    \beta(m_1,m_2)=\beta(nm_1,m_2).
\end{align}
Setting $m_1=1$, we obtain
\begin{align}
    \beta(1,m_2)=\beta(n,m_2).
\end{align}
Combining with~\Cref{eq_gamma_beta}, we get
\begin{align}
    \gamma_1(m)\beta(1,1)=\beta(n,m).\label{eq_gamma_beta_2}
\end{align}
Now the right module is defined similarly to~\Cref{eq_left_module}. Start with
\[
[k_1,m_1]\ot[k_2,m_2]\ot [n]\in \C[k_1N]^{\gamma_1}\ot \C[k_2N]^{\gamma_2}\ot \C[N].
\]
The right module structure is defined by the equation
\begin{align}
    ([k_1,m_1]\diamond [k_2,m_2])\lhd [n]=[k_1,m_1]\diamond([k_2,m_2]\lhd [n]),
\end{align}
which translates to
\begin{align}
\beta(m_1,m_2)\gamma_3(n)=\beta(m_1,m_2n)\gamma_2(n)^{\textcolor{red}{s(k_1)}}.
\end{align}
Here $s(k_1)=1/T$. Complex conjugation appears if $s(k_1)=T$, since  $[k_1,m_1]$ is really $\id_{\delta_{k_1,m_1}}$ an endomorphism of the object $\delta_{k_1,m_1}\in \Vec_{G^T}$. Now using~\Cref{eq_gamma_beta_2}, we obtain
\begin{align}
  \gamma_1(m_2)\gamma_3(n)=\gamma_1(m_2n)\gamma_2(n)^{s(k_1)} \Rightarrow \gamma_3(n)=\gamma_1(n)\gamma_2(n)^{s(k_1)}.
\end{align}
We conclude that the fusion rule in $\Bimod_{\Vec_{G^T}}(\C[N])$ is
\begin{align}
    \C[k_1N]^{\gamma_1}\ot_{\C[N]} \C[k_2N]^{\gamma_2}= \C[k_1k_2N]^{\gamma_1\gamma_2^{s(k_1)}}.
\end{align}

Finally, we compute the associator. Recall that the associator is defined by the following diagram:
\begin{widetext}

\begin{equation}
    \begin{tikzcd}
       (\C[k_1N]^{\gamma_1}\ot \C[k_2N]^{\gamma_2})\ot \C[k_3N]^{\gamma_3}\arrow[r,"="]\arrow[d,"\beta"]&  \C[k_1N]^{\gamma_1}\ot (\C[k_2N]^{\gamma_2}\ot \C[k_3N]^{\gamma_3})\arrow[d,"\beta"]\\
        (\C[k_1N]^{\gamma_1}\ot_{\C[N]} \C[k_2N]^{\gamma_2})\ot \C[k_3N]^{\gamma_3}\arrow[d,"\beta"]& \C[k_1N]^{\gamma_1}\ot (\C[k_2N]^{\gamma_2}\ot_{\C[N]} \C[k_3N]^{\gamma_3})\arrow[d,"\beta"]\\
          (\C[k_1N]^{\gamma_1}\ot_{\C[N]} \C[k_2N]^{\gamma_2})\ot_{\C[N]} \C[k_3N]^{\gamma_3}\arrow[r,dashed,"\alpha"]& \C[k_1N]^{\gamma_1}\ot_{\C[N]} (\C[k_2N]^{\gamma_2}\ot_{\C[N]} \C[k_3N]^{\gamma_3})
    \end{tikzcd}
\end{equation}

Pick
\[
[k_1,n_1]\ot[k_2,n_2]\ot[k_3,n_3]\in (\C[k_1N]^{\gamma_1}\ot \C[k_2N]^{\gamma_2})\ot \C[k_3N]^{\gamma_3}.
\]
The diagram gives
\begin{align}
    \alpha((k_1,\gamma_1),(k_2,\gamma_2),(k_3,\gamma_3))([k_1,n_1]\diamond[k_2,n_2])\diamond [k_3,n_3]=[k_1,n_1]\diamond([k_2,n_2]\diamond[k_3,n_3]),
\end{align}
which leads to
\begin{align}
   & \alpha((k_1,\gamma_1),(k_2,\gamma_2),(k_3,\gamma_3))\gamma_1(n_2)\gamma_1(n_3)\gamma_2(n_3)^{s(k_1)} [k_1k_2k_3,n_1n_2n_3e_2(k_1,k_2)e_2(k_1k_2,k_3)]\nn\\
   &=\gamma_1(n_2n_3e_2(k_2,k_3))\gamma_2(n_3)^{s(k_1)}[k_1k_2k_3,n_1n_2n_3e_2(k_1,k_2k_3)e_2(k_2,k_3)].
\end{align}
Comparing coefficients, we obtain
\begin{align}
     \alpha((k_1,\gamma_1),(k_2,\gamma_2),(k_3,\gamma_3))=\gamma_1(e_2(k_2,k_3)).\label{eq_alpha}
\end{align}
\end{widetext}
It is not difficult to check that all the simples $\C[kN]^{\gamma}\in \Bimod_{\Vec_{G^T}}(\C[N])$ have $\End(-)\simeq \C$, hence are complex objects. Furthermore, $\C[kN]^{\gamma}$ is Galois trivial (resp.\ nontrivial) iff $s(k)=1$ (resp.\ $T$). \qed
\twocolumngrid
\bibliographystyle{ytphys}
\small 
\baselineskip=.7\baselineskip
\let\bbb\bibitem\def\bibitem{\itemsep3.3pt\bbb}
\bibliography{ref}

\end{document}